\documentclass[11pt]{article}
\usepackage[T1]{fontenc}
\usepackage{amsfonts,amsmath,amsthm,amssymb,dsfont,mathtools}
\usepackage{fullpage}
\usepackage{enumitem}
\usepackage[linesnumbered,boxed,ruled,vlined]{algorithm2e}
\usepackage[colorlinks,citecolor=blue,linkcolor=blue,urlcolor=red,pagebackref]{hyperref}
\usepackage[capitalise]{cleveref}
\usepackage{url}
\usepackage{thmtools, thm-restate}
\usepackage{xcolor}
\usepackage{graphicx}
\usepackage{mathdots}
\usepackage[normalem]{ulem}
\usepackage{xspace}
\xspaceaddexceptions{]\}}
\usepackage{comment}
\usepackage{tabu}
\usepackage{framed}
\usepackage{float,wrapfig}
\usepackage{subfigure}
\usepackage{tikz}
\usepackage{xifthen}
\usepackage[mathlines]{lineno}
\usepackage{etoolbox}
\usepackage[bottom]{footmisc}
\newcommand*\linenomathpatch[1]{%
  \cspreto{#1}{\linenomath}%
  \cspreto{#1*}{\linenomath}%
  \csappto{end#1}{\endlinenomath}%
  \csappto{end#1*}{\endlinenomath}%
}

\usepackage[shortcuts]{extdash}

\newenvironment{proofof}[2][]{\begin{proof}[#2]}{\end{proof}}

\linenomathpatch{equation}
\linenomathpatch{gather}
\linenomathpatch{multline}
\linenomathpatch{align}
\linenomathpatch{alignat}
\linenomathpatch{flalign}

\theoremstyle{plain}
\newtheorem{theorem}{Theorem}[section]
\newtheorem{lemma}[theorem]{Lemma}

\newtheorem{claim}[theorem]{Claim}

\theoremstyle{definition}
\newtheorem{definition}[theorem]{Definition}
\newtheorem{remark}[theorem]{Remark}

\setlist[enumerate]{nosep, topsep=1ex}
\setlist[itemize]{nosep, topsep=1ex}
\setlist[description]{nosep}
\allowdisplaybreaks
\def\ShowAuthNotes{1}
\ifnum\ShowAuthNotes=1
\newcommand{\authnote}[2]{\ \\ \textcolor{red}{\parbox{0.9\linewidth}{[{\footnotesize {\bf #1:} { {#2}}}]}}\newline}
\else
\newcommand{\authnote}[2]{}
\fi

\newcommand{\eps}{\varepsilon}
\renewcommand{\Pr}{\operatorname*{\mathbf{Pr}}}
\newcommand{\Ex}{\operatorname*{\mathbf{E}}}

\newcommand{\poly}{\operatorname{\mathrm{poly}}}
\newcommand{\polylog}{\poly\log}
\newcommand{\F}{\mathbb{F}}
\newcommand{\R}{\mathbb{R}}
\newcommand{\N}{\mathbb{N}}

\newcommand{\Z}{\mathbb{Z}}

\let\shorttilde\tilde

\newcommand{\caI}{\mathcal{I}}
\newcommand{\caJ}{\mathcal{J}}

\newcommand{\caL}{\mathcal{L}}

\newcommand{\caP}{\mathcal{P}}

\newcommand{\caS}{\mathcal{S}}
\newcommand{\caT}{\mathcal{T}}

\newcommand{\caX}{\mathcal{X}}

\newcommand{\defeq}{:=}
\newcommand{\Root}{\text{root}}

\newcommand{\overbar}[1]{\mkern 3.mu\overline{\mkern-3.7mu#1\mkern-0.7mu}\mkern 0.7mu}	
\newcommand{\myunderbar}[1]{\mkern 0.5mu\underline{\mkern-0.5mu#1\mkern-3mu}\mkern 3mu}	
	
\newboolean{doubleblind}
\setboolean{doubleblind}{false}

\ifdoubleblind
\author{
Author(s)
}
\else
\author{
Weiming Feng\thanks{ETH Z\"urich. \texttt{weiming.feng@eth-its.ethz.ch}. Weiming Feng acknowledges support of Dr.\ Max R\"ossler, the Walter Haefner Foundation and the ETH Z\"urich Foundation.}  \and Ce Jin\thanks{MIT. \texttt{cejin@mit.edu}. Ce Jin is supported by NSF grants CCF-2129139 and CCF-2127597, and a Simons Investigator Award.}
}
\fi
\ifdoubleblind
\title{Approximately Counting Knapsack Solutions\\  in Subquadratic Time}
\else
\title{Approximately Counting Knapsack Solutions\\ in Subquadratic Time\footnote{This work was done in part while two authors were visiting the Simons Institute for the Theory of Computing.
}}
\fi

\date{}

\begin{document}
\maketitle
\begin{abstract}

We revisit the classic \#Knapsack problem, which asks to count the Boolean points $(x_1,x_2,\dots,$ $x_n)\in \{0,1\}^n$ in a given half-space $\sum_{i=1}^n W_ix_i \le T$.
This \#P-complete problem is known to admit $(1 \pm \varepsilon)$-approximation. 
Before this work, [Dyer, STOC 2003]'s $\widetilde{O}(n^{2.5} + n^2{\varepsilon^{-2}})$-time randomized approximation scheme remains the fastest known in the natural regime of $\varepsilon \ge 1/\operatorname{\mathrm{poly}}\log  n$.

In this paper, we give a randomized $(1 \pm \varepsilon)$-approximation algorithm for \#Knapsack in $\widetilde{O}(n^{1.5}{\varepsilon^{-2}})$ time (in the standard word-RAM model), achieving the first \emph{sub-quadratic} dependence on $n$. 
Such sub-quadratic running time is rare in the approximate counting literature in general, as a large class of algorithms naturally faces a quadratic-time barrier.

Our algorithm follows Dyer's framework, which reduces \#Knapsack to the task of sampling (and approximately counting) solutions in a randomly rounded instance with $\poly(n)$-bounded integer weights.  
We refine Dyer's framework using the following ideas:
\begin{itemize}
  \item We decrease the sample complexity of Dyer's Monte Carlo method, by proving some structural lemmas for ``typical'' points near the input hyperplane via hitting-set arguments, and appropriately setting the rounding scale.
  \item Instead of running a vanilla dynamic program on the rounded instance, we employ  techniques from the growing field of  pseudopolynomial-time Subset Sum algorithms, such as FFT, divide-and-conquer, and balls-into-bins hashing of [Bringmann, SODA 2017].
\end{itemize}
To implement these ideas, we also need other technical ingredients, including a surprising application of the recent Bounded Monotone $(\max,+)$-Convolution algorithm by [Chi, Duan, Xie, and Zhang, STOC 2022] (adapted by [Bringmann, D\"{u}rr, and Polak, ESA 2024]), the notion of sum-approximation from [Gawrychowski, Markin, and Weimann, ICALP 2018]'s \#Knapsack approximation scheme, and a two-phase extension of Dyer's framework for handling tiny weights.
\end{abstract}

\hfill

\pagebreak{}

\section{Introduction}
Half-spaces (also known as linear threshold functions) are fundamental mathematical objects which have played important roles across various areas of theoretical computer science, such as  learning theory \cite{KalaiKMS08,ODonnellS11,DeDFS14}, property testing \cite{MatulefORS10}, pseudorandomness \cite{DiakonikolasGJSV10,RabaniS10}, and circuit complexity \cite{PaturiS90,HastadG91,Sherstov09,KaneW16}.

In this paper, we study the task of computing the number of points on the $n$-dimensional Boolean hypercube that lie in a given half-space, \[\Big \lvert \Big \{x\in \{0,1\}^n : \sum_{i=1}^n W_ix_i\le T\Big \}\Big \rvert,\]
where we assume $W_i\ge 0$ without loss of generality (by possibly substituting $x_i' := 1-x_i$ to flip the sign for negative weights).
This counting problem is also known as (Zero-One) \emph{\#Knapsack}, as it can be combinatorially interpreted as counting the feasible solutions to a (Zero-One) Knapsack instance, in which we need to pack a subset of the $n$ input items with weights $W_1,\dots,W_n$ into a knapsack of capacity $T$. We adopt this combinatorial viewpoint and use the knapsack terminology for the rest of the paper.

Knapsack-type problems are fundamental in computer science and combinatorial optimization \cite{0010031}, and naturally \#Knapsack has also been widely studied as an important counting problem. 
It is well-known that \#Knapsack is \#P-complete and unlikely to have polynomial-time exact algorithms.
Hence, a long line of research has devoted to designing efficient approximation algorithms for \#Knapsack.
An algorithm is said to be an FPRAS (fully polynomial-time randomized approximation scheme) for \#Knapsack if given a knapsack instance $(W_1,W_2,\dots,W_n;T)$ and an error bound $\eps > 0$, it outputs a random number $\tilde{Z}$ in time $\mathrm{poly}(n,1/\eps)$ such that 
\begin{align*}
   \Pr\left[(1-\eps)|\Omega| \leq  \tilde{Z} \leq (1+\eps)|\Omega|\right] \geq \frac{2}{3},
\end{align*}
where $\Omega := \{I\subseteq [n]: \sum_{i\in I}W_i \le T\} \subseteq 2^{[n]}$ denotes the set of all knapsack solutions.

A seminal work of Jerrum, Valiant and Vazirani~\cite{JVV86} proved that for self-reducible problems, sampling and approximate counting problems can be reduced to each other in polynomial time.
Dyer, Frieze, Kannan, Kapoor, Perkovic, and Vazirani~\cite{DyerFKKPV93} studied the MCMC (Markov chain Monte Carlo) algorithm for sampling almost uniform knapsack solutions, which led to an approximate counting algorithm with mildly exponential time $2^{O(\sqrt{n}(\log n)^{2.5})}/\eps^2$. 
Morris and Sinclair~\cite{MorrisS99} proved the first polynomial mixing time for the knapsack problem.
In the journal version~\cite{MorrisS04}, they improved the mixing time to $O(n^{4.5+\delta})$ for any $\delta > 0$, which implies an FPRAS for \#Knapsack 
via Jerrum-Valiant-Vazirani reduction.
In addition to MCMC algorithms, Dyer~\cite{Dyer03} gave an FPRAS with running time $\widetilde{O}(n^{2.5} + n^2/\eps^2)$ based on randomized rounding and dynamic programming, which is the currently fastest FPRAS (in the $\eps\ge 1/\polylog(n)$ regime) for \#Knapsack.

There is also a line of works focusing on (\emph{deterministic}) FPTAS (full polynomial-time approximation scheme) for \#Knapsack.
Gopalan, Klivans, Meka, \v{S}tefankovic, Vempala and Vigoda~\cite{GopalanKMSVV11} gave an FPTAS with running time $O(n^3 \eps^{-1} \log (n/\eps))$. 
Rizzi and Tomescu~\cite{RizziT19} further improved the running time by a $\mathrm{poly}\log n$ factor.
Finally,  Gawrychowski, Markin and Weimann~\cite{GawrychowskiMW18} gave an $\widetilde{O}(n^{2.5}/\eps^{1.5})$ time FPTAS, where the dependence on $n$ matches Dyer's randomized algorithm. 
All the FPTASes above are based on designing certain recursions for \#Knapsack and sparsifying the recursion functions to approximate the total number of solutions.

We are interested in the fine-grained complexity of approximating \#Knapsack.
So far, both the best-known FPTAS and FPRAS have running time $\widetilde{O}(n^{2.5})$ in the natural regime of $\eps \geq 1/\mathrm{poly}\log n$.  
On the other hand, due to the fast reduction from counting to sampling~\cite{SVV09,huber2015approx,Kol18}, many well-studied \#P-complete problems (e.g.\ the hardcore model and Ising model in the uniqueness regime~\cite{CLV21}) admit near-quadratic time approximation algorithms.
Hence, a natural question arises:
\begin{center}
\emph{Does \#Knapsack admit an FPRAS in near-quadratic time, or even sub-quadratic time?}
\end{center}

\subsection{Our results}
In this paper, we give a new algorithm that answers this question affirmatively.

\begin{theorem}[Sub-quadratic time FPRAS for \#Knapsack, main]\label{thm-main}
 There is an FPRAS for (Zero-One) \#Knapsack  in time \[O\Big (\frac{n^{1.5}}{\eps^2}\polylog\big (\frac{n}{\eps}\big ) \Big).\] 
\end{theorem}
The time complexity is measured in the standard word-RAM model, where every word contains $O(\log\frac{n}{\eps})$ bits. Up to logarithmic factors (which we do not optimize in this paper), the time complexity in word-RAM also reflects the bit complexity of the algorithm. We do \emph{not} use infinite-precision real-RAM, or word-RAM with very large word length.
We assume the knapsack capacity $T>0$ and item weights $0\le W_1,W_2,\dots,W_n\le T$ are given as input integers that each fit in one machine word, so that standard operations such as addition and comparison on these integers take unit time. This assumption was made in all previous works on \#Knapsack. If we drop this assumption, then the running time will be multiplied by $O(\log T)$, the bit-length of these integers.

Compared to all the previous results, our algorithm is \emph{the first to achieve sub-quadratic dependence on $n$}.
We also remark that $1/\eps^2$ is a natural dependence on parameter $\eps$ for Monte Carlo algorithms. 

Achieving sub-quadratic running time is rare in the approximate counting literature.
Many algorithms, especially those based on the counting-to-sampling reduction~\cite{JVV86,SVV09,huber2015approx,Kol18}, face a quadratic-time barrier because the algorithm needs to use enough samples to control the variance of the estimate. 
Very recently, a sub-quadratic time approximate counting algorithm~\cite{anand2024sub} was discovered for some special spin systems. 
The algorithm requires the spin system to exhibit a strong enough correlation decay property and utilizes such property to reduce the variance of the estimate. 
Other few recent examples of counting problems admitting sub-quadratic-time FPRASes include all-terminal unreliability estimation \cite{CenHLP24}, and counting spanning trees \cite{ChuGPSSW23} (which is not \#P-hard).

Our sub-quadratic FPRAS for \#Knapsack uses  techniques very different from the standard counting-to-sampling reduction.
Our starting point is Dyer's FPRAS \cite{Dyer03}, which rounds the weights of the items to $\poly(n)$, making it amenable to  dynamic programming, and then uses a Monte Carlo method to estimate the ratio of the solution counts in the original instance versus the rounded instance. 
We refine Dyer's framework with both structural insights and algorithmic advances:
\begin{itemize}
    \item 
    We show that there exists  a ``popular weight class'' parameterized by $\ell$, such that a ``typical'' item subset with total weight $\Theta(T)$ contains at least $\ell/\polylog(n)$ many items $i$ with weights  $W_i\approx T/\ell$.
    By adjusting the rounding scale based on $\ell$, this structural result allows us to use hitting-set arguments to improve Dyer's analysis on the ratio of solution counts in the original versus rounded instances, and hence decrease the sample complexity required by the Monte Carlo method. 
\item  We improve the dynamic programming using a ``partition-and-convolve'' paradigm, which has been successful for the optimization versions of Subset Sum and Knapsack problems in the last decade (e.g.\ \cite{KoiliarisX19,Bringmann17,Chan18a,BateniHSS18,Jin19,BringmannN21,soda2023knapsack,esa,ChenLMZ242,CLMZ24,Mao24,bringmann2024even}).\footnote{A note on nomenclature: In the literature of combinatorial optimization and fine-grained complexity, the term ``Knapsack problem'' commonly refers to maximizing the total profit $\sum_{i\in I}P_i$ subject to the knapsack constraint $\sum_{i\in I}W_i\le T$, while the ``Subset Sum problem'' corresponds to the $P_i=W_i$ special case of Knapsack. In our \#Knapsack problem, however, items only have weights and do not have profits, so it is more aligned with the Subset Sum problem. Terms such as ``knapsack solutions'' in this paper have nothing to do with profit-maximization.} Algorithmic ingredients such as Fast Fourier Transform (FFT) and  balls-into-bins hashing (color-coding) \cite{Bringmann17} have been useful for this paradigm. 
However, as a main difference from previous applications of this paradigm, our algorithm needs to convolve arrays of \emph{counts} (which may be as large as $2^n$) instead of Boolean values.
Such convolutions cannot be handled (even approximately) by the usual FFT on standard word-RAM without a significant slow-down. 
Interestingly, we overcome this challenge by using recent sub-quadratic-time algorithms for \emph{$(\max,+)$-convolution on bounded monotone arrays} \cite{ChiDX022, bringmann2024even}, and enforcing the required monotonicity by adapting the \emph{sum-approximation} framework from the previous FPTAS for \#Knapsack \cite{GawrychowskiMW18}.
\end{itemize}
We give a more detailed technical overview in \Cref{sec:overview}.

\subsection{Technical overview}\label{sec:overview}

 For a subset of items, $X\subseteq [n]$, we often use the shorthand $W_X := \sum_{i\in X}W_i$ to denote their total weight.
Let $\Omega := \{X\subseteq [n]: W_X \le T\} \subseteq 2^{[n]}$ denote the set of all knapsack solutions.

\subsubsection{Dyer's randomized rounding}
\label{sec:overdyer}
We first explain Dyer's randomized rounding framework \cite{Dyer03}, which we later modify and improve upon.

Let $K= \widetilde{\Theta}(\sqrt{n})$ be a parameter.
Assign every item $i \in [n]$  with a randomized rounded weight $w(i)$ with $\Ex[w(i)]=W_i$, such that $w(i)$ is an integer multiple of the \emph{scale parameter} $S = \frac{T}{2Kn}$.\footnote{In this overview, we ignore the (inconsequential) issue that $S$ itself might not be an integer, which could be fixed by rounding appropriately.} More specifically,  let $w(i) = (\lfloor \frac{W_i}{S}  \rfloor+ \Delta_i)\cdot S$, where $\Delta_i \in \{0,1\}$ is a random variable that takes value 1 with probability $\frac{W_i}{S} -  \lfloor \frac{W_i}{S} \rfloor$.
Define a slightly relaxed capacity  $T' := T + K\cdot S$, and
consider a rounded knapsack instance whose solution set is 
$\Omega':= \{X\subseteq [n]: \sum_{i\in X}w(i) \le T'\} \subseteq 2^{[n]}$.

Intuitively, the rounded instance should somehow preserve the original knapsack solutions to some extent.
It can reliably distinguish between original solutions $X\in \Omega$ versus subsets $X\subseteq [n]$ that are far away from $\Omega$: 
\begin{itemize}
    \item For any original solution  $X \in \Omega$,  its total rounded weight  $w(X):= \sum_{i\in X}w(i)$ has expectation $\Ex[w(X)] = W_X  \le T$ since the rounding is unbiased.  
    Since $X$ has at most $n$ items, each independently rounded with scale parameter $S$,  standard concentration bounds show that with high probability, $|w(X) -  W_X| \le  O(\sqrt{n\log n})\cdot S$,  and thus $w(X) \le T+K\cdot S = T'$ (by our choice of $K= \widetilde \Theta(\sqrt{n})$), implying $X \in \Omega'$.
Hence, we have $|\Omega| \approx |\Omega \cap \Omega'|$ with good probability.
    \item On the other hand, for any $X \subseteq [n]$ with original total weight $W_X \geq T + \frac{T}{n} = T+ 2K\cdot S$, its total rounded weight has expectation $\Ex[w(X)] = W_X \geq T' +K\cdot S$. Similarly,  one can show that with high probability, $w(X) > W_X - K\cdot S \geq T'$ and thus $X \notin \Omega'$.
    By a more careful analysis along this line, Dyer showed that a typical element in $ \Omega'$ is either a solution in $\Omega$ or exceeds the knapsack capacity by less than $T/n$, and then he used this property to prove that $|\Omega'|/ |\Omega| = O(n)$ holds with good probability (we omit the proof of this bound for now; later in \cref{sec:over-hitting} we actually discuss how to prove a better bound).
\end{itemize}
Based on these two properties, he proposed the following algorithm for estimating $|\Omega|$:
\begin{itemize}
    \item Count $|\Omega'| = |\{X \subseteq [n]:  \sum_{i\in X}w(i) \leq T'\}|$ by dynamic programming. Here, since the rounded weights and capacity can be treated as integers bounded by $t = \lceil \frac{T'}{S}\rceil = O(Kn) = \widetilde O(n^{1.5})$, the dynamic programming takes time $O(nt) = \widetilde{O}(n^{2.5})$.
    \item Estimate $|\Omega|/ |\Omega'| \approx |\Omega \cap \Omega'|/ |\Omega'| $ by drawing $O\big (\frac{|\Omega'|}{|\Omega|}/\eps^2\big ) = O(n/\eps^2)$ independent uniform samples $X_i \in \Omega'$ and see how many of them satisfy $X_i \in \Omega$. Sampling each $X_i$ can be implemented by backtracing the dynamic programming table  in $O(n)$ time. 
\end{itemize}
Multiplying the results of these two steps gives the desired estimate of $|\Omega|$. The total running time is $\widetilde O (n^{2.5}+n^2/\eps^2)$.

To achieve our target time complexity $\widetilde O(n^{1.5}/\eps^2)$, both steps in Dyer's algorithm need to be improved. In the following, we first focus on improving the $\widetilde{O}(n^{2.5})$-time dynamic programming step.  
For the rest of the overview, we assume $\eps\ge 1/\polylog(n)$ and omit the $\eps$-dependence for simplicity (since our main focus is improving the $n$-dependence).

\subsubsection{FFT-based techniques for Subset Sum type   problems}
Dyer's algorithm uses a vanilla dynamic programming to count $|\{X \subseteq [n]: \sum_{i\in X}w(i) \leq T'\}|$.
This task is closely related to the Subset Sum problem, which can be solved faster than dynamic programming using Bringmann's algorithm
 \cite{Bringmann17}. Specifically, Bringmann's algorithm computes all \emph{$t$-bounded subset sums}, $\{ \sum_{i\in X}W_i : X \subseteq [n]\} \cap \{0,1,\dots,t\}$, in $\widetilde O(n+t)$ time, whereas the textbook dynamic programming for this task runs in $O(nt)$ time \cite{Bellman1957}. 
Hence, a natural idea is to use Bringmann's techniques to improve the dynamic programming step in Dyer's counting algorithm.

A basic idea underlying
Bringmann's Subset Sum algorithm is the following \emph{partition-and-convolve} framework:  
Given two integer arrays $A[0\ldots t], B[0\ldots t]$,
 the Fourier Transform (FFT) algorithm computes their \emph{convolution} $C[0\dots 2t]$ where $C[k]:= \sum_{i+j=k}A[i]B[j]$ in $O(t\log t)$ time.
Hence, given $S_1,S_2\subseteq \{0,1,\dots,t\}$ consisting of  the $t$-bounded subset sums attained by disjoint item sets $I_1,I_2$ respectively,  we can compute the $t$-bounded subset sums attained by items of $I_1\cup I_2$ in $O(t\log t)$ time by convolving the two length-$(t+1)$ arrays representing the sets $S_1,S_2$.
This leads to a simple divide-and-conquer algorithm 
for computing all subset sums of $W_1,\dots,W_n$ in $\widetilde O(\sum_{i=1}^n W_i)$ time \cite{Eppstein97,KoiliarisX19}: in a binary recursion tree where the $i$-th leaf corresponds to $W_i$, let each internal node compute the subset sums attained by the leaves in its subtree, by convolving the results computed by its two children.  (This simple algorithm was refined by Bringmann \cite{Bringmann17} using more ideas.)

In the aforementioned scenario where we only care about whether a subset sum is attainable or not, it suffices to convolve arrays of Boolean values.
However, in Dyer's framework for \#Knapsack, we also need to \emph{count} the subsets which attain that sum.
If we directly adapt the partition-and-convolve paradigm of \cite{Eppstein97,KoiliarisX19,Bringmann17} to also keep track of the counts, then we immediately encounter the following obstacle:  the counts of subsets can be as large as $2^n$, so  the FFT algorithm for convolution needs to operate on $n$-bit numbers, incurring an $n$-factor blow-up in the time complexity in a standard word-RAM model (we do not use unrealistic models of computation such as infinite-precision real-RAM or word-RAM with huge words). 

Dealing with $n$-bit counts would have been an issue in Dyer's dynamic programming step too, but fortunately Dyer's algorithm only performs simple arithmetic operations (such as addition) on the counts, so this issue can be resolved by using floating-point numbers truncated to  $\polylog(n)$-bit precision, without affecting the final precision by too much.
However, such truncation would cause issues for FFT (which relies on properties of roots of unity that involve cancellation).

In fact, as commonly believed in fine-grained complexity, convolution with very large coefficients  is unlikely to have sub-quadratic-time algorithms (even allowing approximation), since it could simulate the difficult \emph{$(\max,+)$-convolution problem}, which asks to compute $C[k]:= \max_{i+j=k}\{A[i]+B[j]\}$ for all $0\le k\le 2n$, where $A[0\ldots n], B[0\ldots n]$ are input arrays of nonnegative integers.
If we define $A'[i]:= (2n)^{A[i]}$ and $B'[j]:= (2n)^{B[j]}$, then the (usual) convolution between $A', B'$ is $C'[k] = \sum_{i+j=k}(2n)^{A[i]+B[j]}$,  which can determine the $(\max,+)$ convolution of $A$ and $B$ by keeping the most significant bit of $C'[k]$. 
Currently, there is no known $O(n^{2-\delta})$-time algorithm (for any constant $\delta>0$) that solves $(\max,+)$-convolution even when $A[i],B[j] \le O(n)$.

\subsubsection{Bounded monotone \texorpdfstring{$(\max,+)$}{}-convolution and sum-approximation}\label{sec:over-sum}
The quadratic-time barrier for $(\max,+)$-convolution seems difficult to break in general, but fortunately it has been broken for important special cases.
 Here, we take inspiration from recent successes on (the profit-maximization version of) the Knapsack problem: Fine-grained complexity indicates that Knapsack has the same quadratic-time hardness as $(\max,+)$-convolution \cite{CyganMWW19,KunnemannPS17}, but  Bringmann and Cassis \cite{BringmannC22} and Bringmann, D\"urr, and Polak \cite{bringmann2024even} overcame this barrier for the bounded-profit-and-bounded-weight version of Knapsack, by using sub-quadratic-time algorithms for \emph{bounded monotone $(\max,+)$-convolution}.  This is the special case of $(\max,+)$-convolution where the two input arrays contain monotone non-decreasing integers from $\{0,1,\dots,M\}$, which can be solved in $\widetilde O(n\sqrt{M})$ time by \cite{bringmann2024even} (extending the result of 
Chi, Duan, Xie, and Zhang \cite{ChiDX022}). 

Back to our \#Knapsack scenario, where we would like to (approximately) convolve two arrays of counts (up to $2^n$), our key insight is that this task can also be solved faster using the bounded monotone $(\max,+)$-convolution machinery of \cite{ChiDX022,bringmann2024even}. In order to do this, we need to bridge the gaps between these two different variants of convolution problems:
\begin{itemize}
\item We need to transform the input arrays of counts into monotone non-decreasing arrays.

To do this, we instead work with \emph{prefix sums} of the array of counts (or equivalently viewed as CDFs), which are automatically monotone. 
This is allowed by the \emph{sum-approximation framework} of Gawrychowski,  Markin, and Weimann \cite{GawrychowskiMW18} developed for their \#Knapsack FPTAS.  
A sum-approximation of an array preserves the original prefix sums up to certain relative error, but the individual elements at any given index are allowed to differ a lot (see formal definitions in \cref{sec:sumapproxconv}).
Intuitively, the reason that sum-approximation is already sufficient for \#Knapsack is because 
we are always counting sets whose total weights are \emph{less than or equal to} some threshold, which is naturally a prefix sum.

    \item We need to simulate (usual) convolution of $A[0\ldots n],B[0\ldots n]$ using $(\max,+)$-convolution of some constructed arrays $A'[0\ldots n], B'[0\ldots n]$.
    
    This is the opposite direction of the simulation we saw earlier, but the idea is similar.  Specifically, represent $A[i]$ as $\hat A[i] \cdot D^{A'[i]}$ for some suitable small base $D$, where $A'[i] = O(\log A[i])$ is an integer and $\hat A[i] \in [1,D)$. Also represent $B[j]$ similarly.
    Then, the $(\max,+)$-convolution of $A',B'$ provides a crude approximation of the (usual) convolution of $A,B$. 
    In order to get a more precise approximation, we need to take the $\hat A[i],\hat B[j]$ values into account, and solve a ``weighted witness-counting'' version of $(\max,+)$-convolution. 
    This requires a white-box modification of the known bounded-monotone $(\max,+)$-convolution algorithms of \cite{ChiDX022, bringmann2024even}, which we describe in \cref{sec:max+}.
\end{itemize}

In this way, we can convolve (allowing sum-approximation)  two length-$n$ arrays with entries bounded by $2^M$  in time $\widetilde O(n\sqrt{M})$  (see \cref{lem:sum-approx-conv}). 
This is our basic tool for implementing the partition-and-convolve paradigm in Dyer's \#Knapsack framework. 
Although it circumvents the quadratic-time barrier,  the time complexity $\widetilde O(n\sqrt{M})$ is still worse than the ideal $O(n\log n)$-time FFT for convolving Boolean arrays, so the runtime analysis of the whole algorithm becomes more delicate, as we will see in \cref{sec:over-bounded}.

\subsubsection{Weight classes and balls-into-bins hashing}\label{sec:groupbin}
Now, we discuss two more ideas from Bringmann's Subset Sum algorithm \cite{Bringmann17} which are useful to us.

A simple but useful idea from \cite{Bringmann17} is to partition the input items into logarithmic many weight classes $(T/\ell, 2T/\ell]$ where $\ell$ ranges over powers of two. %
In Bringmann's scenario, it turned out to be convenient to first separately process items inside each weight class, and finally try to combine the answers from all classes.

Inspired by Bringmann's weight partitioning idea, in this overview, we first focus on the special case of \#Knapsack where \emph{all input items} come from the same weight class, that is,  we assume
\begin{equation}
\label{eqn:boundedratiocase}
\text{$W_i \in (T/\ell, 2T/\ell]$ for all $i\in [n]$, for some $\ell\ge 2$},
\end{equation}
which we call the \emph{bounded-ratio case} (note that we can assume $\ell \leq 2n$, since otherwise the number of solutions is trivially $2^n$).
In the following part of the overview, we will first describe an $\widetilde{O}(n^{1.5})$-time approximate counting algorithm for this bounded-ratio case.
Presenting the bounded-ratio case will allow us to illustrate many of our technical ideas without making the analysis too complicated. 
Later in the overview, we will introduce more ideas to move from the bounded-ratio case to the general case.

In our bounded-ratio case parameterized by $\ell$, any
solution $X \in \Omega = \{X\subseteq [n]: W_X \le T\}$ must contain less than $\ell$ items, since each item has weight $> T/\ell$.
To exploit this property, we borrow from \cite{Bringmann17} the final crucial idea of \emph{balls-into-bins} hashing:
create $\ell$ bins $B_1,B_2,\ldots,B_\ell$, and throw each item $i\in [n]$ independently into a uniform random bin. 
Since $|X|<\ell$, with high probability, every bin only contains at most $B = O(\frac{\log n}{\log \log n})$ items in $X$.
Now, instead of counting the knapsack solutions in $\Omega$, we count the solutions in $\Omega \cap \hat{\Omega}$, where  \[\hat{\Omega} := \{X \subseteq [n]: \forall b \in [\ell], |X \cap B_b| \leq B\}.\]
In words, $\Omega \cap \hat{\Omega}$ contains all knapsack solutions with at most $B= O(\frac{\log n}{\log \log n})$ items in each bin. 
The balls-into-bins
argument
guarantees that $|\Omega \cap \hat{\Omega}|$ gives a good approximation to $|\Omega|$ with high probability.

Now we briefly explain how the 
$|X \cap B_b| \leq B$ condition in the definition of $\hat \Omega$ helps to speed up the partition-and-convolve framework for counting $|\Omega \cap \hat \Omega|$. Consider the binary recursion tree with $\ell$ leaves, where each leaf node $b$ represents a bin $B_b$, and each internal node $u$ takes care of all items in $B_u := \bigcup_{b \in \text{leaf}(u)}B_b$, where $\text{leaf}(u)$ denotes all the leaf nodes in the subtree of $u$.
Then, at node $u$ where $|\text{leaf}(u)|=\ell/2^h$, we only care about item subsets $X \in 2^{B_u} \cap \hat \Omega$, which have size $|X|\le B\cdot \ell/2^h = \widetilde O(\ell/2^h)$ and hence total weight $W_X \le |X|\cdot 2T/\ell =  \widetilde O(T/2^h)$.

\subsubsection{Bounded-ratio case: partition-and-convolve with multi-level rounding}\label{sec:over-bounded}

In this section we give the remaining details of our partition-and-convolve algorithm for the bounded-ratio case parameterized by $\ell$ (which is a power of two). We use the notations from the previous section. 
In the binary recursion tree, we denote the level of the root as $0$, and the level of all $\ell$ leaf nodes  as $H := \log_2 \ell$.

We need a few more ideas in order to improve over Dyer's randomized rounding: 
\begin{itemize}
    \item \textbf{Rounding more aggressively:} 
    In our bounded-ratio case, any subset $X\subseteq [n]$ of total weight $O(T)$ only contains $O(\ell)$ items, which may be much smaller than $n$. Hence, when each item weight is rounded to a random multiple of the scale parameter $S$, the standard deviation of the total rounded weight of $X$ is only $O(\sqrt{\ell }\cdot S)$ instead of $O(\sqrt{n}\cdot S)$ in Dyer's original analysis described in \cref{sec:overdyer}. 
This means we can afford to choose a larger scale parameter $S$, which can make the dynamic programming tables (or the arrays to be convolved) shorter, and improve the time complexity.

In more detail, Dyer's original analysis 
described in \cref{sec:overdyer}
used the scale parameter $\frac{T}{n^{1.5}\polylog n}$, so that the total rounded weight of $X\subseteq [n]$ has rounding error $\le \frac{T}{n\polylog n}$.
 In our case, we can afford to choose 
the scale parameter $S_H:=\frac{T}{\ell^{1.5}\polylog n}$, so that the total rounded weight of $X$ has rounding error $\le \frac{T}{\ell \polylog n}$ (the reason that we now allow a larger rounding error will become clear in \cref{sec:over-hitting}).
    
    \item  \textbf{Multi-level rounding:} In addition to rounding the weight of each input item at the leaf level, we can also perform randomized rounding at each internal node $u$, i.e., we can round the total weight of subsets formed by items in $B_u = \bigcup_{b \in \text{leaf}(u)}B_b$.
    Explicitly maintaining the total rounded weight of every subset is clearly infeasible, so we have to perform the randomized rounding in an implicit way, which we will explain soon.

    When we move closer to the root, there are fewer tree nodes on the current level, which allows for more aggressive rounding at this level (by the same reason as in the previous bullet point).  More specifically, 
 at level $h$, we can perform randomized rounding independently at each of the $2^h$ tree nodes using the scale parameter 
$S_h := \frac{T}{\ell 2^{h/2} \polylog n}$, so that the total rounding error incurred at this level becomes $\widetilde O(\sqrt{2^h}\cdot S_h) = O(\frac{T}{\ell \polylog n})$. Over all $O(\log \ell)$ levels the total rounding error incurred is still $O(\frac{T}{\ell \polylog n})$.
\end{itemize}

Formally, a node $u$ at level $0\le h\le H$
is associated with an (implicit) random rounding function $w_u$ that maps each subset $X \in  2^{B_u} \cap \hat{\Omega}$ to its random rounded weight $w_u(X)$, such that
\begin{itemize}
    \item $w_u(X)$ is an integer multiple of the scale parameter $S_h$.
    \item $\Ex[w_u(X)] = W_X$, and $w_u(X)$ is concentrated around its mean so that with high probability, $w_u(X) \in W_X \pm \frac{T}{\ell 2^{h/2} \polylog n}$.
\end{itemize}
We stress that our algorithm does not need to compute or store $w_u(X)$ for all $X$; they are implicitly defined by the multi-level randomized rounding process of the algorithm. 
What the algorithm explicitly computes (and stores as an array of values) at each node $u$ is an
approximate counting function (which corresponds to the dynamic programming table in Dyer's algorithm),
 $f_u\colon \{0,S_h,2S_h,\ldots,L_hS_h\} \to \N $, such that  $f_u$ is a sum-approximation (see \Cref{sec:over-sum}) of the counting function $f^*_u(zS_h) := |\{ X \in 2^{B_u} \cap \hat{\Omega} : w_u(X) = zS_h \}|$.
Here, the parameter $L_h$ denotes the required length of the array representing that function, and can be bounded as follows:
For any $X \in 2^{B_u} \cap \hat{\Omega}$, we have
$|X| = \widetilde{O}(\ell/2^{h})$ and $W_X = \widetilde{O}(T/2^{h})$ (see the end of \cref{sec:groupbin}), so
 $L_h = O(w_u(X)/S_h) =O(W_X/S_h) = \widetilde{O}(\ell/2^{h/2})$.

We next explain how to compute all approximate counting functions $f_u$, and how the randomized rounding mapping $w_u$ is implicitly defined, in the order from the leaves to the root:
\begin{itemize}
    \item 
For each leaf node $b$, similarly to \cref{sec:overdyer}, we round the weight $W_i$ of each item $i\in B_b$ independently to $w_b(i) \in \{ S_H \lceil W_i/S_H \rceil, S_H \lfloor W_i/S_H \rfloor\}$ such that $\Ex[w_b(i)] = W_i$, and the total rounded weight of $X\in 2^{B_b}\cap \hat \Omega$ is $w_b(X) = \sum_{i \in X} w_b(i)$.
The counting function $f_b$ can be computed exactly by standard dynamic programming. 
Note that each value of $f_b$ is at most $O(n^B)$ (the number of subsets of size $\le B$, where $B= O(\frac{\log n}{\log \log n})$ due to balls-into-bins hashing), which can be represented by $\widetilde{O}(1)$ bits in binary.
The total running time of dynamic programming at all leaf nodes is therefore $\sum_{b\in [\ell]}\widetilde O(|B_b| L_H B) = \widetilde{O}(n\sqrt{\ell})$
 in the standard word-RAM model.
\item 
For each non-leaf node $u$ at level $h < H$ in the binary tree, suppose $u$ has two children $l$ and $r$. 
Any subset $X \in 2^{B_u}\cap \hat \Omega$ can be uniquely partitioned into $X_l = X \cap B_l\in 2^{B_l}\cap \hat \Omega$ and $X_r = X \cap B_r\in 2^{B_r}\cap \hat \Omega$. 
We can first use the children's weight functions $w_l$ and $w_r$ to define $w'_u(X) = w_l(X_l) + w_r(X_r)$, and apply our sum-approximation convolution algorithm (discussed in \Cref{sec:over-sum}) to $f_l$ and $f_r$, which returns an approximate counting function $f'_u$ with respect to the weight function $w'_u$. 
Then, we pick a unbiased random mapping $\alpha$ which rounds multiples of $S_{h+1}$ to neighboring multiples of $S_{h}$, and define $w_u(X) = \alpha(w'_u(X))$ as the weight function for node $u$. 
We can also compute the 
approximate counting function $f_u$ based on $f'_u$ and $\alpha$.
Readers can refer to \Cref{lemma:roundingsampler} and \Cref{lemma:mergingsampler} for details. 
The running time is dominated by computing the sum-approximate convolution, which is bounded as follows: The arrays representing $f_l$ and $f_r$ have length $L_{h+1} = \widetilde O(\ell/2^{h/2})$, and each value of $f_l,f_r$ is at most $O(n^{|X_l|}) \le n^{\widetilde O(\ell/2^h)}$, so each sum-approximate convolution takes time $\widetilde O(\ell/2^{h/2} \cdot \sqrt{\ell/2^h}) =\widetilde O(\ell^{1.5}/2^h)$, and
hence the total running time for all $2^h$ nodes at level $h$ is  $\widetilde{O}(\ell^{1.5})$.
\end{itemize}

Once we finish the computation at the root, we consider the following set of solutions 
\begin{align}\label{eq:over-def-Omega'}
    \Omega' = \left\{X \in \hat{\Omega}: w_{\Root}(X) \leq T + \frac{T}{\ell \polylog n} \right\}.
\end{align}
We can approximate $|\Omega'|$ by the approximate counting function $f_{\Root}$.
By the concentration property of the random function $w_{\Root}$, we have:
\begin{itemize}
    \item For any $X \in \hat{\Omega} \cap \Omega$, 
    the random weight satisfies $w_{\Root}(X) \leq W_X + \frac{T}{\ell \polylog n}\le T + \frac{T}{\ell \polylog n}$ with high probability, which means $X \in \Omega'$ with high probability (provided we set the $\polylog n$ factors in the definitions appropriately). Hence, $|\hat \Omega \cap \Omega| \approx | \Omega' \cap \Omega|$. Combined with the balls-into-bins property, we get $|\Omega|  \approx | \Omega' \cap \Omega|$.
    \item On the other hand, for any $X \in \hat{\Omega}$ such that $W_X > T + \frac{T}{\ell}$, with high probability $w_{\Root}(X) > W_X - \frac{T}{\ell \polylog n} \geq T + (1-o(1))\frac{T}{\ell}$. Hence, we can show $X \notin \Omega'$ with high probability. 
By combining this with an argument in~\cite{Dyer03}, we can show that most of $X \in \Omega' \setminus \Omega $ satisfy $T < W_X \le T + \frac{T}{\ell}$ (see the proof of \cref{claim:Omega'} for details). 
\end{itemize}

Hence, we can estimate the ratio $\frac{|\Omega|}{|\Omega'|} \approx \frac{|\Omega' \cap \Omega|}{|\Omega'|}$ by drawing almost uniform samples $X$ from $\Omega'$ and counting the number of samples satisfying $X \in \Omega$. 
To draw a sample $X\in \Omega'$,
 we go from top to bottom in the recursion tree, where at each internal node we use the computed approximate counting functions to randomly determine how much capacity should be allocated to the left and right children. 
Here we remark that the sum-approximation and rounding make this sampling procedure trickier; we omit the details and refer the interested readers to the proof of \cref{lemma:roundingsampler} and \cref{lemma:mergingsampler}.
  The running time for drawing one sample $X\in \Omega'$ is
  dominated by the total time for computing all the approximate counting functions $f_u$, which is $\widetilde{O}(\ell^{1.5})$. 
Later in \Cref{sec:over-hitting} we will show that it is sufficient to draw $\widetilde{O}(n/\ell)$ samples.
Hence, the total running time of the approximate counting algorithm is (recall that $\ell = O(n)$)
\[\widetilde{O}(n\sqrt{\ell}) + \widetilde{O}(\ell^{1.5})+\widetilde{O}(n/\ell) \cdot \widetilde{O}(\ell^{1.5}) = \widetilde{O}(n^{1.5}).
\]

\subsubsection{Bounded-ratio case: improving the sample complexity}\label{sec:over-hitting}
Recall that Dyer's original Monte Carlo method needs to draw $\widetilde O(n)$ samples from $\Omega'$. 
Now we show how to reduce the sample complexity from $\widetilde O(n)$ down to $\widetilde O(n/\ell)$.
According to the analysis in \Cref{sec:over-bounded}, we need to estimate $\frac{|\Omega' \cap \Omega|}{|\Omega'|}$, %
 where $\Omega'$ is defined in~\eqref{eq:over-def-Omega'} and $\Omega = \{X \subseteq [n]: W_X\le T\}$.  
Hence, the number of samples required is $O(\frac{|\Omega'|}{|\Omega' \cap \Omega|})$. 
From the previous section we already know that $|\Omega' \cap \Omega| \approx |\Omega|$, and most of $X \in \Omega' \setminus \Omega$ are contained in $\Omega_1 := \{X \subseteq [n]: T< W_X \le T + \frac{T}{\ell}\}$. Hence, the sample complexity is $O(\frac{|\Omega_1|}{|\Omega|}+1)$.

Note that the definitions of $\Omega$ and $\Omega_1$ are independent from the randomness of the algorithm. Bounding the ratio $|\Omega_1|/|\Omega|$ is a purely combinatorial problem.
 Since the maximum weight of an item is $W_i \leq 2T/\ell$, every $X \in \Omega_1$ must have $|X|>\ell/2$ items. Consider a random subset $H \subseteq [n]$ that includes each item independently with probability $\frac{ \polylog n }{\ell}$. 
The typical size of $H$ is $\widetilde{O}(n/\ell)$.
For each $X \in \Omega_1$, we say $H$ \emph{hits} $X$ if $X \cap H \neq \emptyset$. 
Every $X$ is hit by $H$ with probability at least $1 - (1 - \frac{\polylog n}{\ell} )^{\ell/2} = 1 - o(1)$.
By the probabilistic method, there exists $H$ with size $\widetilde{O}(n/\ell)$ that hits $1 - o(1)$ fraction of $X$ in $\Omega_1$. 
For those $X\in \Omega_1$ hit by $H$, if we throw away an arbitrary element $x \in X\cap H$, since $W_x >\frac{T}{\ell}$, then $W_{X\setminus \{x\}}< T$ and hence $X \setminus \{x\} \in \Omega$. Viewing this as a mapping from $\{X\in \Omega_1: H \text{ hits }X\}$ to $\Omega$,  we know each $I\in \Omega$ has at most $|H|$ preimages (because $x\in H$). Hence, we have the following bound of the ratio
\begin{align*}
    \frac{(1-o(1))|\Omega_1|}{|\Omega|}=\frac{|\{X \in \Omega_1: H \text{ hits } X\}|}{|\Omega|} \leq |H| = \widetilde{O}(n/\ell),
\end{align*}
which implies $\frac{|\Omega_1|}{|\Omega|} = \widetilde{O}(n/\ell)$. Therefore, the number of samples needed is $\widetilde{O}(n/\ell)$ as claimed.

We remark that this argument 
is a refined version of the mapping argument from Dyer's original analysis \cite{Dyer03}.
Our refinement relies on both sides of the bounded-ratio assumption $\frac{T}{\ell}<W_i \le \frac{2T}{\ell}$: (1) we use $\max_{i\in X} W_i\le \frac{2T}{\ell}$ to conclude $|X|>\ell/2$ for the purpose of constructing a small hitting set $H$, and (2) we need the removed item $x$ to be larger than the additive gap $\frac{T}{\ell}$ in the defintion of $\Omega_1$ (which constrains the allowed rounding error $\frac{T}{\ell \polylog n}$ in the algorithm from the previous section).
In comparison, in Dyer's original analysis without the bounded-ratio assumption, the mapping argument (by throwing away the largest item) could only yield a  weaker bound $|\{X \subseteq [n]: T< W_X \le T + \frac{T}{n}\}|/|\Omega| \le n$.

\subsubsection{General case: identifying a ``popular weight class'' }
\label{overview:popular}
Now we turn to the general case without the bounded-ratio assumption.
Following \cite{Bringmann17}, we start by partitioning the items into
 $g+1 = \lceil \log_2 n \rceil+1$ weight classes as described earlier, where the $j$-th  class ($1\le j\le g$) contains all items with weights $ \frac{T}{2^j}< W_i \leq \frac{T}{2^{j-1}}$, and the $(g+1)$-st class contains all small items with weights $W_i\le \frac{T}{2^g}$.  
Unfortunately, unlike in Bringmann's algorithm \cite{Bringmann17}, here we face some challenges if we try to apply the bounded-ratio algorithm on each weight class separately and then
combine the results from different classes to solve the whole instance. 
A major challenge is that we need to draw samples consisting of items from all the weight classes, whereas our analysis of sample complexity only applies to each individual weight class.
Another challenge is that different weight classes have different rounding errors (e.g. for the weight class $(T/m, 2T/m]$, the rounding error at the root of the binary tree is $\widetilde{\Theta}(T/m)$), making the time complexity worse if we try to directly merge $O(\log n)$ classes together to obtain the final answer. (The challenge of merging weight classes has also appeared in other knapsack-type problems, such as \cite{CLMZ24}.)

To circumvent these challenges, our strategy is to try to transform the general input instance to resemble the bounded-ratio case as closely as possible. 
For this purpose, we prove a structural lemma, which indicates that there exists a parameter $2\le \ell \le O(n)$ (which can be identified efficiently) such that most of the subsets $X\subseteq [n]$ with total weight $W_X\in \Theta(T)$ should contain at least $\ell/\polylog(n)$ items from the (slightly enlarged) weight class $(T/\ell, \polylog(n)\cdot T/\ell]$ (see formal statement in \cref{lem:manyitemssolutions}). 
In this way, we can set the scale parameters so that the final rounding error is $\frac{T}{\ell \polylog n}$ as before, and then the same hitting-set argument from the previous section can still apply, implying that the sample complexity is $\widetilde O( n/\ell)$ as desired (\cref{lem:omega1bound}).
(We also need the additional property that all input items of weights $\le T/\ell$ only have total weight $< T/2$, but we omit the details from this overview. See \cref{eqn:smallgrouptotalub}.)

Now, it remains to compute the approximate counting functions for all weight classes $(T/m,2T/m]$ (but with the final  rounding error still set to $\approx T/\ell$, where $\ell$ can be different from $m$), and merge the results from all the weight classes. 
The algorithm for each class mostly follows the partition-and-convolve paradigm  described earlier for the bounded-ratio case, but with a few additional technical challenges: (1) the analysis of time complexity now involves both parameters $m$ and $\ell$, and becomes slightly more complicated; (2) at the bottom level of the recursion tree, the standard dynamic programming is sometimes too slow, and we have to use a different construction algorithm based on Jin and Wu's Subset Sum algorithm \cite{JinW19} which supports counting modulo a  small prime, as well as the birthday-paradox color-coding technique used by Bringmann \cite{Bringmann17} (see details in \cref{lem:leaf}); (3)
The most serious issue occurs when both the sample complexity and the typical size of a sample are large, which we detail in the next section.

\subsubsection{Handling tiny items in the samples}
\label{sec:overzero}

The following extreme case could happen: The popular weight class mentioned in the previous section has parameter $\ell = \polylog(n)$. On the other hand, there exists $\Omega(n)$ many tiny items each with weight in $[0, T/n^{10}]$. In this case our algorithm generates $n/\ell = \widetilde \Omega (n)$ samples, but each sample is expected to contain $\Omega(n)$ tiny items from the interval $[0, T/n^{10}]$. Hence, producing these samples would take quadratic time. 
These tiny items should not be blindly discarded from the input, as they may affect the final count significantly. For example, consider an instance with $n/2$ large items each of weight $\frac{T}{\ell}\cdot (1- \frac{1}{n^{10}})$ (for some integer $\ell = \Theta(\log n)$) and $n/2$ tiny items each of weight $\frac{10T}{n^{11}}$. The number of subsets with total weight $\le T$ equals $C_1=\binom{n/2}{\le \ell-1} \cdot 2^{n/2} + \binom{n/2}{\ell} \cdot \binom{n/2}{\le n/10}$. If we modify the weight of each tiny item to $\frac{T}{n^{11}}$, then the answer suddenly becomes 
$C_2=\binom{n/2}{\le \ell} \cdot 2^{n/2}$, where 
\[\frac{C_1}{C_2} \le  \frac{\binom{n/2}{\le \ell-1}}{\binom{n/2}{\le \ell}} + \frac{\binom{n/2}{\le n/10}}{2^{n/2}} \le O\big(\frac{\log n}{n}\big ). \]
This phenomenon is in sharp contrast to the bounded-ratio case where we never considered items of weights much smaller than $T/n$.

To resolve this issue,  the key observation is that these tiny items are typically rounded to zero in our algorithm.
Let $I_0$ denote the set of items whose weights are rounded to zero (at the leaf level of the binary recursion tree), but their original weights $(W_j)_{j \in I_0}$ before rounding can be positive.
For an item $i\in I_0$, whether including $i$ in the sample or not does not change the total rounded weight of the sample, so we can independently treat each $i\in I_0$ as being implicitly included in the sample with $1/2$ probability, but we do not explicitly report $i$ in the sample. 
 Hence, we can produce \emph{partial samples} that only include items with positive rounded weights, which typically have much smaller size than a full sample. 

Suppose our algorithm has produced $ N = \widetilde O( n/\ell)$ partial samples $X_1,X_2,\ldots,X_N\subseteq [n]\setminus I_0$, where each $X_i$ should be combined with a uniform random subset $X_0\subseteq I_0$ to form a full sample.  
Then, our goal becomes approximately counting the solutions $X \in \{X_0 \cup X_i: X_0\subseteq I_0 \land i \in [N] \}$ such that $W_X \leq T$.
This task can be phrased as  another instance of (a slight generalization of) \#Knapsack, where the input items have weights $(W_j)_{j \in I_0}$, and additionally we have a generalized item that takes one of the $N$ choices of weights ($W_{X_1},\dots,W_{X_N}$). 
We now proceed to the second phase our algorithm for solving this new (generalized) \#Knapsack instance. The second-phase algorithm requires techniques similar to what we have discussed so far for the first phase, but fortunately the new instance satisfies certain properties making it easier to solve (in particular, we do not need to proceed to a third phase), which we briefly explain next.

By the choice of $\ell$ and the sample complexity bound mentioned in the previous section,
we know the fraction of valid knapsack solutions in this new instance is $\ge \frac{1}{\widetilde O(n/\ell)}$ with good probability.
On the other hand, the fraction of knapsack solutions in the new instance that avoid the  $k$ largest input items in $I_0$ can be upper bounded by $\frac{1}{2^k}$.
Hence, by picking $k= O(\log n)$ large enough, the $k$ largest items of $I_0$ form a hitting set which hits most of the valid knapsack solutions in the new instance. Then, by a similar argument as in \cref{sec:over-hitting}, if we apply our \#Knapsack algorithm again on this new instance, the required sample complexity is only $O(k) =O(\log n)$. 
Therefore, the new instance can be solved easily in time $\widetilde{O}(n^{1.5})$ by our algorithm. See details in \cref{subsec:mainthm,sec:secondalg}.

\subsection{Open problems}
The most important question is whether (Zero-One) \#Knapsack admits an FPRAS with near-linear running time $O(n\polylog(n) \poly(\eps^{-1}))$. Even in the bounded-ratio special case where all input weights $W_1,\dots,W_n \in (T/n, 2T/n]$, we currently do not know any algorithm in $n^{1.5-\Omega(1)} \cdot \poly(1/\eps)$ time.
Another important question is whether the deterministic FPTASes of \cite{GopalanKMSVV11,GawrychowskiMW18} can be improved to near-quadratic or even sub-quadratic time.

\#Knapsack has other versions such as integer \#Knapsack and multi-dimensional \#Knapsack (see e.g., \cite{Dyer03}), and it also connects to many other problems such as contingency tables~\cite{Dyer03}, learning functions of halfspaces~\cite{GopalanKM10} and approximating total variation distances~\cite{GMMPV23}.
One interesting open problem is to explore whether our algorithm or techniques could be applied to these problems.

\subsection*{Paper Organization}
In \cref{sec:prelim} we give a few basic notations and definitions. 

In \cref{sec:structure} we identify the popular weight class for the input instance, and prove structural properties essential for analyzing the sample complexity. 

In order to describe our main algorithm in a modular fashion, 
in \cref{sec:basicalgotools} we introduce the Approximate Knapsack Sampler data structure as a convenient interface. 
It encapsulates the key ingredients of Dyer's framework: randomized rounding,  approximate counting functions (or dynamic programming tables), and approximately sampling  knapsack solutions. 
It further supports rounding and merging operations, which are repeatedly used in our partition-and-convolve framework.
Then, in \cref{sec:algorithm}, we describe our main algorithm using the Approximate Knapsack Samplers as building blocks. 
In \cref{sec:sum-app}, we present the sum-approximation convolution algorithm, which is used to implement the merging operation for Approximate Knapsack Samplers.

As mentioned in the technical overview, our algorithm becomes simpler in the bounded-ratio case 
(\cref{eqn:boundedratiocase}).
For easier reading, we add remarks when some technical parts of the paper can be simplified for the bounded-ratio case.

\section{Preliminaries}
\label{sec:prelim}
\subsection{Basic notations and assumptions}
We use $\log$ to denote logarithm with base $e$ and $\log_2$ to denote logarithm with base $2$.
Let $[n] = \{1,2,\dots,n\}$. Let $\N = \{0,1,2,\dots\}$. Let $2^S = \{S': S'\subseteq S\}$.

We often use the shorthand $x\pm \delta$ to denote the interval $[x-\delta,x+\delta]$.

For a set $I$, let $2^I$ denote $\{I' : I'\subseteq I\}$, the collection of all subsets of $I$.
For two collections of subsets $\Omega_1 \in 2^{I_1}$ and $\Omega_2 \in 2^{I_2}$, where $I_1\cap I_2 = \emptyset$, denote
\begin{align*}
    \Omega_1 \times \Omega_2 = \{X_1 \cup X_2 \mid X_1 \in \Omega_1, X_2 \in \Omega_2\}.
\end{align*}

Let $W_1\le W_2\le \dots \le W_n$ be the weights of the $n$ input items sorted in non-decreasing order, and $T>0$ be the knapsack capacity.
Assume $0< W_i\le T$ (since otherwise we can remove the item), and assume $W_1+\dots +W_n >T$ (since otherwise the number of knapsack solutions equals $2^n$).
Denote the total weight of a subset of items $I\subseteq [n]$ by \[W_I := \sum_{i\in I}W_i.\]

We assume $n$ is greater than a large constant; otherwise, the input has constant size and we can solve the problem by brute force.
We assume $n^{-1.5} < \eps < 1/10^4$. This is because for smaller $\eps$, the previous FPTAS by \cite{GopalanKMSVV11} actually has better running time $\widetilde O(n^3/\eps)$ than our time bound $\widetilde O(n^{1.5}/\eps^2)$.
If $\eps > 1/10^4$, we can simply set $\eps = 1/ 10^4$ and it only affects the running time by a constant factor.

Finally, we assume the input weights and capacity are sufficiently large,
\begin{align}\label{eq:assume}
    T \geq \left(\frac{n}{\eps}\right)^{50} \text{ and } \min_{i \in [n]}W_i \geq \left(\frac{n}{\eps}\right)^{50}.
\end{align}
Otherwise, we can scale all $T$ and $(W_i)_{i=1}^n$ by a factor $C = \poly\left(\frac{n}{\eps}\right)$. Since every number is given in binary, it increases the input size by at most an $O(\log \frac{n}{\eps})$ factor.
This assumption is merely for the technical reason that we want to always work with integers  in the algorithm, so when we work with fractions such as $T/\poly(n/\eps)$ we want their rounded values $\lfloor T/\poly(n/\eps)\rfloor $ to not introduce too much error.

Our algorithm uses the well-known Fast Fourier Transform (FFT) algorithm (e.g., \cite[Chapter 30]{cormen2022introduction}, \cite[Section 4.3.3]{knuth2014art}) to compute the convolution of two arrays:
\begin{lemma}[FFT]\label{lemma-FFT}
Given two arrays $A[0\dots n]$ and $B[0\dots n]$ where all $A[i],B[j]$ are integers from $[-2^M, 2^M]$, we can deterministically compute $C[k] = \sum_{i+j=k} A[i] B[j]$ for all $0\le k\le 2n$ in $O(Mn\polylog(Mn))$ time (in the standard word-RAM model with $O(\log n + \log M)$-bit machine words).
\end{lemma}

\subsection{Probability background}

A random variable $X$ with mean $\mu$ is called \emph{$\sigma^2$-subgaussian} (or \emph{subgaussian with parameter} $\sigma^2$) if 
\[ \Ex[e^{\lambda(X-\mu)}]\le e^{\lambda^2 \sigma^2/2} \text{ for all } \lambda\in \R.\]
A $\sigma^2$-subgaussian random variable $X$ with mean $\mu$ satisfies the following Hoeffding-type concentration property: for any $t > 0$, 
\begin{align*}
    \Pr[X - \mu > t] \leq \exp \left(-\frac{t^2}{2\sigma^2}\right) \quad \text{and} \quad \Pr[X - \mu < -t] \leq \exp \left(-\frac{t^2}{2\sigma^2}\right).
\end{align*}
A bounded random variable $X\in [a,b]$ is $\frac{(b-a)^2}{4}$-subgaussian. For independent subgaussian random variables $X_1,X_2$ with parameters $\sigma_1^2,\sigma_2^2$, their sum $X_1+X_2$ is a $(\sigma_1^2+\sigma_2^2)$-subgaussian random variable.

We use the following lemma to analyze the randomized rounding performed by our algorithm.
\begin{lemma}[Rounding a subgaussian]\label{lem-subg}
Let $X$ be a $\sigma^2$-subgaussian random variable. Let $S > 0$.
Let $\overbar{X} = S \cdot \lceil X/S \rceil$ and $\myunderbar{X} = S\cdot  \lfloor X/S \rfloor$. Define random variable $X' = \myunderbar{X}$ with probability $\frac{\overbar{X} - X}{S}$ and $X' = \overbar{X}$ otherwise. Then, $X'$ is $(\sigma^2 + \frac{S^2}{4})$-subgaussian with mean $\Ex[X]$.
\end{lemma}
\begin{proof}
Conditioned on $X$, we have
\begin{align*}
\Ex[X' \mid X] = \myunderbar{X} \cdot \frac{\overbar{X}-X}{S} + \overbar{X}\cdot \frac{S-\overbar{X}+X}{S} = X.
\end{align*}
Hence, $ \Ex[X'] = \Ex\big [\Ex[X'\mid X]\big ] = \Ex[X]$ by the law of total expectation.
Denote $\mu = \Ex[X]=\Ex[X']$. For any $\lambda$, by the law of total expectation again,
 \[\Ex[e^{\lambda(X'-\mu)}] = \Ex\big [\Ex[e^{\lambda(X'-\mu)} \mid X]\big ] = \Ex\big [e^{\lambda(X - \mu)}\Ex[e^{\lambda(X'-X)} \mid X]\big ] .\] 
Conditioned on $X$, the random variable $X'-X$ belongs to the interval $[\myunderbar{X},\overbar{X}]$ of length $\le S$, and hence is  $S^2/4$-subgaussian. Therefore, we have
\begin{align*}
   \Ex\big [e^{\lambda(X - \mu)}\Ex[e^{\lambda(X'-X)} \mid X]\big ] \leq \Ex[e^{\lambda(X - \mu)} e^{\lambda^2 S^2/8}] \leq e^{\lambda^2(\sigma^2 + S^2/4)/2}.
\end{align*}
Hence, $X'$ is $(\sigma^2 + \frac{S^2}{4})$-subgaussian.
\end{proof}

Let $U(S)$ denote the uniform distribution over the finite set $S$.

Let $\mu$ and $\pi$ be two distributions over some finite state space $\Omega$. The \emph{total variation distance} between $\mu$ and $\pi$ is defined by 
\begin{align*}
    \Vert \mu - \pi \Vert_{\text{TV}} = \frac{1}{2}\sum_{x \in \Omega}|\mu(x)-\pi(x)|.
\end{align*}
Let $X \sim \mu$ and $Y \sim \pi$ be two random variables. We may abuse the notation and denote the total variation distance $\Vert \mu - \pi \Vert_{\text{TV}}$ by $\Vert X - Y \Vert_{\text{TV}}$ or $\Vert X - \pi \Vert_{\text{TV}}$.

A \emph{coupling} of two distributions $\mu$ and $\pi$ is a joint distribution of $(X,Y)$ such that $X \sim \mu$ and $Y \sim \pi$. We say $X$ and $Y$ are coupled successfully in a coupling if the event $X = Y$ occurs. The following coupling inequality is well-known (e.g. see~\cite{mitzenmacher2017probability}).
\begin{lemma}[Coupling Inequality]\label{lemma-Coupling}
For any coupling $(X,Y)$ of $\mu$ and $\pi$, it holds that 
\begin{align*}
    \Vert \mu-\pi \Vert_{\mathrm{TV}} \leq  \Pr[X \neq Y].
\end{align*}
Furthermore, there exists an optimal coupling that achieves equality.
\end{lemma}

\section{Structural properties of the popular weight class}
\label{sec:structure}
Recall that the input weights are sorted, $W_1\le W_2\le \dots \le W_n$, with total sum greater than $T$.
Pick the unique index $i^*\in [n]$ such that
\begin{equation}
    \label{eqn:defnistar}
 W_1+W_2+\dots + W_{i^*-1} < T/2 \le W_1+W_2+\dots + W_{i^*}.
\end{equation}
The following simple lemma identifies a popular weight class:
\begin{lemma}
    \label{claim:existsell}
    There exists an integer $\ell \in [2, 8n)$ such that   
    \[ |\{j: 1\le j\le i^{*}, W_j \in (T/\ell,2T/\ell]\} |> \frac{\ell}{8\log_2(8n)}. \]
Moreover, such $\ell$ can be found in $\widetilde{O}(n)$ time.
\end{lemma} \begin{proof} 
    Let $G:= \{j \in [i^*]: W_j> T/ 2^{\lceil \log_2 4n\rceil}  \}$.
     Then, the items in $[i^*]\setminus G$ have total weight at most $i^*\cdot T/ 2^{\lceil \log_2 4n\rceil}\le T/4$, so $G$ has total weight at least $(W_1+\dots +W_{i^*}) - T/4 \ge  T/4$ by \eqref{eqn:defnistar}.
    Partition $G$ into subsets  $G_\ell :=\{j\in [i^*]: W_j \in (T/\ell, 2T/\ell]\}$, where $\ell$ ranges from $\{2^1,2^2,\dots,$ $ 2^{\lceil \log_2 4n\rceil }\}\subset [2,8n)$. 
     Among these $\lceil \log_2 4n\rceil$ subsets, there must exist a $G_\ell$ with total weight at least $\frac{T/4}{\lceil \log_2 4n\rceil}$, and hence $|G_\ell| \ge \frac{T/4}{\lceil \log_2 4n\rceil}/ (2T/\ell)> \frac{\ell}{8\log_2 8n }$ as desired.
 \end{proof}
 From now on, we fix the value of $\ell \in [2,8n)$ obtained from \cref{claim:existsell}. 

Recall  $\Omega := \{X\subseteq [n] : \sum_{j\in X}W_j\le T\}$ is the set of all solutions for the input knapsack instance. 
Now partition all the non-solutions $2^{[n]}\setminus \Omega$ into the following sets:
\begin{align}\label{eq:def-Omega-d}
 \text{for all integers } d \geq 1,\quad    \Omega_d \defeq \Big\{X \subseteq [n] \, : \, T + (d-1)\cdot \frac{T}{\ell} < \sum_{j \in X}W_j \leq T+ d\cdot \frac{T}{\ell}  \Big\}.
\end{align}
The goal of this section is to establish the following two upper bounds on $|\Omega_d|$ in terms of $|\Omega|$, which will be used later in \cref{subsec:mainthm} to  analyze the sample complexity of our main algorithm. The first bound is similar to \cite{Dyer03}, while the second bound  achieves a nearly $\ell$-factor improvement specifically for $|\Omega_1|$:
\begin{lemma}
    \label{lem:omegadbound}
For every integer $d\ge 1$, $ |\Omega_d|\le n^{d} \cdot |\Omega|$.
\end{lemma}
\begin{lemma}
    \label{lem:omega1bound}
    $|\Omega_1| \le \frac{15000n(\log_2 n )^2}{\ell}\cdot |\Omega|$.
\end{lemma}

In the following, we call input items $j\in [n]$ with $W_j> T/\ell$ the \emph{large items}. 
Note that \cref{claim:existsell} implies $W_{i^*}> T/\ell$. 
In particular,  this means all the non-large items have total weight
\begin{equation}
\label{eqn:smallgrouptotalub}
\sum_{j\in [n]: W_j \le  T/\ell} W_j \le W_1+W_2+\dots +W_{i^*-1} \overset{\text{by \eqref{eqn:defnistar}}}{<} T/2.
\end{equation}
Now we can prove \cref{lem:omegadbound}.
\begin{proofof}{Proof of \cref{lem:omegadbound}}
Fix any $X\in \Omega_d$. If we remove all the large items from $X$, then the remaining items have total weight $< T/2<T$ (by \eqref{eqn:smallgrouptotalub}). Since $X$ has total weight $\le T + d\cdot T/\ell$, and each large item has weight $> T/\ell$, we know there exists $Y\subseteq X$ of size $1\le |Y|\le d$ such that $X\setminus Y$ has total weight $< T$ (that is, $X\setminus Y \in \Omega$).
View this as a mapping from $X \in \Omega_d$ to $X\setminus Y \in \Omega$. Then the number of pre-images of any $X'\in \Omega$ is at most the number of possibilities for 
the subset $Y$, which is $\le \sum_{y=1}^d\binom{n}{y} \le n^d$. Hence, $|\Omega_d| \le |\Omega| \cdot n^d$.
\end{proofof}

To prove \cref{lem:omega1bound}, we need the following technical lemma, which informally says that a typical set of items with total weight $\Theta(T)$ should contain many items of weight $\approx T/\ell$: 
\begin{lemma}
    \label{lem:manyitemssolutions}
 Let 
 \[ \Omega^\triangle:= \Big \{X\subseteq [n] : \sum_{j\in X}W_j\in (T,2T]\text{ and } |\{j\in X: W_j> T/\ell\}| < \frac{0.01\ell}{(\log_2(8n))^2}\Big \}.\]
Then, $ |\Omega^\triangle|\le  0.01|\Omega|$.
\end{lemma}
\begin{proofof}{Proof of \cref{lem:omega1bound} using \cref{lem:manyitemssolutions}}
To bound $|\Omega_1|$, our plan is to bound the size of \[\Omega_1 \setminus \Omega^{\triangle} = \Big \{X\subseteq [n] : \sum_{j\in X}W_j\in (T,T+T/\ell]\text{ and } |\{j\in X: W_j> T/\ell\}| \ge  \frac{0.01\ell}{(\log_2(8n))^2}\Big \}\] in terms of $|\Omega|$.
 Consider a random hitting set $H \subseteq \{j\in [n]: W_j > T/\ell\}$ constructed by sampling each large item $j$ independently with probability $\frac{1600(\log_2  n)^2}{\ell}$.  
We say $X \in \Omega_1\setminus \Omega^{\triangle}$ is \emph{good} if $X \cap H \neq \emptyset$,  and let $\Omega_{\text{good}} \subseteq \Omega_1 \setminus \Omega^\triangle$ denote the set of all good $X$. For $X \in \Omega_1 \setminus \Omega^\triangle$, since $|\{j\in X: W_j >T/\ell\}| \ge 0.01\ell /(\log_2(8n))^2$, we have $\Pr[X \text{ is not good}] \le  \big (1 - \frac{1600 (\log_2 n)^2}{\ell}\big )^{0.01\ell/(\log_2(8n))^2} < 1/e$. 
 By Markov's inequality, with probability greater than $\frac{1}{2}$, $\frac{|\Omega_{\text{good}}|}{|\Omega_1\setminus \Omega^\triangle|} \geq 1 - \frac{2}{e}$.
 On the other hand, with probability at least $1/2$, $|H| \leq \frac{3200n (\log_2 n)^2}{\ell}$.
 By probabilistic method, there exists $H \subseteq  \{j\in [n]: W_j > T/\ell\}$ with $|H| \leq \frac{ 3200n (\log_2 n)^2}{\ell}$ such that $|\Omega_{\text{good}}| \geq (1 - \frac{2}{e})|\Omega_1\setminus \Omega^\triangle|$.

 For any $X \in \Omega_{\text{good}}$, if we remove an arbitrary $i \in H \cap X$ from $X$, then the total weight of $X\setminus \{i\}$ is at most $(T+T/\ell) - T/\ell = T$, so $X \setminus \{i\}\in \Omega$. 
 Hence 
 \begin{align*}
     |\Omega| \cdot |H| \geq |\Omega_{\text{good}}| \geq  (1 - \frac{2}{e})|\Omega_1\setminus \Omega^\triangle| .
 \end{align*}
Combining with \cref{lem:manyitemssolutions}, we have 
\begin{align*}
 |\Omega_1| \le |\Omega_1\setminus \Omega^\triangle| + |\Omega^\triangle| \le \frac{|\Omega|\cdot |H|}{1-2/e} + 0.01|\Omega| <  \frac{15000|\Omega|n(\log_2 n)^2}{\ell},&%
\end{align*}
as claimed.
\end{proofof}
Now we prove \cref{lem:manyitemssolutions}.
\begin{proofof}{Proof of~\cref{lem:manyitemssolutions}}
    We first consider the easy case where $\ell\le  100(\log_2 (8n))^2$. The definition of $\Omega^\triangle$ simplifies to $\Omega^\triangle = \big \{X\subseteq [n]: \sum_{j\in X}W_j \in (T,2T] \text{ and } \{j\in X: W_j >T/\ell\} = \emptyset\big \}$.
By \eqref{eqn:smallgrouptotalub}, if $\{j\in X: W_j >T/\ell\} = \emptyset$, then $\sum_{j\in X}W_j< T/2 < T$. Hence, $|\Omega^\triangle| = 0 < 0.01|\Omega|$.
   
   From now on we assume $\ell> 100(\log_2 (8n))^2$.
 Our plan is to construct a nonempty bipartite graph between $\Omega^\triangle$ and $\Omega$, in which the degree of any $X\in \Omega$ is at most $0.01$ times the degree of any $X^\triangle\in\Omega^\triangle$. This would imply $|\Omega^\triangle| \cdot \min_{X^\triangle\in \Omega^\triangle}\deg(X^\triangle)  \le |\Omega| \cdot \max_{X\in \Omega}\deg(X) \le |\Omega| \cdot 0.01\cdot  \min_{X^\triangle\in \Omega^\triangle}\deg(X^\triangle)$ and hence $|\Omega^\triangle|\le 0.01|\Omega|$ as desired.

We call input items $j\in [n]$ with $W_j\in (T/\ell,2T/\ell]$ the \emph{good items}, and call items with $W_j \ge 40(\log_2   (8n))^2 \cdot T/\ell$ the \emph{huge items}.
Note that huge items and good items are disjoint, and they all are large items.

For any $X^\triangle \in \Omega^\triangle$, define $\hat X^\triangle\subseteq [n]$ to be $X^\triangle$ with all huge items removed, and define the neighbors of $X^\triangle$ in our bipartite graph as
\[ N(X^\triangle):= \{ X\in \Omega: \hat X^\triangle\subseteq X, \text{ and $X\setminus \hat X^\triangle$ only contains good items}\}. \]
We first lower bound the degree of $X^\triangle \in \Omega^\triangle$ in this bipartite graph:
\begin{claim}
For all $X^\triangle\in \Omega^\triangle$,   $|N(X^\triangle)|\ge 2^{0.1\ell/\log_2 (8n)}$
\end{claim}
\begin{proof}
   First we bound the total weight of $\hat X^\triangle$.
   Since $\hat X^\triangle\subseteq X^\triangle\in \Omega^\triangle$, by definition of $\Omega^\triangle$, $\hat X^\triangle$ contains fewer than $\frac{0.01\ell}{(\log_2 (8n))^2}$ large items.
   Since $\hat X^\triangle$ does not have huge items, every item in $\hat X^\triangle$ has weight $< 40(\log_2 (8n))^2 \cdot T/\ell$, so the large items in $\hat X^\triangle$ have total weight $\le \frac{0.01\ell}{(\log_2 (8n))^2} \cdot 40(\log_2 (8n))^2 \cdot T/\ell = 0.4T$.
The non-large items in $\hat X^\triangle$ have total weight at most $T/2$ by \eqref{eqn:smallgrouptotalub}. Hence, the total weight of $\hat X^\triangle$ is at most $0.4T+T/2 = 0.9T$.

Since $X^\triangle$ has $< \frac{0.01\ell}{(\log_2 (8n))^2}$ good items (recall that all good items are large), and $[n]$ has more than $\frac{\ell}{8\log_2(8n)}$ good items (by \cref{claim:existsell}), we know $[n]\setminus X^\triangle$ has more than $\frac{\ell}{8\log_2(8n)} - \frac{0.01\ell}{(\log_2 (8n))^2}  >  \frac{\ell}{10\log_2 (8n)}$ good items.  
Pick an arbitrary subset of good items $G\subseteq [n]\setminus X^\triangle$ with $|G| = \lceil \frac{\ell}{10\log_2 (8n)} \rceil$. 
The total weight of $G$ is at most 
$|G| \cdot (2T/\ell) < 0.1T$ (here we used the assumption $\ell> 100(\log_2 (8n))^2$).
Hence, adding any subset of items of $G$ to $\hat X^\triangle$ would yield a knapsack solution of total weight $\le 0.9T + 0.1T = T$. This means $|N(X^\triangle)| \ge 2^{|G|} \ge 2^{0.1\ell/\log_2 (8n)}$ as claimed.
\end{proof}

The other direction is to upper bound the degree of $X\in \Omega$ in this bipartite graph:
\begin{claim}
    For all $X\in \Omega$, $|\{X^\triangle\in \Omega^\triangle: N(X^\triangle)\ni X\}| \le 2^{0.06\ell/ \log_2(8n)}$.
\end{claim}
\begin{proof}
    Fix $X\in \Omega$. 
    We proceed in two steps according to the definition of $N(X^\triangle)$: First, we bound the number of possibilities for $\hat X^\triangle$ such that $\hat X^\triangle \subseteq X $, and $X\setminus \hat X^\triangle$ only contains good items. Second, we bound the number of possibilities for $X^\triangle\setminus \hat X^\triangle$. 
Since $X^\triangle$ is determined by $X^\triangle = \hat X^\triangle \cup  (X^\triangle\setminus \hat X^\triangle)$, the product of these two upper bounds gives the desired upper bound on the number of possibilities for $X^\triangle$.

For the first step, since $\hat X^\triangle\subseteq X^\triangle\in \Omega^\triangle$, by definition of $\Omega^\triangle$, we know $\hat X^\triangle$ can only contain fewer than $\frac{0.01\ell}{(\log_2 (8n))^2}$ good items (recall  that all good items are large).
Since $\hat X^\triangle \subseteq X$, and $X\setminus \hat X^\triangle$ only has good items, the only  possibilities for $\hat X^\triangle$ are sets obtained by removing all but $<\frac{0.01\ell}{(\log_2 (8n))^2}$ good items from $X$.  
The number of such possibilities is the number of ways to choose $<\frac{0.01\ell}{(\log_2 (8n))^2}$ good items that should remain in $\hat X^\triangle$, which is at most $\sum_{0\le x < \frac{0.01\ell}{(\log_2 (8n))^2}}\binom{n}{x} <2^{\frac{0.01\ell}{\log_2(8n)}}$.  

For the second step, by definition of $\hat X^\triangle$, we know $ X^\triangle\setminus  \hat X^\triangle$ can only contain huge items. 
Since $X^\triangle \in \Omega^\triangle$, by definition of $\Omega^\triangle$, the huge items in $ X^\triangle\setminus  \hat X^\triangle$ have total weight at most $2T$. Since each huge item has weight $\ge 40(\log_2 (8n))^2  \cdot T/\ell$, we have  $| X^\triangle\setminus  \hat X^\triangle| \le \frac{2T}{40(\log_2 (8n))^2  \cdot T/\ell} = \frac{\ell}{20 (\log_2 (8n))^2}$. Hence, the number of possibilities for $ X^\triangle\setminus \hat X^\triangle$ is at most $\sum_{0\le x \le \frac{\ell}{20(\log_2(8n))^2 }}\binom{n}{x} < 2^{\frac{\ell}{20 \log_2(8n) }}$.

By multiplying the two bounds, we conclude $|\{X^\triangle\in \Omega^\triangle: N(X^\triangle)\ni X\}| \le 2^{\frac{0.01\ell}{\log_2(8n)}} \cdot 2^{\frac{\ell}{20\log_2(8n)}} = 2^{\frac{0.06\ell}{\log_2(8n)}}$.
\end{proof}

Combining the two claims together, we get
\begin{align*}
     \log_2 \frac{|\Omega|} {|\Omega^\triangle|} \ge \log_2 \frac{\min_{X^\triangle\in \Omega^\triangle}\deg(X^\triangle)}{\max_{X\in \Omega}\deg(X)}  \ge \frac{0.1\ell}{\log_2 (8n)} - \frac{0.06\ell}{\log_2(8n)}> 4\log_2(8n),
\end{align*}
where the last step follows from the assumption $\ell> 100(\log_2 (8n))^2$.
Hence, $|\Omega|/|\Omega^\triangle| \ge (8n)^4> 100$. This finishes the proof of \cref{lem:manyitemssolutions}.
\end{proofof}

\section{Basic algorithmic primitives}
\label{sec:basicalgotools}
In \cref{sec:sumapproxconv}, we state the definition of sum-approximation and several useful lemmas, which will be used in \cref{subsec:aks}. Then, in \cref{subsec:aks}, we define the Approximate Knapsack Sampler data structure, which is an interface that allows us to describe our main algorithm in a modular fashion later. In \cref{subsec:roundmerge} we describe the merging and rounding operations for Approximate Knapsack Samplers.

\subsection{Sum-approximation}
\label{sec:sumapproxconv}
For two functions $f,g\colon \N \to \R$, define their \emph{convolution} \[(f\star g)(x) = \sum_{0\le y\le x} f(y)g(x-y).\]
For a function $f\colon \N \to \R_{\ge 0}$, define $f^{\le}(x) = \sum_{y\le x} f(y)$. Note that $f^\le\colon \N \to \R_{\ge 0}$ is monotone non-decreasing.
 For two functions $f,\shorttilde f\colon \N \to \R_{\ge 0}$ and $\delta \ge 0$, we say $\shorttilde f$ is a \emph{$(1\pm\delta)$-sum-approximation} of $f$, if for every $x\in \N$, 
\[ (1-\delta)  f^\le(x)\le  \shorttilde f^\le(x) \le (1+\delta) f^\le(x).\]
This notion is weaker than the standard pointwise $(1\pm \delta)$-approximation which would require $(1-\delta)  f(x)\le  \shorttilde f(x) \le (1+\delta) f(x)$ for all $x\in \N$.

The sum-approximation factor grows multiplicatively under convolution:
\begin{lemma}[\cite{GawrychowskiMW18}]\label{prop:apx}
If for $i\in \{1,2\}$, $\shorttilde f_i$ is a $(1\pm\delta_i)$-sum-approximation of $f_i$,  then $\shorttilde f_1\star \shorttilde f_2$ is a $(1\pm \delta_1)(1 \pm \delta_2)$-sum-approximation of $f_1\star f_2$. 
\end{lemma}

We will use $(1\pm\delta)$-sum-approximation in the regime $\delta \approx 1/2^{\polylog(n)}$. This is in contrast to the work of \cite{GawrychowskiMW18} where $\delta$ is a fixed inverse polynomial in $n$ (and the time dependence on $1/\delta$ is polynomial).

A key subroutine in our algorithm is to approximately compute the convolution of two sequences whose values can be as large as $2^M$.
By allowing sum-approximation, we can achieve a nearly $\sqrt{M}$ factor improvement over the standard FFT algorithm (\Cref{lemma-FFT}):
\begin{theorem}\label{lem:sum-approx-conv}
Given any $\delta > 0$ and two functions $f,g\colon \{0,1,\dots,n\} \to \{0\}\cup [1,2^M]$ (assume $f(x)=g(x)=0$ for all $x>n$), where $f(x),g(x)$ are given in $O(\log(1/\delta)+\log M)$-bit floating-point representation,  
we can compute a function $h\colon \{0,1,\dots,2n\} \to \{0\}\cup \R_{\ge 1}$ by a randomized algorithm in $\widetilde O(n\sqrt{M})$ time such that 
$h$ is a $(1\pm\delta)$-sum-approximation of $f\star g$  with probability at least $1 - \frac{1}{\poly(n)}$, where $\widetilde{O}(\cdot)$ hides a $\polylog(n,M,1/\delta)$ factor.
\end{theorem}
The proof of \cref{lem:sum-approx-conv} is given in \cref{sec:sum-app}.

\subsection{Approximate knapsack sampler}
\label{subsec:aks}
In this section, we introduce the main data structure, \emph{Approximate Knapsack Sampler} (or \emph{sampler} for short) which serves as a convenient interface that encapsulates (and generalizes) the behavior of Dyer's randomized rounding framework.
In particular, the definition includes a scale parameter $S$, an (implicit) rounded weight function $w$, and an approximate counting function $\hat f$, whose roles have been explained in \cref{sec:over-bounded}.
 It also includes parameters such as the length $L$ of the array representing the function $\hat f$, and the subgaussianity parameter $\sigma^2$ which controls the rounding error.

Recall that $(W_i)_{i=1}^n$ are the weights of the input knapsack instance, and for an item subset $I\subseteq [n]$ we use $W_I$ to denote $\sum_{i\in I}W_i$.

For easier understanding, we first define a basic version of Approximate Knapsack Sampler, and then extend it to the full version. 
\begin{definition}[Basic Approximate Knapsack Sampler]\label{def:basic}
   Let $I$ be a set of input items, and $\Omega \subseteq 2^{I}$ be a collection of subsets of items.
   Define a \emph{Basic Approximate Knapsack Sampler for $\Omega$} with parameters  $S\in \Z^+, L\in \Z^+$, $\sigma^2 \ge 0$, $\delta\in [0,1]$, and $\caT_{c},\caT_q$, as a randomized data structure with the following behavior.
    With at least $1-\delta$ success probability (over the randomness of the construction stage): 
   \begin{itemize}
    \item \textbf{Construction stage}: 
    It tosses random coins and implicitly defines a random mapping $w\colon \Omega \to \{0,S,2S,\ldots,LS\}=S\cdot \{0,1,\ldots,L\} $, such that for every $X \in \Omega$, $w(X)$ is a $\sigma^2$-subgaussian random variable with mean $W_X$.
   
   It then outputs a $(1 \pm \delta)$-sum-approximation $\hat{f}$ of the function $f\colon \{0,S,2S,\dots,LS\} \to  \N$ defined as 
$f(t\cdot S) := |\{X\in \Omega: w(X) = t S\}|$, and each value of $\hat{f}$ is in  $\{0\}\cup \R_{\ge 1}$ and is represented by an $O(\log(\frac{1}{\delta}) + \log \log |\Omega|)$-bit floating-point number. The whole construction stage takes $\caT_c$ time.
    \item \textbf{Query stage}: Given a query $T\in \N^+$, it tosses fresh random coins and in $\caT_q$ time samples a random pair $(X, w(X))$ where $X\in \Omega$, such that $X$ follows the uniform distribution on $\{X\in \Omega: w(X) \le T\}$ up to at most $2\delta$ error in TV distance.
   \end{itemize}
   \label{defn:knapsacksampler}
   \end{definition}

The basic version in \cref{defn:knapsacksampler} is sufficient for a wide class of input knapsack instances (including the bounded-ratio case, see \cref{remark:boundedratio}), %
but to solve general instances we need a more generalized version which is allowed to return a much smaller ``partial sample'' during the query stage.
The purpose here is to avoid spending query time outputting those tiny input items whose weights are rounded to zero, as discussed in \cref{sec:overzero} in the technical overview. 
   \begin{definition}[Approximate Knapsack Sampler, full version]
       \label{defn:knapsacksamplertiny}
Apply the following modifications to \cref{defn:knapsacksampler}:  

\begin{itemize}
    \item \textbf{Construction stage}:  It should additionally output a random subset $I_0\subseteq I$ of items called \emph{tiny items}, which satisfies $X_0 \cup (X\setminus I_0)\in \Omega$ for all $X\in \Omega,X_0 \subseteq I_0$, and $w(X)=w(X\setminus I_0)$ for all $X\in \Omega$ (in particular, $w(X_0)=w(\emptyset) = 0$ for all $X_0\subseteq I_0$).
        \item \textbf{Query stage}: 
Given a query $T\in \N^+$, it tosses fresh random coins and in $\caT_q$ time samples a random pair $(X_+, w(X_+))$ where $X_+\in \Omega \cap 2^{I\setminus I_0}$, such that $X_+$ follows the uniform distribution on $\{X_+\in \Omega \cap 2^{I\setminus I_0}: w(X_+) \le T\}$ up to at most $2\delta$ error in TV distance.
\end{itemize}

   \end{definition}
   \begin{remark}
       \begin{itemize}
        \item If $I_0=\emptyset$ always holds, then \cref{defn:knapsacksamplertiny} is the same as the basic version in \cref{defn:knapsacksampler}.
               \item In the query stage, if we additionally sample $X_0\subseteq I_0$ uniformly and independently, then $X_+\cup X_0$ follows the uniform distribution on $\{X\in \Omega: w(X)\le T\}$ up to at most $2\delta$ error in TV distance.
       \end{itemize}
   \end{remark}

\subsection{Rounding and merging samplers}
\label{subsec:roundmerge}
We present two primitives for manipulating Approximate Knapsack Samplers: rounding (\cref{lemma:roundingsampler}) and merging (\cref{lemma:mergingsampler}). In this section, we use the notation $\widetilde{O}(\cdot)$ to hide a factor of
\begin{align*}
    \polylog\left( |I|, \log|\Omega|, L, \frac{1}{\delta}\right) \cdot \log(T)
\end{align*}
where $I,\Omega,L,\delta$ are parameters of  Approximate Knapsack Sampler and $T$ is the maximum capacity for the query stage. 
\Cref{lemma:mergingsampler} has two approximate samplers and $\widetilde{O}(\cdot)$ hides the above $\polylog$ factor for parameters of both two samplers.

We will focus on proving \cref{lemma:roundingsampler,lemma:mergingsampler} for basic Approximate Knapsack Samplers. For the full version of Approximate Knapsack Samplers, the proofs only need minor modifications, which we will discuss at the end of this section.

The following lemma captures the action of performing another more aggressive rounding on top of an existing approximate knapsack sampler.
   \begin{lemma}[Rounding]
       \label{lemma:roundingsampler}
   Let $S\le S'$ be positive integers. Suppose there is an Approximate Knapsack Sampler for $\Omega\subseteq 2^I$ with parameters $ S, L, \sigma^2,\delta$, and $\caT_q$.

   Then, there is an Approximate Knapsack Sampler for the same $\Omega$ with parameters $S', \lceil \frac{LS}{S'}\rceil, \sigma^2 + (S')^2/4, \delta$ with $\Delta\caT_c'= \widetilde{O}(L), \caT_q' = \caT_q + \widetilde{O}(1)$ time,  where $\Delta T_c'$ denotes the incremental cost to modify the input sampler to the new sampler. %
   \end{lemma}
   \begin{proof}
We first construct an unbiased random mapping $\alpha$ that rounds every weight $xS$ where $x\in \{0,1,\dots, L\}$ to a new weight that is a multiple of $S'$:
For each integer $0\le y\le  \lceil \frac{LS}{S'} \rceil$, we independently sample an integer 
 $r_y \in [yS',(y+1)S')$ uniformly at random, and we jointly round all $xS \in [yS',(y+1)S')$ via the same random threshold $r_y$ by mapping $xS$ to
\begin{align*}%
 \alpha(xS) := \begin{cases}
    yS' &\text{ if } xS \le  r_y;\\
    (y+1)S' &\text{ if } xS>r_y.
    \end{cases}
\end{align*}
Observe that $\Ex_{r_y}[\alpha(xS)] =xS$, and $\alpha(\cdot)$ is monotone non-decreasing.
Let $w\colon \Omega \to S\cdot \{0,1,\dots,L\}$ be the (implicit) weight mapping of the input  Approximate Knapsack Sampler.
Then, let $L' \defeq  \left\lceil\frac{LS}{S'}\right\rceil$, and define the (implicit) weight mapping $w'\colon \Omega \to S'\cdot \{0,1,\dots,L'\}$ 
for the new sampler as
\begin{align*}
 w'(X) = \alpha(w(X)).
\end{align*}
For any $X\in \Omega$, random variables $w(X)$ and $w'(X)$ fit the condition in \Cref{lem-subg}. Hence $\Ex[w'(X)] = \Ex[w(X)] = W_X$ and $w'(X)$ is $(\sigma^2 + (S')^2/4)$-subgaussian.

Let $\hat{f}$ be the approximate counting function of the input Approximate Knapsack Sampler. 
For the new Approximate Knapsack Sampler (with scale parameter $S'$ and length parameter $L'$),  we compute the new approximate counting function $\hat{f}'\colon S'\cdot \{0,1,\ldots,L'\} \to  \{0\}\cup \R_{\ge 1}$ defined as
\begin{align*}
 \forall\, 0 \leq y \leq L', \quad \hat{f}'(yS')  :=  \sum_{\substack{ x \in \{0,1,\ldots,L\}:\\ \alpha(xS)=yS'}}\hat{f}(xS).
\end{align*}
We now show $\hat{f}'$ is a $(1 \pm \delta)$-sum-approximation of $f'(yS')=|\{X \in \Omega: w'(X) = yS'\}|$. By definition of $\hat f'$ and $\alpha$, we have 
\begin{align*}
 (\hat{f}')^{\le}(yS')  =  \sum_{\substack{ x \in \{0,1,\ldots,L\}:\\ \alpha(xS)\le yS'}}\hat{f}(xS) = \sum_{\substack{ x \in \{0,1,\ldots,L\}:\\ xS\le r_y}}\hat{f}(xS) = \hat{f}^{\le}(r_y).
\end{align*}
Since $\hat{f}$ is a $(1 \pm \delta)$-sum-approximation of $f(xS)=|\{X \in \Omega: w(X) = xS\}|$, we have 
\[\hat f^{\le}(r_y)\in (1\pm \delta)f^{\le}(r_y) = (1\pm \delta) |\{X\in \Omega: w(X) \le r_y\}|.\]
By definition of $w'$ and $\alpha$, we have $|\{X\in \Omega: w(X) \le r_y\}| = |\{X\in \Omega: w'(X) \le yS'\}| = (f')^{\le}(yS')$. Chaining together gives $(\hat{f}')^{\le}(yS')  \in (1\pm \delta)(f')^{\le}(yS')$ as desired.

The incremental cost to construct the new sampler is the time for computing the function $\alpha$ %
together with the time for computing $\hat{f}'$.
Hence, the total incremental cost is $\Delta\caT_c' = \widetilde{O}(L)$.

We next show how to generate samples for the new Approximate Knapsack Sampler.
Given a query capacity $T \in \N$, our goal is to sample almost uniformly from $\{X\in \Omega: w'(X)\le T\}$. By definition of $w'$ and $\alpha$, this set is the same as 
$\{X\in \Omega: w(X)\le xS\}$, where $x$ is the largest
integer in $\{0,1,\ldots,L\}$ such that $\alpha(xS) \leq T$, which can be found by binary search.
We query the input Approximate Knapsack Sampler\footnote{In the above construction stage, we only computed some additional data but did not change anything about the input sampler. Hence, we can still query in the input sampler.} with capacity $xS$ to obtain a pair $(X,w(X))$. Then let the new sampler return $(X, w'(X))$, where $w'(X)=\alpha(w(X))$.
Hence, for the new data structure, the query stage takes time $\caT'_q = \caT_q + \widetilde{O}(1)$ due to binary search.

The success probability for the input sampler is at least $1 - \delta$, so the success probability for the new sampler is also at least $1 - \delta$.
\end{proof}

The following lemma describes merging two existing approximate samplers for $\Omega_1 \subseteq 2^{I_1}$ and $\Omega_2 \subseteq 2^{I_2}$, where $I_1,I_2$ are disjoint sets of items.
This is the only part of the paper that uses the sum-approximation convolution algorithm (\cref{lem:sum-approx-conv}).
Recall that $\Omega_1 \times \Omega_2$ denotes $\{X_1 \cup X_2 \mid X_1 \in \Omega_1, X_2 \in \Omega_2\}$.
\begin{lemma}[Merging]
    \label{lemma:mergingsampler}
    Suppose there are two Approximate Knapsack Samplers on $\Omega_1\subseteq 2^{I_1}$ and $\Omega_2\subseteq 2^{I_2}$ respectively (where $I_1\cap I_2 =\emptyset$) with the same scale parameter $S$, and the $i$-th sampler ($i\in \{1,2\}$) has parameters $S,L_i,\sigma_i^2, \delta_i$, $\caT_{i,q}$ 
   (where $\delta_i < 1/10$).

  Then, %
  there is an Approximate Knapsack Sampler for $\Omega_1\times \Omega_2$ with parameters $S,L_1+L_2,\sigma_1^2+\sigma_2^2$, $4\delta_1 + 4\delta_2$, and the incremental cost to modify two input samplers to the new sampler
 \[\Delta \caT_c = \widetilde{O}\left((L_1+L_2)\sqrt{\log(|\Omega_1||\Omega_2|)}\right),\]
 and query time
 \[
\caT_q = \caT_{1,q} + \caT_{2,q} + \widetilde{O}(L_1+L_2).\]
\end{lemma}
\begin{proof}
Let $\delta := 4\delta_1+4\delta_2$.
For $X_1\in \Omega_1, X_2\in \Omega_2$,   in the new sampler let $w(X_1\cup X_2) = w_1(X_1)+w_2(X_2)$.
We have $\Ex[w(X_1\cup X_2)]=\Ex[w_1(X_1)] + \Ex[w_2(X_2)]= W_{X_1}+W_{X_2} = W_{X_1 \cup X_2}$, and since the two samplers are independent, $w(X_1\cup X_2)$ is $(\sigma_1^2+\sigma_2^2)$-subgaussian.

We first describe the construction stage. 
Let $\hat{f}_i\colon S\cdot \{0,1,\ldots,L_i\} \to  \{0\}\cup \R_{\ge 1}$ denote the approximate counting function of the $i$-th input sampler, which is a $(1 \pm \delta_i)$-sum-approximation of the function $f_i(tS) = \{X \in \Omega_i: w_i(X) = tS\}$.
(In particular, $\hat{f}_i(tS) \leq (1+\delta_i)|\Omega_i|$.)
The new counting function that we want to approximate is the convolution $(f_1 \star f_2)(tS) = |\{X \in \Omega_1 \times \Omega_2: w(X) = tS\}|$.
By \Cref{prop:apx}, $\hat{f}_1 \star \hat{f}_2$ is a $(1 \pm \delta_1)(1 \pm \delta_2)$-sum-approximation of $f_1 \star f_2$. 
We 
use \Cref{lem:sum-approx-conv} to compute a $(1\pm \frac{\delta}{10})$-sum-approximation of $\hat{f}_1 \star \hat{f}_2$ with success probability at least $1 - \frac{\delta}{20}$.\footnote{The functions $\hat{f}_i$ can be viewed as $\{0,1,\ldots,L_i\} \to  \{0\}\cup \R_{\ge 1}$.}
The algorithm outputs a function $\hat{f}'\colon S\cdot \{0,1,\ldots,L\} \to  \{0\}\cup \R_{\ge 1}$, where $L = L_1 + L_2$, and the output values of $\hat{f}'$ are represented by $O(\log\log(|\Omega_1||\Omega_2|) + \log(1/\delta'))$-bit floating-point numbers.
Hence, if the algorithm succeeds, then $\hat{f}'$ is a sum-approximation of $f_1 \star f_2$ with approximation ratio $(1 \pm \frac{\delta}{10})(1\pm \delta_1)(1\pm \delta_2) \in (1\pm \delta)$ by the assumption $\delta_1,\delta_2 < 1/10$.
The incremental cost for construction is dominated by the time complexity of \cref{lem:sum-approx-conv}, which is 
\[\Delta\caT_c =\widetilde{O}(L\sqrt{\log (|\Omega_1||\Omega_2|) }),\]
where $\widetilde{O}(\cdot)$ hides a $\polylog (L, \log(|\Omega_1||\Omega_2|), 1/\delta)$ factor.

Now we describe the query stage.
Given a query $T\in \N^+$,
our goal is to sample almost uniformly from $\{X\in \Omega_1\times \Omega_2: w(X)\le T\}$.
Without loss of generality we assume $T=xS$ where $x\in \{0,1,\dots,L\}$. %
Recall $\hat{f}_i^\leq$ denotes the prefix sum function of $\hat{f}_i$.
The sampling algorithm contains the following steps.
\begin{enumerate}
 \item \label{step-1} Sample $y \in \{0,1,\ldots,x\}$ with probability $\propto \hat{f}_1(yS)\hat{f}_{2}^{\leq}((x-y)S)$;
        \item \label{step-2} Query the second sampler with $(x-y)S$ to obtain $(X_2,w_2(X_2))$; 
        \item \label{step-3} Query the first sampler with $xS-w_2(X_2)$ to obtain $(X_1,w_1(X_1))$;
        \item \label{step-4} Return the pair ($X_1 \cup X_2$, $w_1(X_1) + w_2(X_2) $). (Recall $w(X_1 \cup X_2) = w_1(X_1) + w_2(X_2)$.)
\end{enumerate}
    It is easy to see that the time complexity of this sampling algorithm is $\caT_q = \caT_{1,q} + \caT_{2,q} + \widetilde{O}(L_1+L_2)$.
It remains to prove the correctness of this sampling algorithm.
By the assumption on the input sampler, Step~\ref{step-2} returns $X_2$ that is $2\delta_2$-close to the uniform distribution of $\{X \in \Omega_2: w_2(X) \leq (x-y)S\}$, and Step~\ref{step-3} returns $X_1$ that is $2\delta_1$-close to the uniform distribution of $\{X \in \Omega_1: w_1(X) \leq xS - w_2(X_2)\}$. Hence, $X_1$ (and $X_2$) can be coupled successfully with true uniform samples with probability at least $1 - 2\delta_1$ (and $1 - 2\delta_2$);  for now we first assume that $X_1$ and $X_2$ are replaced by true uniform samples.
Under this assumption, we bound the TV distance $D$ between the output $X_1 \cup X_2$ and the uniform distribution of $\{X \in \Omega_1 \times \Omega_2: w(X) \leq xS\}$ (recall $T=xS$), as follows:
For any $X_2 \in \Omega_2$ with $w_2(X_2) \leq xS$, it holds that
\begin{align*}
    \Pr[\text{Step~\ref{step-2} returns }X_2] &= \frac{\sum_{y=0}^x \hat f_1(yS) \hat f_2^\le ((x-y)S)\Pr[\text{Step~\ref{step-2} returns }X_2 \mid \text{Step~\ref{step-1} samples }y ]}{\sum_{z=0}^x \hat f_1(zS) \hat f_2^\le ((x-z)S)}\\
&= \frac{\sum_{y=0}^x \hat f_1(yS) \hat f_2^\le ((x-y)S) \mathbf{1}[w_2(X_2)\le (x-y)S] \cdot \frac{1}{f_2^{\le}((x-y)S)}}{\sum_{z=0}^x \hat f_1(zS) \hat f_2^\le ((x-z)S)}\\
& = \frac{\sum_{y=0}^x \hat f_1(yS) ( 1 \pm \delta_2 )   \mathbf{1}[w_2(X_2)\le (x-y)S] }{\sum_{z=0}^x \hat f_1(zS) \hat f_2^\le ((x-z)S)}\\
& = \frac{( 1 \pm \delta_2 ) \hat f_1^\le (xS-w_2(X_2))}{\sum_{z=0}^x \hat f_1(zS) \hat f_2^\le ((x-z)S)}.
\end{align*}
For every $(X_1,X_2)\in \Omega_1\times \Omega_2$ such that $w_1(X_1)+w_2(X_2)\le xS$, it holds that
   \begin{align*}
    \Pr[\text{Step~\ref{step-4} returns }X_1 \cup X_2] &= \Pr[\text{Step~\ref{step-2} returns }X_2]  \Pr[\text{Step~\ref{step-3} returns }X_1 \mid \text{Step~\ref{step-2} returns }X_2]\\
&= \Pr[\text{Step~\ref{step-2} returns }X_2] \cdot  \frac{1}{f_1^\le(xS-w_2(X_2))}\\
    & =  \frac{( 1 \pm \delta_2 ) \hat f_1^\le (xS-w_2(X_2))}{\sum_{z=0}^x \hat f_1(zS) \hat f_2^\le ((x-z)S)}
\cdot   \frac{1}{f_1^\le(xS-w_2(X_2))}\\
& = \frac{(1 \pm \delta_1)(1 \pm \delta_2)}{\sum_{z=0}^x \hat f_1(zS) \hat f_2^\le ((x-z)S)}\\ &= \frac{(1 \pm \delta_1)(1 \pm \delta_2)}{(\hat f_1 \star \hat f_2)^\le (xS)}\\
& = \frac{( 1\pm \delta_1 )^2 (1 \pm \delta_2)^2 }{( f_1 \star  f_2)^\le (xS)},
   \end{align*} 
   where the last equation holds by \Cref{prop:apx}. 
   As a consequence, the TV distance $D$ satisfies 
   \begin{align*}
       D \leq \max\{ (1 + \delta_1)^2(1 + \delta_2)^2 - 1, 1 - (1 - \delta_1)^2(1 - \delta_2)^2 \} \leq 4(\delta_1 + \delta_2).
   \end{align*}
   where the last inequality holds because we assumed $\delta_1,\delta_2 < 1/10$. Since Step~\ref{step-2} and Step~\ref{step-3} return samples with TV distance error $2\delta_2$ and $2\delta_1$, by a coupling argument, the TV distance between the output and the uniform distribution of $\{X \in \Omega_1 \times \Omega_2: w(X) \leq T\}$ is at most $6(\delta_1+\delta_2) \le 2\delta$ as desired.

   The success probabilities for the two input samplers are at least $1 - \delta_1$ and $1 - \delta_2$, so the success probability for the new sampler is also at least $1 - \delta_1 - \delta_2 - \delta/20 \geq 1 - \delta$, where $\delta/20$ is the probability that the sum-approximation convolution algorithm in \Cref{lem:sum-approx-conv} fails.
\end{proof}

We observe that our proofs for  rounding (\cref{lemma:roundingsampler}) and merging (\cref{lemma:mergingsampler}) basic Approximate Knapsack Samplers can be easily extended to work for the full version of Approximate Knapsack Samplers (\cref{defn:knapsacksamplertiny}). 
Specifically, when merging two samplers, we union their $I_0$ together to obtain a new set $I_0$; when rounding a sampler, for any $X \in \Omega$, the weight $w(X) = w(X \setminus I_0)$, we round the weight $ w(X \setminus I_0)$ to $\alpha(w(X \setminus I_0))$ and thus the new weight $w'(X) = w'(X \setminus I_0) = \alpha(w(X \setminus I_0))$. 
Intuitively, we can think all items in $I_0$ have the fixed weight 0. During the merging and rounding, we do the same processes for $\{X \setminus I_0:X \in \Omega\}$. Also, in the sampling process, we handle all items in $I_0$ separately. We repeat the sampling process for $\{X \setminus I_0:X \in \Omega\}$ and then sample every item in $I_0$ independently with $1/2$ probability.

We remark that the result of rounding/merging basic Approximate Knapsack Samplers is still a basic Approximate Knapsack Sampler.

\section{Main algorithm}\label{sec:algorithm}
In \cref{cref:secbottom}, we describe our sampler to be used at the bottom level of the divide and conquer, as well as a sampler that handles tiny items by allowing partial samples.
In \cref{subsec:mergesamplers}, we use divide and conquer to construct the sampler for every weight class, and merge them into a sampler for the whole input instance. 
In \cref{subsec:mainthm}, we draw (partial) samples using the sampler.
In \cref{sec:secondalg}, we describe the second-phase algorithm for the new (generalized) \#Knapsack instance defined by the tiny items and partial samples.

\subsection{Bottom-level samplers}
\label{cref:secbottom}
The following lemma will be used for the leaf nodes in our divide-and-conquer algorithms, where each leaf corresponds to one bin in the balls-into-bins hashing.  Here the parameter $B$ corresponds to the bin size, which will be $\polylog(n/\eps)$ in our applications. 
  
\begin{lemma}[Basic bottom-level sampler]\label{lem:leaf}
   Let $I$ be a set of items with weights $(W_i)_{i \in I}$, and $S,B\in \N^+, \delta>0$ be parameters. Define $\Omega = \{ X\subseteq I: |X|\le B\}$ and $W_{\max}= \max_{i\in I}W_i$.

   Then,  there is a basic Approximate Knapsack Sampler for $\Omega$ with parameters $S, L=B\cdot\lceil \frac{ W_{\max}}{S}\rceil$, $\sigma^2 = B\cdot S^2/4$, $\delta$, and either choice of the following two variants: 
  \begin{enumerate}
    \item  \label{item:leafdp}
   $\caT_c = \widetilde O(|I|\cdot B^3 \cdot \lceil \frac{ W_{\max}}{S}\rceil)$, and $\caT_q = \widetilde{O}(B^2)$. Or, 
   \item $\caT_c = \widetilde O(|I| +  \poly(B) \cdot \lceil \frac{ W_{\max}}{S}\rceil )$, and $\caT_q = \widetilde{O}(|I| + \poly(B)\cdot \lceil \frac{ W_{\max}}{S}\rceil)$.
       \label{item:leafvariant}
  \end{enumerate} 
  Here we use $\widetilde{O}(\cdot)$ to hide a $\polylog(|I|,\lceil \frac{ W_{\max}}{S}\rceil,1/\delta)$ factor.
\end{lemma}

For comparison, 
 variant \ref{item:leafvariant} achieves better construction time but worse query time.
 In terms of techniques, variant~\ref{item:leafdp} uses simple dynamic programming and binary search, whereas variant~\ref{item:leafvariant} uses FFT-based techniques with color coding.
For bounded-ratio input instances, variant~\ref{item:leafvariant} is not needed (see \cref{remark:boundedratio}).

\begin{proofof}{Proof of \cref{lem:leaf}}
At the beginning of the construction stage, we perform the following rounding step which is common to both variants:   
For each item $i \in I$, we round the weight $W_i$ independently to the weight $w(\{i\})$ such that
\begin{align}\label{eq:def-wi}
    w(\{i\}) &\defeq \begin{cases}
   \lceil\frac{W_i}{S} \rceil S  &\text{with prob. } \frac{W_i}{S} - \lfloor\frac{W_i}{S} \rfloor;\\
    \lfloor\frac{W_i}{S} \rfloor S  &\text{with prob. } 1 -(\frac{W_i}{S} - \lfloor\frac{W_i}{S} \rfloor).
    \end{cases}
\end{align}
By definition, it holds that $\Ex[w(\{i\})]=W_i$. 
For $X \in \Omega$, define $w(X) = \sum_{i \in X}w(\{i\})$.
Then, each $w(X)$ is the sum of at most $|X|\le B$ independent random variables $w(\{i\})$, each sampled from an interval of length at most $S$.  Hence, $w(X)$ is $B\cdot S^2/4$-subgaussian with mean $W_X$ as required. 

\begin{proofof}{Proof of variant \ref{item:leafdp}}
    For convenience, assume $I = \{1,2,\dots,|I|\}$ (by possibly renaming the items).
In the construction stage, we compute a table $M$ defined as
\[ M(k,x,z):= | \{ X \subseteq \{1,2,\dots,k\}: w(X) \le  xS, |X|=z \} |,  \]
where $0 \leq k \leq |I|, 0 \leq x \leq L := B\cdot\lceil \frac{ W_{\max}}{S}\rceil$, and $0 \leq z \leq B$, using dynamic programming with the update rule $M(k,x,z) = M(k-1,x,z) + M(k-1,x - \frac{w(\{k\})}{S},z-1)$ and initial values $M(0,x,0)=1$.
The table has $O(|I|LB)$ entries, each being a nonnegative integer at most $|I|^B$, which can be represented by $\widetilde{O}(B)$ bits in binary, so the dynamic programming takes total time $\widetilde O(|I|LB^2)$.

Then, note that $|\Omega| = \sum_{z=0}^B M(|I|,L,z)$, and the function $f(xS) = |\{X\in \Omega: w(X)=xS\}|$ required in \cref{def:basic} can be computed as $f(xS) = \sum_{z = 0}^B (M(|I|,x,z) - M(|I|,x-1,z))$.
Finally, we can round the function $f$ to obtain a function $\hat{f}:\{0,S,2S,\ldots,LS\} \to  \{0\}\cup \R_{\ge 1}$, where each $\hat{f}(xS)$ is a $(1 \pm \delta)$-approximation of $f(xS)$ and is represented by an $O(\log (\frac{1}{\delta}) + \log \log |\Omega|)$-bit floating-point number.
Overall, the construction stage takes time $\caT_c = \widetilde O(|I|LB^2) = \widetilde O(|I|\cdot B^3 \cdot \lceil \frac{ W_{\max}}{S}\rceil )$, and succeeds with probability $1 \geq 1 - \delta$.

Now we describe the query stage.
Given a query $T\in \N^+$, let $T \gets \min\{T,LS\}$ and $x \gets \lfloor \frac{T}{S} \rfloor$. Then our goal is to sample $X$ uniformly from $\{X\in \Omega: w(X) \le xS\}$. First, we decide the size $z= |X|$ by sampling $z\in \{0,\dots,B\}$ with probability $\propto  M(|I|,x,z)$, which can be done in time $\widetilde{O}(B^2)$ because $z\le B$ and each integer in the table $M$ has $\widetilde{O}(B)$ bits.
Then, we generate a random sample $X \in \Omega$ in the following iterative fashion which reports the items of $X$ in decreasing order of the indices:
We initialize $m\gets |I|, X \gets \emptyset$. In each iteration, we need to determine the largest item $j \in \{1,2,\dots,m\}$ that remains to be added to $X$.
Note that $j$ should follow the distribution $\Pr[j \leq a] = \frac{M(a,x,z)}{M(m,x,z)}$.
To sample $j$, we sample a uniformly random integer $1\le r\le M(m,x,z)$, and use binary search to find the smallest $j\leq m$ such that $M(j,x,z) \ge r$. Then, we update $X \gets X \cup \{j\}$, $m \gets j - 1, x \gets x - \frac{w(\{j\})}{S}$, and $z \gets z -1$ to sample the next item. Repeat this process until $z = 0$. 
In total there are $|X|\le B$ iterations, each with a binary search that takes $O(\log |I|)$ steps of comparing $\widetilde O(B)$-bit integers, so the total running time is $\widetilde{O}(B^2)$.
The rounded weight $w(X)$ can be computed in time $\widetilde{O}(B)$. Hence, $\caT_q = \widetilde{O}(B^2)$.
\end{proofof}

\begin{proofof}{Proof of variant \ref{item:leafvariant}}
Instead of dynamic programming, we need the following FFT-based lemma from Jin and Wu's Subset Sum algorithm \cite{JinW19}.
\begin{lemma}[{\cite[Lemma 5]{JinW19}}]
    \label{lem:jinwu}
For input integers $1\le a_1,\dots,a_n \le t$,
define $g_s:= |\{ X\subseteq [n]: \sum_{i\in X}a_i = s\}|$.
Then, for any given prime $p > t$,  we can deterministically compute $g_s\bmod p$ for all $0\le s\le t$ in $\widetilde O(n+ t\cdot \polylog(p))$ time.
\end{lemma}
We will compute a table $M$ defined as 
\[ M(x,z) = |\{X \subseteq I: w(X) = xS, |X|=z\}|,\]
where $0 \leq x \leq L := B\cdot\lceil \frac{ W_{\max}}{S}\rceil$, and $0 \leq z \leq B$. 
Each entry in $M$ is an integer with $O(B\log |I|)$ bits in binary.
Note that the function $f$ required in \cref{def:basic} can be expressed as $f(xS) = \sum_{z = 0}^B M(xS,z)$.

Now we explain how to compute $M$ using \cref{lem:jinwu}.
First, since $M(x,z) \le |I|^B$, we can pick the prime $p$ to be from $[2|I|^B,4|I|^B]$ so that $M(x,z)\bmod p = M(x,z)$. Such prime $p$ can be found by a Las Vegas algorithm with \cite{aks} in $\polylog(|I|^B) = \poly(B \log |I|)$ time, and each arithmetic operation over $\F_p$ takes $\widetilde O(\log p) = \widetilde O(B)$ time. 
For each rounded weight $w(\{i\})$, define the integer $a_i = w(\{i\})/S + (L+1)$. 
Then we apply \cref{lem:jinwu} to $\{a_i\}_{i\in I}$ with $t:= B (L+1) + L$ and obtain $g_s\bmod p$ for all $0\le s\le t$. Observe that for any $X\subseteq I$ with $|X|\le B$, it holds that $\sum_{i\in X}a_i = |X|\cdot (L+1) + \sum_{i\in X}w(\{i\})/S$. Conversely, for any $X\subseteq I$ with $\sum_{i\in X}a_i\le t$, we have $|X|\le \lfloor \frac{t}{L+1}\rfloor = \lfloor \frac{B(L+1)+L}{L+1} \rfloor = B$, and hence $\sum_{i\in X} w(\{i\}) /S\le |X|\cdot \lceil \frac{ W_{\max}}{S}\rceil \le L$. This means $(\sum_{i\in X}a_i)\bmod (L+1) = \sum_{i\in X}w(\{i\})/S$ and $\lfloor (\sum_{i\in X}a_i)/ (L+1)\rfloor = |X|$. Therefore, the table $M$ can be obtained by $M(x,z) = g_{x + z\cdot (L+1)}$.
The time complexity of \cref{lem:jinwu} is $\widetilde O(|I| + t\polylog p) = \widetilde O(|I| +\poly(B)L)$. Hence, $\caT_c =\widetilde O(|I| +\poly(B)L) =  \widetilde O(|I| +  \poly(B) \cdot \lceil \frac{ W_{\max}}{S}\rceil )$.

Now we describe the query stage.
Given a query $T\in \N^+$, let $T \gets \min\{T,LS\}$ and $x = \lfloor \frac{T}{S} \rfloor \le L$. 
Then our goal is to sample $X$ uniformly from $\{X\in \Omega: w(X) \le xS\}$.
First,  we determine the size $z= |X|$, by sampling $z\in \{0,1,\dots,B\}$ with probability proportional to  $C_z:= \sum_{0\le x'\le x}M(x',z)$, in $O(BL)\cdot \widetilde O(B) = \widetilde O(B^2 L)$ time. 
Then, we need to sample $X$ uniformly from 
\[\caX:= \{X\subseteq I: |X|=z, \sum_{i\in X}w(\{i\})/S \le x\},\]
which has known size $|\caX| = C_z$.
To do this, we use the following rejection sampling procedure: First sample a random partitioning $\caP = (I_1,I_2,\dots, I_{z'})$ where $z':= z^2$ and $I$ is the disjoint union  $I_1 \cup \dots \cup I_{z'}$, by assigning each $i\in I$ independently to a random one of the $z'$ groups.
Then, define\footnote{This is reminiscent of Bringmann's color-coding technique \cite{Bringmann17}.}
\[ \caX_{\caP}:= \Big \{X \in \caX: |X\cap I_k|\le 1\,\, \forall k \in [z']\Big \}.\]
We compute $|\caX_{\caP}|$, and uniformly sample $X \in \caX_{\caP}$ (if not empty). With probability $\frac{|\caX_{\caP}|}{|\caX|}$, we accept and return this sample $X$. Otherwise, we reject the sample $X$ and start over with a fresh new random partitioning $\caP$. The correctness and time complexity of this rejection sampling procedure follow from the following three claims: 
\begin{claim}
When accepted,  the returned sample $X$ is uniformly distributed over $\caX$.
\end{claim}
\begin{proof}
    Fix arbitrary $X_1,X_2\in \caX$. Because $|X_1|=|X_2|$,  over a random partitioning $\caP$ we have $\Pr_{\caP}[X_1 \in \caX_{\caP}] = \Pr_{\caP}[X_2 \in \caX_{\caP}]$ by symmetry. 
    In each step of our rejection sampling, the probability that $X_1$ is sampled and accepted is
    \begin{equation}
     \Ex_{\caP}\Big [\mathbf{1} [X_1 \in \caX_{\caP}]\cdot \frac{1}{|\caX_{\caP}|}\cdot \frac{|\caX_{\caP}|}{|\caX|}\Big ] = \frac{1}{|\caX|}\Pr_{\caP}[X_1\in \caX_\caP].
     \label{eqn:accprob}
    \end{equation}
    The same holds for $X_2$ as well. Hence, $X_1$ and $X_2$ have equal probability of being sampled and accepted. The claim follows since $X_1,X_2 \in \caX$ are arbitrary.
\end{proof}
\begin{claim}
Each step of the rejection sampling has acceptance probability is at least $1/2$.
\end{claim}
\begin{proof}
    Fix an arbitrary $X_1\in \caX$, where $|X_1|=z$. Over a random partitioning $\caP$ into $z'=z^2$ groups, the probability that some pair of $i,i'\in X_1, i\neq i'$ are assigned to the same group is at most $\binom{z}{2}\cdot \frac{1}{z^2}\le 1/2$ by a union bound. Hence, $\Pr_{\caP}[X_1\in \caX_{\caP}]\ge 1-1/2= 1/2$. Therefore, by \eqref{eqn:accprob},  the probability that a sample is accepted is $\sum_{X_1\in \caX}\frac{1}{|\caX|}\Pr_{\caP}[X\in \caX_{\caP}]\ge 1/2$. 
\end{proof}
\begin{claim}
Given $\caP = (I_1,I_2,\dots, I_{z'})$,
computing $|\caX_{\caP}|$ and sampling from $\caX_{\caP}$ can be done in time $\widetilde O(|I|+ \lceil \frac{W_{\max}}{S} \rceil \cdot \poly(B))$.
\end{claim}
\begin{proof}
 For each group $k\in [z'] = [z^2]$, define a bivariate polynomial $g_k(u,v) = 1 + \sum_{i\in I_k}v\cdot u^{w(\{i\})/S}$, which has $v$-degree 1 and $u$-degree $\le \lceil \frac{W_{\max}}{S} \rceil$. 
 We use FFT to multiply these polynomials together in increasing order of $k$,  and obtain the polynomials $g_1g_2\cdots g_{k}$ for all $k\in [z']$. 
 Let $[u^xv^z]f$ denote the coefficient of the $u^xv^z$-term in a polynomial $f$.
By definition of $\caX_{\caP}$, we can verify that $[u^x v^z](g_1g_2\cdots g_{z'}) = |\{X\in \caX_{\caP}:\sum_{i\in X}w\{(i)\}/S = x\}|$. Summing over all $x\in \{0,\dots, L\}$ gives $|\caX_{\caP}|$.

In order to sample a uniform $X\in \caX_{\caP}$, we first determine $x =
\sum_{i\in X}w\{(i)\}/S$ by sampling $x\in \{0,\dots,L\}$ with probability $\propto [u^x v^z](g_1\cdots g_{z'})$. 
Then, we iteratively determine $X\cap I_k$ in decreasing order of $k\gets z',\dots,1$, as follows: 
\begin{itemize}
    \item Sample $(x_0,z_0) \in \{0,\dots,x\}\times \{0,\dots, z\}$ with probability proportional to
        \begin{equation}
            \label{eqn:propo}
      \big ([u^{x_0}v^{z_0}](g_1\dots g_{k-1})\big )\cdot \big ([u^{x-x_0}v^{z-z_0}]g_k\big ).
        \end{equation}
        \item If $z-z_0=1$, then sample a uniformly random $i\in I_k$ satisfying $w(\{i\})/S = x-x_0$, and let $X\cap I_k = \{i\}$. Otherwise, we must have $z-z_0=x-x_0=0$, in which case we let $X\cap I_k = \emptyset$.
            \item Let $x\gets x_0, z \gets z_0$ and proceed to the next iteration $k-1$.
\end{itemize}
It is clear that this procedure generates a uniform sample $X\in \{X\in \caX_{\caP}:\sum_{i\in X}w\{(i)\}/S = x\}$ as desired.
It remains to analyze the total time complexity.

We first analyze the time complexity for performing all the polynomial multiplications.
Each polynomial $g_1g_2\cdots g_k$ has $v$-degree $k$ and $u$-degree $\le k  \lceil \frac{W_{\max}}{S} \rceil$,  and its coefficients are nonnegative integers bounded by $g_1(1,1)\cdots g_k(1,1) = (1+|I_1|)\cdots (1+|I_k|) \le  (1+|I|)^{B^2}$ which can be represented by $\widetilde O(B^2)$ bits in binary.
So the time complexity for multiplying $g_1\cdots g_{k}$ with $g_{k+1}$ using FFT is $\widetilde O(k\cdot k \lceil \frac{W_{\max}}{S} \rceil) \cdot \widetilde O(B^2) = \widetilde O(B^6\lceil \frac{W_{\max}}{S} \rceil)$, and  over all $k\in [z']$ the total time is $\widetilde O(B^8\lceil \frac{W_{\max}}{S} \rceil)$.
 At the beginning we also need to spend $\widetilde O(|I|)$ additional time to prepare all the polynomials $g_1,\dots,g_{z'}$.
 
 Now we analyze the time complexity for producing the sample $X$. In each of the $z' \le B^2$ iterations, we need to compute the expression in \eqref{eqn:propo} (which is an $\widetilde O(B^2)$-bit integer) for all  $(x_0,z_0) \in \{0,\dots,x\}\times \{0,\dots, z\} \subseteq \{0,\dots,L\}\times \{0,\dots, B\}$, so the total time is $\widetilde O(B^5 L) = \widetilde O(B^6 \lceil \frac{ W_{\max}}{S}\rceil)$. The step of  sampling $i\in I_k$ for a given value of $w(\{i\})/S$ can be implemented in $O(\log |I|)$ time after precomputing a standard data structure.

 To summarize, the total time complexity is $\widetilde O(|I|+ \lceil \frac{W_{\max}}{S} \rceil \cdot \poly(B))$.
\end{proof}
The probability that the rejection sampling process lasts $\ge \log_2(10/\delta)$ steps without acceptance is at most $\delta/10$. (If that happens, we can return an arbitrary answer, and the TV distance from the uniform distribution over $\caX$ is still $\le \delta$.) 

Therefore, the overall time complexity of the query stage can be bounded by $\caT_q = \widetilde{O}(|I| + \poly(B)\cdot \lceil \frac{ W_{\max}}{S}\rceil)$ (where the $O(\log(1/\delta))$ factor is hidden by the $\widetilde O(\cdot )$).
\end{proofof}

\end{proofof}

The following lemma gives an Approximate Knapsack Sampler (as per the full version in \cref{defn:knapsacksamplertiny}) in the case where every item has weight in $[0,S]$ (and hence rounded weight either $0$ or $S$).
\begin{lemma}[Bottom-level sampler for small items]
    \label{lem:bottomsamplertiny}
Let $S \subseteq \N^+$ be a parameter, and let $I$ be a set of items where $W_{\max}= \max_{i\in I}W_i \le S$. Define $\Omega = 2^I$. Then, for any $\delta > 0$, there is an Approximate Knapsack Sampler  for $\Omega$ with parameters $S, L= |I|, \sigma^2 = |I|\cdot S^2/4, \delta, \caT_c = \widetilde{O}(|I|) $ and $\caT_q =\widetilde{O}(|I|\cdot \frac{W_{\max}}{S}+1 )$, where $\widetilde{O}(\cdot)$ hides a
$\polylog(|I|,\lceil \frac{ W_{\max}}{S}\rceil,1/\delta)$ factor.
\end{lemma}
\begin{proof}
    In this case, we round $W_i$ to $w(\{i\})=0$ or $w(\{i\})=S$ so that $\Ex[w(\{i\})] = W_i$. Then we extend the definition of $w$ to $w(X) = \sum_{i\in X}w(\{i\})$ for all $X \in 2^I = \Omega$.  The tiny items (in \cref{defn:knapsacksamplertiny}) are $I_0 = \{i\in I: w(\{i\})=0\}$. Note that all random variables $w(\{i\})$ are $\frac{S^2}{4}$-subgaussian and independent.  Hence, for any $X \in \Omega$, $w(X)$ is $\sigma^2$-subgaussian. We can exactly compute the counting function $f(tS) = \{X \in \Omega: w(X) = tS\}$ because $f(tS) = 2^{|I_0|}\binom{|I \setminus I_0 |}{t}$.
    We can use an $O(\log |I| + \log(1/\delta))$-bit floating-point number to represent every $f(tS)$, which gives a $(1 \pm \delta)$-sum-approximation to $f$. The construction time is $\widetilde{O}(|I|)$.

    The probability that each $W_i$ is rounded to $S$ is at most $\frac{W_{\max}}{S}$. The expected size of $|I \setminus I_0|$ is at most $|I| \frac{W_{\max}}{S}$. With probability at least $1 - \delta$, the size of $|I \setminus I_0|$ is at most $O(|I|\cdot \frac{W_{\max}}{S} +\log(1/\delta) )$ by Chernoff bound and $\log(1/\delta)$ can be absorbed into $\widetilde{O}(\cdot)$. Given $T$, we only need to sample $\leq \lfloor T/S \rfloor$ element from $I \setminus I_0$ uniformly at random, which can be done in time $\widetilde{O}(|I \setminus I_0|)$.

    The Approximate Knapsack Sampler succeeds with probability at least $1 - \delta$, where the failure comes from the event that $|I \setminus I_0|$ is too large.
\end{proof}

\subsection{Merging the samplers}
\label{subsec:mergesamplers}
  Given the knapsack instance $W_1 \leq W_2 \leq \ldots \leq W_n$ with capacity $T$, we partition all the $n$ input items into $g + 1 = \lceil \log_2 n \rceil+1$ weight classes. 
  The $j$-th  class $(1\le j\le g)$ contains all items $\{i \mid \frac{T}{2^j} < W_i \leq \frac{T}{2^{j-1}}\}$. The last $(g+1)$-st class contains all small items $\{i\mid W_i \leq {T}/{2^g} \}$.
  For simplicity, we use \emph{class $m$} to refer to the weight class $I_m = \{i \mid T/m < W_i \leq 2T/m\}$, and use \emph{class $2n$} to refer to the last class $I_{2n}=\{i\mid W_i \leq {T}/{2^g} \}$ with all small items.

In this section, we use $\Omega$ to denote the set of knapsack solutions $\{X\subseteq [n] : W_X\le T\}$.
Let $\ell \in [2,8n)$ be the parameter obtained from \Cref{claim:existsell}.
We first construct the sampler for each weight class of input items:
\begin{lemma}\label{lem:AKS-group}
For each class $m \leq 2n$, there exists an Approximate Knapsack Sampler  for $\Omega(I_m)$  with parameters $S= \lceil\frac{T}{\ell \log^{50}(n/\eps)}\rceil$, $L=\widetilde{O}(n)$, $\sigma^2 = \frac{T^2}{\ell^2 \log^{50}(n/\eps)}$, $\delta = (\frac{\eps}{n})^{20}$ with  $\caT_c = \widetilde{O}(n^{1.5})$, $\caT_q = \widetilde{O}(\ell\sqrt{n})$ time, 
where $\Omega(I_m) \subseteq 2^{I_m}$ is a random collection of subsets such that for any $X \in \Omega$, 
\[\Pr[(X \cap I_m) \notin \Omega(I_m)] \leq \exp(-\log^{5}\frac{n}{\eps}).\] 
The $\widetilde{O}(\cdot)$ hides a $\polylog(n,1/\eps) \cdot \log(T)$ factor.
\end{lemma}
\begin{proofof}{Proof of \cref{lem:AKS-group}}
We assume $m$ is a power of 2. The assumption may not hold only if $m = 2n$. In this case, we round $m$ up to the nearest power of $2$ and we still use $I_m$ to denote the last class of small items. If $m = 2n$, then after the rounding, it holds that $2n \leq m < 4n$.
Hence, we always have $m = O(n)$.
To construct the Approximate Knapsack Sampler, we now split into two cases depending on the value of $\ell$ and $m$:\footnote{For bounded-ratio input instances, only Case~(a) is needed. See \cref{remark:boundedratio}.}

\paragraph*{Case (a): $m < 20\ell^2 \log^{100}(\frac{n}{\eps})$.}
We create $m$ bins $B_1,B_2,\ldots,B_m$.
For each item $i \in I_m$, we uniformly and independently hash $i$ into one of $m$ bins. By the definition of $I_m$, for any solution $X \in \Omega$, $|X \cap I_m| \leq m$. Let $B = \log^{10}\frac{n}{\eps}$. By balls-into-bins, for any $X \in \Omega$, any bin $B_i$, it holds 
\begin{align}\label{eq-BB}
 \Pr[|X \cap B_i| > B] \leq \exp\left(-\log^{8}\frac{n}{\eps}\right).  
\end{align}

We use \Cref{lem:leaf} with $\delta = 2^{-\log^{10}(n/\eps)}$ to construct an Approximate Knapsack Sampler for each $\Omega(B_i) = \{X \subseteq B_i: |X| \leq B\}$ with parameters $S = \lceil \frac{T}{\sqrt{m}\ell \log^{50}(n/ \eps)} \rceil$, $L = B \lceil \frac{\max_{i \in I_m}W_i}{S} \rceil =\widetilde{O}(\frac{\ell}{\sqrt{m}}+1) $, $\sigma^2 = B S^2/4$, $\delta = 2^{-\log^{10}(n/\eps)}$ with construction time and query time being either 
\begin{itemize}
    \item Variant~\ref{item:leafdp}:
$\caT_{c,i} = \widetilde{O}(|B_i|(\frac{\ell}{\sqrt{m}}+1))$ and $\caT_{q,i} = \widetilde{O}(1)$. Or,
\item Variant~\ref{item:leafvariant}:
$\caT_{c,i} = \widetilde{O}(|B_i| + \frac{\ell}{\sqrt{m}})$ and $\caT_{q,i} = \widetilde{O}(|B_i| + \poly(B)\cdot \lceil \frac{T/m }{S}\rceil) = \widetilde{O}(|B_i| + \frac{\ell}{\sqrt{m}})$. 
\end{itemize}

We construct a complete binary tree with $m$ leaves. 
Suppose all the leaves are in the level $H= \log_2 m$ and the root is in the level $0$. Define the scale parameter at each level by
\begin{align*}
    \forall\, 0 \leq h \leq H,\quad S_h = \left\lceil\frac{T}{2^{h/2} \ell \log^{50}(n/\eps)}\right\rceil.
\end{align*}
Initially, each leaf node has the Approximate Knapsack Sampler described in the last paragraph,
where parameter $S = S_H = \lceil\frac{T}{\sqrt{m} \ell \log^{50}(n/\eps)}\rceil$. 
We go through all nodes from the bottom to the top.
Suppose we are at a non-leaf node $u$ at level $h < \log_2 m$. 
Suppose it has two children $l$ and $r$ and each of them has an Approximate Knapsack Sampler.
We modify two samplers to obtain a new sampler at node $u$.
We first use \Cref{lemma:mergingsampler} to merge two samplers in $l$ and $r$. We then use \Cref{lemma:roundingsampler} to round the scale $S$ of the merged sampler from $S_{h+1}$ to $S_{h}$.

We claim that the Approximate Knapsack Sampler in the root satisfies the condition in the lemma. First note that $\Omega(I_m) = \Omega(B_1) \times \ldots \times \Omega(B_m)$. By~\eqref{eq-BB} and a union bound, we have for any $X \in \Omega$, $\Pr[(X \cap I_m) \notin \Omega(I_m)] \leq \exp(-\log^{5}\frac{n}{\eps})$.
The parameter $S = S_0 = \lceil\frac{T}{\ell \log^{50}(n/\eps)}\rceil$. %
Next, we inductively bound the parameter $L,\sigma,\delta$ and $\caT_q$.
 For any node $u$ at level $h$, denote parameters in Approximate Knapsack Sampler by $L_u = L$, $\sigma^2_u = \sigma^2$, $\delta_u = \delta$ and $\caT_q = \caT_{q,u}$. 
 Let $l$ and $r$ denote two children of $u$ at level $h+1$.
 By \Cref{lemma:mergingsampler} and \Cref{lemma:roundingsampler}, we have the following recursion
 \begin{align*}
     L_u &= \left\lceil \frac{(L_{l}+L_r)S_{h+1}}{S_{h}} \right\rceil \leq \frac{(L_l+L_r)}{\sqrt{2}} + 1,\\
     \sigma_u^2 &= \sigma^2_{l} + \sigma^2_r + \frac{S_h^2}{4},\\
     \delta_u &= 4\delta_{l}+ 4\delta_r,\\
     \caT_{q,u} &= \caT_{q,l} + \caT_{q,r}+\widetilde{O}(L_l + L_r),
 \end{align*}
 where the first inequality holds because of the lower bound of $T$ in~\eqref{eq:assume}. In the bottom level $H = \log_2 m$, we have the following bound for leaf node $u$,
 \begin{align}\label{eq:bottom}
     L_u =\widetilde{O}\left(\frac{\ell}{\sqrt{m}}+1\right),\, \sigma_u^2 \leq \frac{T^2}{ m \ell^2 \log^{80} \frac{n}{\eps}},\, \delta_u  = 2^{ - \log^{10}(n/\eps)}.%
 \end{align}
 Solving the recurrence gives 
 \begin{align}\label{eq-recur}
 \begin{split}
     L_\Root &= \widetilde{O}\left({\ell}+\sqrt{m}\right),\\
     \sigma^2_\Root &\leq m \cdot \frac{T^2}{ m \ell^2 \log^{80} \frac{n}{\eps}}  + \sum_{h=0}^{H-1} \frac{2^{h} S^2_h}{4} \leq \frac{T^2}{\ell^2 \log^{50}\frac{n}{\eps}},\\
     \delta_\Root &\leq m 2^{-\log^{10}(n/\eps)} + \sum_{h=0}^{H-1}8^h (\frac{\eps}{n})^{50} \leq (\frac{\eps}{n})^{20}, \\
     \caT_{q,\Root} &\leq \sum_{i=1}^m \caT_{q,i} + \sum_{h=0}^{H-1} \sum_{u \text{ in level }h}\widetilde{O}(L_u) = \sum_{i=1}^m \caT_{q,i} + \widetilde{O}( \ell\sqrt{m} + m)  
     \end{split}
 \end{align}
We explain how to obtain the bound for $L_{\text{root}}$, and other bounds  for $\sigma^2_{\text{root}},$ $\delta_{\text{root}}$, and $\caT_{q,\Root}$ can be obtained similarly. For each level $0 \leq h \leq H$, define $\mathcal{L}(h)$ as the sum of all $L_u$ for $u$ in level $h$. Using~\eqref{eq:bottom}, it holds that $\mathcal{L}(H) \leq \widetilde{O}(\ell\sqrt{m} + m)$. Using the above recursion, it holds that $\mathcal{L}(h) \leq \frac{\mathcal{L}(h+1)}{\sqrt{2}} + 2^h$, which is equivalent to $\frac{\mathcal{L}(h)}{2^h} + \frac{1}{\sqrt{2}-1} \leq \sqrt{2}(\frac{\mathcal{L}(h+1)}{2^{h+1}} + \frac{1}{\sqrt{2}-1})$. Thus, $L_{\text{root}} = \mathcal{L}(0) = \sqrt{2}^H \widetilde{O}(\frac{\ell}{\sqrt{m}}+1) = \widetilde{O}(\ell + \sqrt{m})$.

 Finally, the construction time $\caT_{c,\Root}$ is the sum of construction time for leaves and incremental time for non-leaves. 
For any node $u$ at level $h$, we compute an Approximate Knapsack Sampler at $u$ with $|\Omega| \leq n^{B2^{H-h}}$. 
 By \Cref{lemma:mergingsampler} and \Cref{lemma:roundingsampler}, we have
 \begin{align}\label{eq:proofcroot}
  \caT_{c,\Root} =  \sum_{i=1}^m \caT_{c,i} + \sum_{h = 0}^{H-1} 2^h \widetilde{O}( L_{h+1} \sqrt{2^{H-h}}) 
=  \sum_{i=1}^m \caT_{c,i} + \widetilde O(\ell \sqrt{m} + m),
 \end{align}
 where the last equation holds because $2^h \widetilde{O}( L_{h+1} \sqrt{2^{H-h}}) \leq \widetilde{O}(\mathcal{L}(h+1)\sqrt{2^{H-h}})$ and by the recursion above, any $\mathcal{L}(h)$ for $0 \leq h \leq H$ can be bounded by $\sqrt{2^{h-H}}\mathcal{L}(H) + \sqrt{2^{H + h}}$ and $\mathcal{L}(H) = \widetilde{O}(\ell \sqrt{m} + m)$.
 
 Now we further divide into two cases: 
 \begin{itemize}
   \item Case $\ell \le \sqrt{n}$: 
       
   We use variant~\ref{item:leafdp} with $\caT_{c,i} = \widetilde{O}(|B_i|(\frac{\ell}{\sqrt{m}}+1))$ and $\caT_{q} = \widetilde O(1)$.
  Then,  
 \begin{equation}
     \label{eqn:tcrootcase1}
  \caT_{c,\Root} =  \sum_{i=1}^m \widetilde{O}(|B_i|(\frac{\ell}{\sqrt{m}}+1)) + \widetilde O(\ell \sqrt{m} + m)=\widetilde O(\frac{n\ell}{\sqrt{m}} + n) + \widetilde O(\ell \sqrt{m} + m)= \widetilde{O}(n^{1.5}),
 \end{equation}
 where the last equation holds because we assume $\ell \leq \sqrt{n}$ and $m = O(n)$.
 And, we have \[\caT_{q,\Root} = \sum_{i=1}^m \widetilde O(1) + \widetilde O(\ell \sqrt{m} +m) =  \widetilde O(m + \ell \sqrt{m} )\le  \widetilde O(\ell\sqrt{m} )\le  \widetilde O(\ell\sqrt{n} ),\]
 where we used the Case~(a) assumption that $m\le \widetilde O(\ell^2)$.
 \item Case $\ell > \sqrt{n}$:
We use variant~\ref{item:leafvariant} with
$\caT_{c,i} = \widetilde{O}(|B_i| + \frac{\ell}{\sqrt{m}})$ and $\caT_{q,i} = \widetilde{O}(|B_i| + \frac{\ell}{\sqrt{m}})$. 
Then,
\begin{align*}
    \caT_{q,\Root} = \sum_{i=1}^m \widetilde{O}\left(|B_i| + \frac{\ell}{\sqrt{m}}\right) +  \widetilde O(\ell \sqrt{m}+m)= \widetilde{O}(n + \ell \sqrt{m}) \le \widetilde{O}(\ell \sqrt{n}), %
\end{align*}
where we used the assumption that $\ell>\sqrt{n}$.
And,
\begin{align*}
    \caT_{c,\Root} =  \sum_{i=1}^m \widetilde{O}(|B_i| + \frac{\ell}{\sqrt{m}}) + \widetilde O(\ell \sqrt{m} + m) = \widetilde{O}(n + \ell \sqrt{m} ) = \widetilde{O}(n^{1.5}).
   \end{align*} 
 \end{itemize}

\paragraph*{Case (b): $m \ge 20\ell^2 \log^{100}(\frac{n}{\eps})$.}
This means $\max_{i \in I_m}W_i \le \frac{4T}{m} \le   \lceil \frac{T}{\sqrt{m} \ell \log^{50}(n/\eps)}  \rceil = S$, where $S$ is the scale parameter for the bottom level.
Hence,    we can use \cref{lem:bottomsamplertiny} with parameter $S$ and $\delta = 2^{-\log^{10}(n/\eps)}$ to prepare a bottom-level sampler, with parameter $L=|I_m|, \sigma^2 = |I_m|\cdot S^2/4, \caT_{c} = \widetilde O(|I_m|), \caT_{q} =\widetilde O(|I_m|\cdot \frac{4T/m}{S}+1)$.

 We claim $|I_m| \le m/2$. This is because: 
   If $I_m$ is the last weight class of small items, then $2n\leq m <4n$ and $|I_m|\le n \leq m/2$.
   If $I_m$ is not the last weight class, then by definition, $W_i \in (\frac{T}{m},\frac{2T}{m}]$ for all $i \in I_m$.
   Since $m  \ge 20\ell^2 \log^{100}(\frac{n}{\eps})\ge 2\ell$, 
   from \eqref{eqn:smallgrouptotalub} we know
   $\sum_{i\in I_m} W_i \le \sum_{i\in [n]:W_i\le T/\ell}W_i < T/2$, so $|I_m| < \frac{T/2}{\min_{i\in I_m}W_i} < \frac{T/2}{T/m} = m/2$.

   Next, we use \Cref{lemma:roundingsampler} to round the bottom-level sampler so that the new $S$ parameter becomes $\lceil \frac{T}{\ell \log^{50}(n/\eps)} \rceil$.
   Combining~\Cref{lem:bottomsamplertiny} and \Cref{lemma:roundingsampler}, the final sampler has parameters $S = \lceil \frac{T}{\ell \log^{50}(n/\eps)} \rceil, L = O(|I_m|)=O(m)$, 
   \begin{align*}
      \sigma^2 = \frac{|I_m|}{4}\left\lceil \frac{T}{\sqrt{m}\ell \log^{50}(n/\eps)} \right\rceil^2 + \frac{1}{4}\left\lceil \frac{T}{\ell \log^{50}(n/\eps)} \right\rceil^2 \leq \frac{T^2}{\ell^2 \log^{50}(n/\eps)},
   \end{align*}
   $\delta = 2^{-\log^{10}(n/\eps)} \leq (\frac{\eps}{n})^{20}$, $\caT_c = \widetilde{O}(|I_m|)=\widetilde{O}(m)$ and $\caT_{q} = \widetilde{O}(|I_m| \frac{\ell}{\sqrt{m}} +1) =  \widetilde{O}(\sqrt{m}\ell)$.
   The lemma holds by noting that $m = O(n)$. 
  \end{proofof} 

  \begin{remark}
      \label{remark:boundedratio}
      In the bounded-ratio case where the only nonempty weight class of input items is $I_\ell =\{i: T/\ell < W_i \le 2T/\ell \}$, \cref{lem:AKS-group} is only applied to $m=\ell$. Then, in the proof above, we are in ``Case (a): $m < 20\ell^2 \log^{100}(\frac{n}{\eps})$'' and use the basic bottom-level sampler from \cref{lem:leaf} (instead of \cref{lem:bottomsamplertiny} which requires the full definition of Approximate Knapsack Sampler), so \cref{lem:AKS-group} finally returns a basic Approximate Knapsack Sampler as well. Moreover, inside Case~(a) we only need variant~\ref{item:leafdp} out of the two variants in \cref{lem:leaf}, because when $m=\ell$ we can bound $\caT_{c,\Root} \le \widetilde O(\ell\sqrt{n})$ in \eqref{eqn:tcrootcase1} even without the $\ell \le\sqrt{n}$ assumption.
  \end{remark}

Finally, we merge the samplers in \Cref{lem:AKS-group} across all the weight classes:
\begin{lemma}
    \label{lem:AKS-all}
  There exists an Approximate Knapsack Sampler  for $\hat{\Omega}$  with parameters $S= \lceil\frac{T}{\ell \log^{50}(n/\eps)}\rceil$, $L=\widetilde{O}(n)$, $\sigma^2 = \frac{T^2}{\ell^2 \log^{40}(n/\eps)}$, $\delta = (\frac{\eps}{n})^{10}$ with  $\caT_c = \widetilde{O}(n^{1.5})$, $\caT_q = \widetilde{O}(\ell\sqrt{n})$ time, 
where $\hat{\Omega} \subseteq 2^{[n]}$ is a random collection of subsets such that for any $X \in \Omega$, 
\begin{align}\label{eq:bound-union}
\Pr[X  \notin \hat{\Omega}] \leq \exp(-\log^{4}\frac{n}{\eps}).    
\end{align}
The $\widetilde{O}(\cdot)$ hides a $\polylog(n,1/\eps) \cdot \log(T)$ factor.
\end{lemma}
\begin{proof}
For each weight class $m \leq 2n$, we construct the Approximate Knapsack Sampler in \Cref{lem:AKS-group}. 
Note that all samplers have the parameters $S=\lceil\frac{T}{\ell \log^{50}(n/\eps)}\rceil, L=\widetilde{O}(n), \sigma^2= \frac{T^2}{\ell^2 \log^{50} ({n}/{\eps})}$, $\delta = (\frac{\eps}{n})^{20}$ and time $\caT_c=\widetilde{O}(n^{1.5})$, $\caT_{q}=\widetilde{O}(\ell \sqrt{n})$.
There are $O(\log n)$ samplers in total. 
We use \Cref{lemma:mergingsampler} to merge all these $O(\log n)$ samplers one by one into one sampler $\caS$ for
\begin{align}\label{eq:def-Omega-hat}
    \hat{\Omega} = \Omega(I_1) \times \Omega({I_2}) \times \Omega(I_4) \times \ldots \times \Omega(I_{2n})
\end{align}
with parameters  $S=\lceil\frac{T}{\ell \log^{50}(n/\eps)}\rceil, L=\widetilde{O}(n), \sigma^2= \frac{T^2}{\ell^2 \log^{40} ({n}/{\eps})}$, 
$\delta \leq \sum_{i=1}^{\lceil\log_2(2n)\rceil}2^i (\frac{\eps}{n})^{20}  \leq (\frac{\eps}{n})^{10}$ with time $\caT_q = \widetilde{O}(\ell \sqrt{n})$. 
The construction time is
\[\caT_c \leq O(\log n)\widetilde{O}(n^{1.5}) + \widetilde{O}(n \sqrt{n}) = \widetilde{O}(n^{1.5}),\]
where the first term is the time for constructing $O(\log n)$ samplers and the second term is the total running time contributed by \Cref{lemma:mergingsampler} (note that $|\Omega(I_m)| \leq 2^n$ for all weight classes $m$).
The probability in~\eqref{eq:bound-union} follows from \Cref{lem:AKS-group} and a union bound over $O(\log n)$ classes.
\end{proof}

\subsection{Main algorithm (proof of \texorpdfstring{\Cref{thm-main}}{})}
\label{subsec:mainthm}
We use \cref{lem:AKS-all} to construct our Approximate Knapsack Sampler $\caS$.
In the proof, we assume $\caS$ succeeds, which happens with high probability $1-\delta = 1 - o(1)$.
Recall that $S= \lceil\frac{T}{\ell \log^{50}(n/\eps)}\rceil$ is the parameter in \Cref{lem:AKS-all}.
Define 
\begin{align}\label{eq:deft}
    t = \left\lceil \frac{T + \frac{T}{\ell \log^{10}(n/\eps)}}{S} \right\rceil.
\end{align}

For the sake of analysis, we define the following random subset of items
\begin{align}\label{eq:def-Omega-proof}
    \Omega' = \{X \in \hat{\Omega}: w(X) \leq tS\},
\end{align}
where $\hat{\Omega}$ is the random set in~\eqref{eq:def-Omega-hat} and $w$ is the random weight function implicitly defined by the Approximate Knapsack Sampler $\caS$.
We have the following lemma for the set $\Omega'$. 
\begin{lemma}\label{claim:Omega'}
With probability $1 - o(1)$, it holds that
\begin{align}\label{q-1}
 |\Omega| \geq {|\Omega \cap \Omega'|} &\geq \left( 1-\frac{\eps^2}{n^2} \right) |\Omega|   
\end{align}
and
\begin{align}\label{eq:ratio-lower-bound}
    \frac{|\Omega \cap \Omega'|}{|\Omega'|} \geq \frac{\ell}{20000n(\log_2 n )^2}.
\end{align}
\end{lemma}
\begin{proof}
    Recall that $\Omega$ is the set of input knapsack solutions and $\Omega_d$ is defined in~\eqref{eq:def-Omega-d} for any integer $d \geq 1$. We claim that the random set $\Omega'$ satisfies the following properties:
\begin{enumerate}
    \item \label{PI} For any $X \in \Omega$, it holds that $\Pr[X \in \Omega'] > 1 - \exp(-\log^3 \frac{n}{\eps})$.
    \item \label{PII} For any $d \geq 2$, any $X \in \Omega_d$, it holds that $\Pr[X \in \Omega'] \leq \exp\left( -d \log^{10}\frac{n}{\eps} \right)$.
\end{enumerate}
Using~\eqref{eq:bound-union} in \Cref{lem:AKS-all},
\begin{align*}
    \forall X \in \Omega,\quad \Pr[X \notin \hat{\Omega}] \leq \exp(-\log^{4}\frac{n}{\eps}).
\end{align*}
Next, we fix the set $\hat{\Omega}$. 
For any $X \in \hat{\Omega}$, $w(X)$ is a $\sigma^2$-subgaussian random variable with mean $W_X$, where $\sigma^2 = \frac{T^2}{\ell^2 \log^{40}(n/\eps)}$. If $X \in \Omega$, then $W_X \leq T$. By the concentration of subgaussian,
\begin{align*}
\forall X \in \Omega \cap  \hat{\Omega},\quad   \Pr[w(X) > tS] \leq \Pr\left[w(X) > W_X + \frac{T}{\ell \log^{10} (n/\eps)}\right]\leq \exp\left( -\frac{1}{2} \log^{20} \frac{n}{\eps} \right).
\end{align*}
Since $\Omega' \subseteq \hat{\Omega}$, we have for any $X \in \Omega$, it holds that
\begin{align*}
    \Pr[X \notin \Omega'] &\leq \Pr[X \notin \hat{\Omega}] + \Pr[X \notin \Omega' \land X \in \hat{\Omega}] \leq \exp\left(-\log^{4}\frac{n}{\eps}\right) + \exp\left( -\frac{1}{2} \log^{20} \frac{n}{\eps} \right)\\
    &\leq \exp\left(-\log^{3}\frac{n}{\eps}\right).
\end{align*}
This proves Property \ref{PI}.

To prove  Property \ref{PII}, first note that we only need to consider $X \in \Omega_d \cap \hat{\Omega}$, otherwise the probability of $X \in \Omega'$ is 0. For any $X \in \Omega_d \cap \hat{\Omega}$, using the definition of $t$ in~\eqref{eq:deft},
\begin{align*}
\Pr[X \in \Omega'] \leq \Pr \left[ w(X) \leq tS \right]  &\leq  \Pr \left[ w(X) \leq T + \frac{T}{\ell \log^{10}(n/\eps)} +S \right]\\
&\leq \Pr \left[ w(X) \leq T + \frac{T}{\ell \log^9(n/\eps)} \right],
\end{align*}
where the last inequality uses the fact that $T$ is sufficiently large.
By the definition of $\Omega_d$, $W_X > T + (d-1)\frac{T}{\ell}$. By the concentration of subgaussian, we have
\begin{align*}
    \Pr[X \in \Omega'] &\leq \Pr \left[ w(X) \leq W_X - \left(d - 1 - \frac{1}{\log^{10}(n/\eps)}\right)\frac{T}{\ell} \right]\\
    &\leq \exp\left( -\frac{1}{2} \left(d - 1 - \frac{1}{\log^{10}(n/\eps)}\right)^2 \log^{40}\frac{n}{\eps}\right)\\
    &\leq \exp\left(-d \log^{10}\frac{n}{\eps}\right),
\end{align*}
where the last inequality holds when $n/\eps$ is larger than an absolute constant. 
This proves Property \ref{PII}.

Now, we are ready to prove the lemma.
By  Property \ref{PII} and \Cref{lem:omegadbound}, we have 
\begin{align*}
    \Ex\Big[ \Big|\Big( \bigcup_{d \geq 2} \Omega_d \Big) \cap \Omega' \Big| \Big] \leq \sum_{d \geq 2} n^d|\Omega|\exp\left( -d \log^{10}\frac{n}{\eps} \right) \leq |\Omega| \exp\left(-\log^8\frac{n}{\eps}\right).
\end{align*}
By Markov's inequality, we know that
\begin{align}\label{eq-g1}
     \Pr\Big[  \Big| \Big( \bigcup_{d \geq 2} \Omega_d \Big) \cap \Omega'  \Big| \leq \frac{\eps^2}{n^2}|\Omega| \Big] \geq 1 - \exp\left(-\log^5\frac{n}{\eps}\right).
\end{align}
By  Property \ref{PI}, we have 
$\Ex[|\Omega \setminus \Omega'|] \leq e^{-\log^3(n/\eps)}|\Omega|$.
By Markov's inequality, we have
\begin{align}\label{eq-g2}
    \Pr\left[ \frac{|\Omega \cap \Omega'|}{|\Omega|} \geq 1 - \frac{\eps^2}{n^2} \right] = 1 - \Pr\left[|\Omega \setminus \Omega'| \geq \frac{\eps^2}{n^2}|\Omega|\right] \geq 1 - \exp\left(-\log^{1.5}\frac{n}{\eps}\right).
\end{align}
Note that $\Omega$ and $\Omega_d$ for $d \geq 1$ form a partition of $2^{[n]}$. We can write $\Omega' = (\Omega \cap \Omega') \cup (\bigcup_{d \geq 1}(\Omega'\cap \Omega_d))$. 
Hence,
Combining~\eqref{eq-g1} and~\eqref{eq-g2}, with probability at least $1 - \exp\left(-\log^{1.1}\frac{n}{\eps}\right)-\exp\left(-\log^{1.5}\frac{n}{\eps}\right)$,
\begin{align}
|\Omega| &\geq {|\Omega \cap \Omega'|} \geq \left( 1-\frac{\eps^2}{n^2} \right) |\Omega|,\label{eq:first}\\
    \frac{|\Omega' \cap \Omega|}{|\Omega'|} &\geq \frac{|\Omega'\cap \Omega|}{|\Omega' \cap \Omega| + |\Omega' \cap \Omega_1| + \frac{\eps^2}{n^2}|\Omega|} \overset{\text{by}~\eqref{eq:first}}{\geq} \frac{(1-\eps^2/n^2)|\Omega|}{|\Omega| + | \Omega_1| + \frac{\eps^2}{n^2}|\Omega|}.\notag
\end{align}
Finally, by \Cref{lem:omega1bound}
    \[|\Omega_1| \le \frac{15000n(\log_2 n )^2}{\ell} |\Omega|.\]
Assume $n/\eps$ is larger than an absolute constant.
The above inequality implies that 
\begin{align*}%
    \frac{|\Omega \cap \Omega'|}{|\Omega'|} \geq \frac{\ell}{20000n(\log_2 n )^2}.
\end{align*}
This proves the lemma.
\end{proof}

Recall that $\Omega'$ is defined in  \eqref{eq:def-Omega-proof}.
By the definition of Approximate Knapsack Sampler $\caS$, since we assume it succeeds, we know that
\begin{align}\label{p-1variant}
 \sum_{j=0}^{\min(t,L)}\hat{f}(jS) \in \left(1 \pm 
\left(\frac{\eps}{n}\right)^{10} \right) |\Omega'|, 
\end{align}
and the sum $\sum_{j=0}^{\min(t,L)}\hat{f}(jS)$ can be computed in time $O(L) = \widetilde{O}(n)$ using the output of $\caS$.
Our algorithm next repeatedly queries the Approximate Knapsack Sampler $\caS$ with input $tS$ for 
\begin{align*}
    N = \frac{5000 n (\log n)^{10}}{\ell \eps^2} 
\end{align*}
times to generate $N$ independent (partial) samples $X_i \subseteq [n]\setminus I_0$,  where $I_0$ is the set of tiny items returned by $\caS$.
The sampler $\caS$ generates each sample in time $\widetilde{O}(\ell\sqrt{n})$, so the total running time for generating the partial samples is $N\cdot \widetilde{O} (\ell\sqrt{n}) = \widetilde{O}(n^{1.5}\eps^{-2})$. %
Let $U(2^{I_0})$ denote the uniform distribution over all subset of $I_0$.
Given each sample $X_i \subseteq [n]\setminus I_0$,  define 
\begin{align}\label{eq-def}
\forall i \in [N], \quad \tau_i :=  \Pr_{X_0 \sim U(2^{I_0})}[W_{X_0} + W_{X_i}\le T]  
\end{align}
which is the probability that a full random sample $X_0\cup X_i \subseteq [n] $ extended from the partial sample $X_i$ is a knapsack solution.
Note that $\tau_i$ is random because $X_i$ is random.
Note that in the case when $\caS$ is a basic Approximate Knapsack Sampler and hence $I_0 = \emptyset$ (which holds for bounded-ratio input instances, for example), $\tau_i \in \{0,1\}$ can be directly computed by checking whether the sample $X_i$ satisfies $\sum_{x \in X_i}W_x \leq T$ and the time cost of checking is at most the cost of generating samples, which is $\widetilde{O}(n^{1.5}/\eps^2)$.
However, for the general case where $I_0\neq \emptyset$, we do \emph{not} compute or estimate $\tau_i$ explicitly, because even the running time for sampling $X_0\subseteq I_0$ for every $\tau_i$   can be $|I_0|N = \widetilde{O}(\frac{n^2}{\ell^2 \eps^2})$ in total, which is too large when $\ell$ is small.

We next show how to use the sum in~\eqref{p-1variant} and $\tau_i$ (without explicitly computing $\tau_i$) to estimate the size of $\Omega$.
Let $U(\Omega')$ be the uniform distribution over $\Omega'$. 
If the sampler $\caS$ succeeds, then the samples $X_i$ returned by the $\caS$ are distributed so that
\begin{align}\label{p-2variant}
\forall i \in [N], \quad \Vert U(\Omega') - X_i\times U(2^{I_0}) \Vert_{\text{TV}} \leq 2\left( \frac{\eps}{n} \right)^{10},
\end{align}
where $X \times U(2^{I_0})$ denotes the random sample $Y = X \cup X'$, where $X' \sim U(2^{I_0})$.
For now, we assume the TV distance in \eqref{p-2variant} is 0, and we will remove this assumption at the end by a coupling argument. Under this assumption, we have
\[\Ex_{X_i}[\tau_i] = \frac{|\Omega \cap \Omega'|}{|\Omega'|}.\]
By a Chernoff bound,
with probability (over the randomness of $(X_i)_{i=1}^N$) at least $1 - o(1)$, we have 
\begin{align}\label{eq-lower-N}
 \frac{1}{N}\sum_{i=1}^N \tau_i \in \left(1 \pm \frac{\eps}{3}\right) \left(\frac{|\Omega \cap \Omega'|}{|\Omega'|}\right).   
\end{align}
Assume the above good event holds.
Hence, the remaining goal is to give an $(1\pm O(\eps))$-approximation of $\frac{1}{N}\sum_{i=1}^N\tau_i$, with the lower bound promise that $\frac{1}{N}\sum_{i=1}^N \tau_i \ge \frac{(1-\eps/3)\ell}{20000n(\log_2 n )^2} \ge \frac{1}{15000n(\log_2 n )^2}$ due to \eqref{eq:ratio-lower-bound}. 
Note that it suffices to get an hybrid approximation of 
$\frac{1}{N}\sum_{i=1}^N\tau_i$ with multiplicative approximation $(1\pm \eps/6)$ and additive error 
$\frac{\eps}{90000n(\log_2 n)^2}$: such a hybrid approximation $\tilde Z$ would satisfy $\tilde Z \in (1\pm \eps/6)\cdot \frac{1}{N}\sum_{i=1}^N\tau_i \pm \frac{\eps}{90000n(\log_2 n)^2} \in (1\pm \eps/3)\cdot \frac{1}{N}\sum_{i=1}^N\tau_i$ due to the lower bound promise.
We rewrite this quantity as
\begin{align}
     \frac{1}{N}\sum_{i=1}^N \tau_i &= \frac{1}{N}\frac{1}{2^{|I_0|}} \cdot |\{(X_0,X_i)\in 2^{I_0}\times \{X_1,X_2,\dots,X_N\}: W_{X_0}  + W_{X_i}\le T\}|\notag%
\end{align}
Hence, in order to get an $(1\pm \eps/3)$-approximation of this quantity, it suffices to compute $|\{(X_0,X_i)\in 2^{I_0}\times \{X_1,X_2,\dots,X_N\}: W_{X_0}  + W_{X_i}\le T\}|$ up to a hybrid approximation with multiplicative approximation $(1\pm \eps/6)$ and additive error $\frac{\eps}{90000n(\log_2 n)^2} \cdot N\cdot 2^{|I_0|}$.  %
This amounts to solving a second-phase instance for (some version of) \#Knapsack.

\begin{lemma}[Algorithm for the second-phase instance]\label{lem:sub}
There exists an algorithm such that given $I_0,X_1,X_2,\ldots,X_N$, it counts $|\{(X_0,X_i)\in 2^{I_0}\times \{X_1,X_2,\dots,X_N\}: W_{X_0}  + W_{X_i}\le T\}|$ up to a hybrid error with multiplicative approximation $(1\pm \eps/6)$ and additive error $\frac{\eps}{90000n(\log_2 n)^2} \cdot N\cdot 2^{|I_0|}$
 with probability at least $1 - o(1)$.
The running time of the algorithm is $\widetilde{O}( \frac{|I_0|^{1.5}+\max_{i \in [N]}|X_i|}{\eps^2} + \sum_{i \in [N]}|X_i| )$.
\end{lemma}
The proof of the above lemma is given in \Cref{sec:secondalg}. Note that $\sum_{i \in [N]}|X_i|$ is at most the total running time contributed by sampling $X_1,\ldots,X_N$, which is $\widetilde{O}(n^{1.5}/\eps^2)$. Also note that $|I_0| \leq n$ and $|X_i| \leq n$. Hence, the running time of \Cref{lem:sub} is $\widetilde{O}(n^{1.5}/\eps^2)$.

Note that all the good events mentioned above hold together with probability $1 - o(1)$.
The answer of \Cref{lem:sub} multiplied by $\frac{1}{N2^{|I_0|}}$ gives an estimate of $\frac{|\Omega \cap \Omega'|}{\Omega}$, and our algorithm outputs this estimate times the value in~\eqref{p-1variant} as our final estimate for $|\Omega|$. 
The total running time is $\widetilde{O}(n^{1.5}\eps^{-2})$. The overall approximation error is 
\begin{align*}
\left(1 \pm \left(\frac{\eps}{n}\right)^{10} \right)\left( 1\pm\frac{\eps^2}{n^2} \right)\left(1 \pm \frac{\eps}{3}\right)  \left(1 \pm \frac{\eps}{3}\right)  \in (1 \pm \eps), 
\end{align*}
where the error bounds come from~\eqref{q-1},~\eqref{p-1variant}, \eqref{eq-lower-N} and the error bound due to \Cref{lem:sub}.
All the above analysis ignores the error in~\eqref{p-2variant} due to our sampler only generating approximate samples with total variation distance at most $2(\frac{\eps}{n})^{10}$.
To deal with this error, we can couple $N$ approximate samples with $N$ perfect samples such that the coupling fails with probability at most $2N(\frac{\eps}{n})^{10}$.
(and we assume our algorithm fails if the coupling fails). Hence, with probability at least $1 - o(1)$, our algorithm achieves $(1 \pm \eps)$-approximation. 
Combining the coupling bound with the $1 - o(1)$ probability bounds for good events in~\eqref{p-1variant},\eqref{p-2variant},\eqref{eq-g1},\eqref{eq-g2}, it holds that with probability at least
$ 1 -o(1) - 2N\left( \frac{\eps}{n} \right)^{10} \geq \frac{3}{4}$ our algorithm does not fail.
To obtain an algorithm with a fixed running time, we can set a timer $T_{\text{tot}} = \widetilde{O}(n^{1.5}\eps^{-2})$. If the running time exceeds $T_{\text{tot}}$, we stop the algorithm and output $0$. Then, by Markov's inequality, with probability $1 - o(1) \geq \frac{2}{3}$, the algorithm outputs a correct answer $\in (1 \pm \eps)|\Omega|$ in fixed time $\widetilde{O}(n^{1.5}\eps^{-2})$.

\subsection{Second-phase approximation algorithm (proof of \texorpdfstring{\Cref{lem:sub}}{})}\label{sec:secondalg}
This section only focuses on the sub-problem in \Cref{lem:sub}.
For simplicity of notation, we use 
\[\Omega =\{(X_0,X_i)\in 2^{I_0}\times \{X_1,X_2,\dots,X_N\}: W_{X_0}  + W_{X_i}\le T\}\] to denote the set whose size we need to approximate. 
Since we pick exactly one element from $X_1,\ldots,X_n$,
we first do the following clean-up: 
\begin{itemize}
    \item For every $i \in [N]$, if $W_{X_i} > T$, then we can remove $i$ from $[N]$ because the contribution of $X_i$ to the number of solutions is $0$.
    \item For every $i\in [N]$, if $W_{X_i} + W_{I_0} \le T$, then the contribution of $X_i$ to the number of solutions is exactly $2^{I_0}$. We can add this contribution to the answer, and then also remove $i$ from $[N]$.
\end{itemize}
Hence, we can without loss of generality assume $T- W_{I_0}< W_{X_i}\leq T$ for all $i\in [N]$, without hurting the approximation that we want to achieve.

Recall $n$ is the number of items of the original knapsack instance and $\eps$ is the input error bound. By the setting of~\Cref{lem:sub}, it is clear that $\mathrm{poly}(n/\eps)$ is an upper bound on $|I_0|,|X_i|$ and $N$. We will use $n$ and $\eps$ to set the polylog factor parameters in the following analysis.

Sort the items of $I_0=\{i_1,\dots,i_{|I_0|}\}$ so that $W_{i_1}\ge W_{i_2} \ge \dots \ge W_{i_{|I_0|}}$.  For $0\le j\le |I_0|$, denote 
\[ I_0^{j}:= \{i_{j+1},i_{j+2},\dots, i_{|I_0|}\},\]
that is, $I_0$ with the heaviest $j$ items removed.
We define parameter 
\[r = \min\Big \{|I_0|, \lceil \log_2(\eps^{-1} \cdot 100000 n(\log_2 n)^2 )\rceil \Big \}. \]
Our plan is for all $1\le j\le r$, to approximate the size of 
\[ 
\Omega^{(j)}:= \{(X_0,X_i)\in 2^{I_0^{j}}\times \{X_1,X_2,\dots,X_N\}:   W_{X_0} + W_{i_j} + W_{X_i}\le T\},
\]
which allows us to additively approximate $|\Omega|$.
Intuitively, $\Omega^{(j)}$ denotes the set of all solutions such that $i_j$ is the heaviest item picked from $I_0$.
Formally, we have the following simple lemma:
\begin{lemma}
\label{lem:sumjrapprox}
   \[\Big \lvert  N + \sum_{j=1}^r |\Omega^{(j)}| - |\Omega|   \Big \rvert \le 
 N\cdot 2^{|I_0|} \cdot \frac{\eps}{100000 n(\log_2 n)^2}. \]
\end{lemma}
\begin{proof}
We classify all solutions $(X_0,X_i)\in \Omega$ by the heaviest item in $X_0$: If $X_0$ is empty, then the number of such solutions is exactly $N$. Otherwise, if the heaviest item is $i_j$ ($1\le j\le |I_0|$), then the number of such solutions exactly equals $|\Omega^{(j)}|$.  Note that the number of solutions for which the heaviest item $i_j$ exists but $j> r$ can be upper-bounded   as
\[N\cdot (2^{|I_0^{r}|} - 1)  = N\cdot (2^{|I_0|} 2^{-r}-1) \le N\cdot 2^{|I_0|} \cdot \frac{\eps}{100000 n(\log_2 n)^2},\]
by the definition of $r$. The claimed inequality then immediately follows from this classification.
\end{proof}

Now we fix a $1\le j\le r = O(\log(n/\eps))$, and focus on the task of additively approximately counting $\Omega^{(j)}$.
We again perform the similar clean-up as before, so that without loss of generality we can assume $W_{X_i} \in (T - W_{i_j} - W_{I_0^j} , T - W_{i_j}]$ for all $i\in [N]$.
 Define the lower bound parameter 
\begin{align}\label{eq:defWlow}
    W_{\text{low}} := T - W_{i_j}-W_{I{_0^j}}
\end{align}
so that $W_{X_i}\in (W_{\text{low}}, W_{\text{low}} + W_{I_0^j}]$ for every $i$.
Since we only consider the solution space that contains \emph{exactly one} item from $\{X_1,X_2,\ldots,X_N\}$, in the algorithm, we shift every weight $W_{X_i}$ to
\begin{align*}
    W'_{X_i} \defeq W_{X_i} -  W_{\text{low}} \in (0, W_{I_0^j}].
\end{align*}
We remark that the value of $W_{\text{low}}$ can be negative but the value of $W'_{X_i}$ is always positive.

Since our task of approximating $|\Omega^{(j)}|$ can be thought of as a slightly generalized instance of \#Knapsack with an $N$-choose-1, we use the same framework based on Approximate Knapsack Samplers to solve this instance. For this reason, we first abstract this $N$-choose-1 into the following approximate knapsack sampler lemma:
\begin{lemma}
    \label{lem:samplernchoose1}
Let $S \subseteq \N^+$ be a parameter, and we are given $X_1,\dots,X_N$ such that $W'_{X_i} \in (0, W_{I_0^j}]$ for all $i\in [N]$.
 Then, there is a basic Approximate Knapsack Sampler for $\{X_1,X_2,\ldots,X_N\}$ with parameters $S, L= \lceil W_{I_0^j}/S\rceil, \sigma^2 =  S^2/4, \delta=0, \caT_c = \widetilde{O}(L+\sum_{i\in N}|X_i|) $ and $\caT_q = \widetilde{O}( L + \max_{i\in N}|X_i|)$, where $\widetilde{O}$ hides a $\log(\max_{i}W'_{X_i})$ factor.
\end{lemma}

\begin{remark}
To fit the definition of the basic Approximate Knapsack Sampler in  \Cref{def:basic}, one can view each $X_i$ as an item and its weight is $W'_{X_i}$.  
The set $\Omega$ in \Cref{def:basic} is set as $\{X_1,X_2,\ldots,X_N\}$, which means we consider $N$-choose-1 knapsack.
\end{remark}

\begin{proofof}{Proof of \Cref{lem:samplernchoose1}}
For each $i \in [N]$, we round the weight $W_{X_i}$ independently such that
\begin{align*}
    w(X_i) = \begin{cases}
        S \lfloor \frac{W_{X_i}}{S} \rfloor &\text{with prob. }  \frac{W_{X_i}}{S} - \lfloor \frac{W_{X_i}}{S} \rfloor,\\
        S \lceil \frac{W_{X_i}}{S} \rceil &\text{with prob. } 1 -(\frac{W_{X_i}}{S} - \lfloor \frac{W_{X_i}}{S} \rfloor).
    \end{cases}
\end{align*}
Since each $w(X_i)$ has bounded range and we only pick exactly one $X_i$, $\sigma^2 = \frac{S^2}{4}$.
 Since $W'_{X_i} \in (0, W_{I_0^j}]$, we can set parameter $L = \lceil {W_{I_0^j}}/{S} \rceil$. 
 To compute the output function $\hat{f}$, we scan each $X_i$ and let $\hat{f}(w(X_i)) \gets \hat{f}(w(X_i)) + 1$. Note that $\hat{f}(tS)$ exactly counts the number of $X_i$ with $w(X_i)=tS$. We have $\delta = 0$.
 The construction time $\caT_c = O(L+\sum_{i\in N}|X_i|) $ because computing each $W_{X_i}$ takes time $|X_i|$. Given a query $T$, we can scan $\hat{f}$ in time $O(L)$ and sample one $X_i$ and output all items of $X_i$ in time $|X_i|$.
\end{proofof}

Now we use divide and conquer to construct a sampler for all items in $I_0^j$.
\begin{lemma}
    \label{lem:sampleri0}
There exists a basic Approximate Knapsack Sampler for $2^{I_0^j}$ with parameters $S = \Theta({W_{I_0^j}}/(|I_0^j|\log^{50}(n/\eps)))$, $L = \widetilde O(|I_0^j|)$, $\sigma^2 =  {W_{I_0^j}^2}/{(|I_0^j|^2 \log^{50}(n/\eps))},\delta= (\frac{\eps}{n})^{20}$ in time $\caT_c = \widetilde O(|I_0^j|^{1.5})$, $\caT_q = \widetilde O(|I_0^j|^{1.5})$, where $\widetilde{O}(\cdot)$ hides a $\polylog(n,1/\eps) \cdot \log(\max_{i \in I_0^j} W_{I_0^j})$ factor.
\end{lemma}
\begin{proof}
The proof is similar to the proof of \Cref{lem:AKS-group} but without balls-into-bins.
Let $m = |I_0^j|$.
We may assume $m$ is a power of $2$. If not, we round $m$ up to the nearest power of 2 and add some dummy leaf nodes, where every dummy node has a sampler for $\{\emptyset\}$.

Each item $i \in I_0^j$ corresponds to a leaf node at level $H=\log_2 m$. The leaf node has an Approximate Knapsack Sampler for $\{\emptyset, \{i\}\}$ with parameters $S_{\text{leaf}} = \lceil {W_{I_0^j}}/({m^{1.5}\log^{50}(n/\eps)}) \rceil, L_{\text{leaf}} = \widetilde{O}\left( m^{1.5}{W_i}/{W_{I_0^j}} + 1\right), \sigma^2_{\text{leaf}} ={S_{\text{leaf}}^2}/{4}, \delta_{\text{leaf}} = 0, \caT_{c,\text{leaf}} =\widetilde{O}(L_{\text{leaf}})$ and $\caT_{q,\text{leaf}} = O(1)$. To implement the sampler, we round $W_i$ to $w(\{i\})$ such that $w(\{i\}) = S\lfloor \frac{W_i}{S} \rfloor$ with probability $\frac{W_i}{S}-\lfloor \frac{W_i}{S}\rfloor$ and $w(\{i\}) = S \lceil \frac{W_i}{S} \rceil$ otherwise. Since the sampler is only for one element, 
\[L_{\text{leaf}} = \left\lceil \frac{W_i}{S_{\text{leaf}}} \right\rceil \leq \left\lceil \frac{W_i}{W_{I_0^j}} m^{1.5} \polylog\frac{n}{\eps} \right\rceil = \widetilde{O}\left( \frac{W_i}{W_{I_0^j}}m^{1.5} + 1\right).  \]
In the output function, we set $\hat{f}(w(\{i\})) = 1$ and other values as 0. It is easy to see the construction time for each sampler is $\caT_c = \widetilde{O}({W_i}m^{1.5}/{W_{I_0^j}} + 1)$ and query time $\caT_q = O(1)$.

Let $H= \log_2 m$.
For any $0 \leq h \leq H$, define $S_h = \lceil {W_{I_0^j}}/{(m2^{h/2} \log^{50}(n/\eps))} \rceil$.
For any non-leaf node at level $h$, we use \Cref{lemma:mergingsampler} to merge two samplers at its children and use \Cref{lemma:roundingsampler} to round the $S$ from $S_{h+1}$ to $S_h$. The final sampler is on the root of the tree. 

The analysis is similar to the analysis in \Cref{lem:AKS-group}.  Also note that we have a lower bound for $W_{I_0^j}$ due to~\eqref{eq:assume}.
We have already analyzed the parameters for leaf node $u$. 
For non-leaf nodes, we can use the same recursion in~\eqref{eq-recur}.
Solving the recursion, we have for the Approximate Knapsack Sampler at the root, it holds that $S = \Theta({W_{I_0^j}}/({m\log^{50}(n/\eps)})), L = \widetilde{O}(m)$, $\sigma^2= {W_{I_0^j}^2}/{(m^2 \log^{50}(n/\eps))}$, $\delta = (\frac{\eps}{n})^{20}$ and $\caT_q = \widetilde{O}(m^{1.5})$. The lemma holds by noting that $\frac{m}{2}<  |I_0^j| \leq m$.

Finally, we need to bound the construction time. 
For each node $u$ at level $h$, the values (in other words, the size of $\Omega$ for the approximate knapsack sampler at node $u$) of the output function are at most $2^{2^{H-h}}$. By \Cref{lemma:mergingsampler}, \Cref{lemma:roundingsampler}, and a similar analysis as that in~\eqref{eq:proofcroot}, we have
 \begin{align*}
  \caT_c =  \sum_{i\in I_0^j} \widetilde{O}\left(\frac{W_i}{W_{I_0^j}}m^{1.5} + 1\right) + \sum_{h=0}^{H-1}\sum_{u \in \caL(h+1)}\widetilde{O}(L_u \sqrt{2^{H-h}} ) = \widetilde{O}(m^{1.5})  = \widetilde{O}(|I_0^j|^{1.5}). %
 \end{align*}
 This proves the lemma.
\end{proof}

We use the same parameter $S = \Theta ({W_{I_0^j}}/{(|I_0^j|\log^{50}(n/\eps)}))$ to construct  samplers in \cref{lem:samplernchoose1} and \cref{lem:sampleri0}, where $\Theta(\cdot)$ hides a factor between $(1/2,1]$ because we use $m$ instead of $|I_0^j|$ in the proof of \cref{lem:sampleri0} to construct the approximate knapsack sampler.
By merging the samplers in \cref{lem:samplernchoose1} and \cref{lem:sampleri0} using \cref{lemma:mergingsampler}, we get:
\begin{lemma}
    \label{lem:secondstagemainsampler}
There exists a basic Approximate Knapsack sampler for $2^{I_0^j}\times \{X_1,\dots,X_N\}$ with parameters $S = \Theta ({W_{I_0^j}}/(|I_0^j|\log^{50}(n/\eps)))$, $L = \widetilde O(|I_0^j|)$, $\sigma^2 =  {W_{I_0^j}^2}/{(|I_0^j|^2 \log^{40}(n/\eps))}$, $\delta = (\frac{\eps}{n})^{10}$ in time $\caT_c = \widetilde O(|I_0^j|^{1.5} + \sum_{i\in [N]}|X_i|), \caT_q = \widetilde O(|I_0^j|^{1.5} + \max_{i\in [N]}|X_i|)$, where $\widetilde{O}(\cdot)$ hides a $\polylog(n,\frac{1}{\eps}) \cdot \log(T)$ factor.
\end{lemma}
\Cref{lem:secondstagemainsampler} is a simple corollary, so the proof is omitted.
Since the samplers in \cref{lem:samplernchoose1} and \cref{lem:sampleri0} are basic, the merged sampler is also basic.

The Approximate Knapsack Sampler in \Cref{lem:secondstagemainsampler} outputs a function $\hat{f}\colon S\cdot \{0,1,\ldots,L\} \to  \{0\}\cup \R_{\ge 1}$. 
Define the threshold parameter
\begin{align*}
    t = \left\lceil \frac{T-W_{i_j}-W_{\text{low}} + \frac{W_{I_0^j}}{|I_0^j| \log^{10}(n/\eps)}}{S} \right\rceil = \left\lceil \frac{W_{I_0^j} + \frac{W_{I_0^j}}{|I_0^j| \log^{10}(n/\eps)}}{S} \right\rceil,
\end{align*}
where the term $W_{i_j}$ comes from the fact that we only consider solutions that contain $i_j$ and the term $W_{\text{low}}$ comes from~\eqref{eq:defWlow}. The value of $t$ is always positive.
Next, we repeatedly query the Approximate Knapsack Sampler $\caS$ (\cref{lem:secondstagemainsampler}) with input $tS$ for 
\begin{align*}
    N' = 5000\eps^{-2} \log^{10}\left(\frac{n}{\eps}\right)
\end{align*}
times to generate $N'$ independent samples $Y_i$. 
Let $\caI_{Y_i} \in \{0,1\}$ indicate the event that $\sum_{x \in Y_i \cap I_0^j}W_x + W_{X_i} \leq T-W_{i_j}$, where $\{X_i\} = Y \cap \{X_1,X_2,\ldots,X_{N}\}$.
Finally, output
\begin{align*}
    \left(\frac{1}{N'}\sum_{i=1}^{N'} \caI_{Y_i}\right)\sum_{j=0}^{\min(t,L)}\hat{f}(jS).
\end{align*}

We now analyze the additive error between the output value and the exact answer $|\Omega^{(j)}|$.
Define the set 
\begin{align*}
    \hat{\Omega} = 2^{I_0^j} \times \{X_1,X_2,\ldots,X_{N}\}.
\end{align*}
Define the set $\Omega'$
\begin{align*}
    \Omega' = \{X \in \hat{\Omega}: w(X) \leq tS\}, 
\end{align*}
where $w(\cdot)$ is the random function defined by the sampler in \Cref{lem:secondstagemainsampler} (that is, the approximate knapsack sampler samples approximately from the uniform distribution over $\Omega'$).

For any $X \in \hat{\Omega}$, $w(X)$ is a $\sigma^2$-subgaussian random variable with mean $W_X - W_{\text{low}}$, where $\sigma^2 = {W_{I_0^j}^2}/{(|I_0^j|^2 \log^{40}(n/\eps))}$. If $X \in \Omega^{(j)}$, then $W_X \leq T - W_{i_j}$. By the concentration of subgaussian,
\begin{align*}
\forall X \in \Omega^{(j)} ,\quad   \Pr[w(X) > tS] \leq \Pr\left[w(X) > W_X - W_{\text{low}} + \frac{W_{I_0^j}}{|I_0^j| \log^{10} (n/\eps)}\right]\leq \exp\left( -\frac{1}{2} \log^{20} \frac{n}{\eps} \right).
\end{align*}
That is,  for any $X \in \Omega^{(j)}$, it holds that
\begin{align*}
    \Pr[X \notin \Omega'] & \leq \exp\left( -\frac{1}{2} \log^{20} \frac{n}{\eps} \right).
\end{align*}
This means
$\Ex[|\Omega^{(j)} \setminus \Omega'|] \leq e^{-0.5\log^{20}(n/\eps)}|\Omega^{(j)}|$.
By Markov's inequality, we have
\begin{align}\label{eq-g2var}
    \Pr\left[ \frac{|\Omega^{(j)} \cap \Omega'|}{|\Omega^{(j)}|} \geq 1 - \frac{\eps^2}{n^2} \right] = 1 - \Pr\left[|\Omega^{(j)} \setminus \Omega'| \geq \frac{\eps^2}{n^2}|\Omega^{(j)}|\right] \geq 1 - \exp\left(-\log^{10}\frac{n}{\eps}\right).
\end{align}

On the other hand, let \[\Omega_{\text{big}} = \left\{(X_0,X_i)\in 2^{I_0^j}\times \{X_1,X_2,\dots,X_N\}: W_{X_0}  + W_{X_i}\ge T -W_{i_j} + \frac{3W_{I_0^j}}{|I_0^j|\log^{10}(n/\eps)}\right\},\] and for any $X\in \Omega_{\text{big}}$,
\begin{align*}
\Pr[X \in \Omega'] \leq \Pr \left[ w(X) \leq tS \right]  &\leq  \Pr \left[ w(X) \leq T - W_{\text{low}} - W_{i_j} + \frac{W_{I_0^j}}{|I_0^j| \log^{10}(n/\eps)} +S \right]\\
& \le \Pr \left[ w(X) \leq W_X -W_{\text{low}} - \frac{2W_{I_0^j}}{|I_0^j|\log^{10}(n/\eps)}  +S \right]\\
\quad & \le \Pr \left[ w(X) \leq W_X - W_{\text{low}} -\frac{W_{I_0^j}}{|I_0^j|\log^{10}(n/\eps)} \right]\\
& \le \exp\left( -\frac{1}{2} \log^{20} \frac{n}{\eps} \right)
\end{align*}
 by the concentration of subgaussian.
Hence, we have
\begin{align*}
    \Ex\Big[ \Big| \Omega_{\text{big}} \cap \Omega' \Big| \Big] \leq |\Omega_{\text{big}}|\exp\left( -\frac{1}{2} \log^{20} \frac{n}{\eps} \right) \le |\hat \Omega|\cdot \exp\left( -\frac{1}{2} \log^{20} \frac{n}{\eps} \right).
\end{align*}
By Markov's inequality, we know that
\begin{align}\label{eq-g1variant}
     \Pr\Big[ \Big| \Omega_{\text{big}} \cap \Omega' \Big|  \leq \frac{\eps^2}{n^2}|\hat \Omega| \Big] \geq 1 - \exp\left(-\log^5\frac{n}{\eps}\right).
\end{align}
Combining~\eqref{eq-g1variant} and~\eqref{eq-g2var}, with probability at least $1 - \exp\left(-\log^{10}\frac{n}{\eps}\right)-\exp\left(-\log^{5}\frac{n}{\eps}\right)$,
\begin{align}\label{eq:bd-1}
\begin{split}
|\Omega^{(j)}| &\geq {|\Omega^{(j)} \cap \Omega'|} \geq \left( 1-\frac{\eps^2}{n^2} \right) |\Omega^{(j)}|,\\
    |\Omega'| &= |\Omega' \setminus \Omega_{\text{big}}| + |\Omega' \cap \Omega_{\text{big}}| \le |\Omega' \setminus \Omega_{\text{big}}| + \frac{\eps^2}{n^2}|\hat \Omega|. 
    \end{split}
\end{align}
Now we give an upper bound on $|\Omega' \setminus \Omega_{\text{big}}|$.  For any $(X_0,X_i) \in \Omega'\setminus \Omega_{\text{big}}$, we have 
\[ W_{X_0}+W_{X_i} < T-W_{i_j} + \frac{3W_{I_0^j}}{|I_0^j|\log^{10}(n/\eps)}. \]
By the way we sorted the items in $I_0$, we have $W_{i_j}\ge W_{i_{j'}}$ for all $i_{j'}\in I_0^j$, so $W_{i_j} \ge \frac{W_{I_0^j}}{|I_0^j|}$ by averaging. In particular, the above inequality implies 
\[ W_{X_0} + W_{X_i} < T.\]
Hence, we have 
\begin{align}\label{eq:bd-2}
   |\Omega'\setminus \Omega_{\text{big}}| \le | \{(X_0,X_i)\in 2^{I_0^j}\times \{X_1,X_2,\dots,X_N\}: W_{X_0}  + W_{X_i}< T\}|
    \le |\Omega|,
\end{align}
where $\Omega$ is the final set (defined at the beginning of \Cref{sec:secondalg}) we want to count in this subsection.
Therefore, we have 
    \[|\Omega'| \le {|\Omega| + \frac{\eps^2}{n^2}|\hat \Omega|}. \]

Recall $\delta = (\frac{\eps}{n})^{10}$ is the error of the Approximate Knapsack Sampler.
By Chernoff bound, we know our empirical mean (with $N'=\eps^{-2} 5000\log^{10}(n/\eps)$ samples) satisfies 
\[\left(\frac{1}{N'}\sum_{i=1}^{N'} \caI_{Y_i}\right) \in \frac{|\Omega' \cap \Omega^{(j)}|}{|\Omega'|} \pm \frac{\eps}{\log^{3}(n/\eps)} \pm \delta\]
with probability at least $1-\exp(\log^3(n/\eps))$. 
Let $ Z = \left(\frac{1}{N'}\sum_{i=1}^{N'} \caI_{Y_i}\right)\sum_{j=0}^{\min(t,L)}\hat{f}(jS)$ denote our output. 
This means our output value satisfies
\begin{align*}
 (1 - \delta)|\Omega'|\cdot \Big (\frac{|\Omega' \cap \Omega^{(j)}|}{|\Omega'|} - \frac{\eps}{\log^{3}(n/\eps)} - \delta\Big ) \leq  Z \leq (1+ \delta)|\Omega'|\cdot \Big (\frac{|\Omega' \cap \Omega^{(j)}|}{|\Omega'|} + \frac{\eps}{\log^{3}(n/\eps)} + \delta\Big ).
\end{align*}
Since $\delta = (\frac{\eps}{n})^{10} \ll \frac{\eps}{\log^3(n/\eps)}$, we can simplify it to
\begin{align*}
 (1 - \delta) |\Omega' \cap \Omega^{(j)}| - \left(\frac{2\eps}{\log^3(n/\eps)}\right)\cdot |\Omega'| \leq    Z \leq (1 + \delta) |\Omega' \cap \Omega^{(j)}| + \left(\frac{2\eps}{\log^3(n/\eps)}\right)\cdot |\Omega'|.
\end{align*}
By \eqref{eq:bd-1}  and \eqref{eq:bd-2}, we have the following bound
\begin{align*}
(1 - \delta) (1 - \frac{\eps^2}{n^2})|\Omega^{(j)}| - \frac{2\eps}{\log^3\frac{n}{\eps}} (|\Omega | + \frac{\eps^2}{n^2}|\hat \Omega|) \leq    Z \leq (1+ \delta) (1+ \frac{\eps^2}{n^2})|\Omega^{(j)}| + \frac{2\eps}{\log^3\frac{n}{\eps}} (|\Omega | + \frac{\eps^2}{n^2}|\hat \Omega|).
\end{align*}
By using the fact that $|\Omega^{(j)}| \leq |\hat{\Omega}| \leq 2^{|I_0|}N$ and $\delta \leq (\frac{\eps}{n})^{10}$, we have
\begin{align*}
    Z \in |\Omega^{(j)}| \pm \left(\frac{3\eps^2}{n^2}2^{|I_0|}N + \frac{2\eps}{\log^3(n/\eps)}|\Omega|\right).
\end{align*}

Now, if we sum up these output values obtained for all $1\le j\le r = O(\log(n/\eps))$, then by \cref{lem:sumjrapprox}, this sum (plus $N$) approximates $|\Omega|$ up to an additive error
\begin{align*}
& N\cdot 2^{|I_0|}\cdot \frac{\eps}{100000n(\log_2 n)^2} + \sum_{j=1}^{O(\log(n/\eps))} \left(\frac{3\eps^2}{n^2}2^{|I_0|}N + \frac{2\eps}{\log^3(n/\eps)}|\Omega|\right)\\
  & \le N\cdot 2^{|I_0|} \cdot \frac{\eps}{90000n(\log_2 n)^2} + O\left(\frac{\eps}{\log^2(n/\eps)}\right)\cdot |\Omega|,
\end{align*}
which satisfies the desired hybrid approximation claimed in \cref{lem:sub}.

The running time of the whole algorithm can be analyzed as follows. The total running time of the clean-up stage is at most $\widetilde{O}(\sum_{i=1}^N|X_i| + |I_0|)$. 
For each $1 \leq j \leq r$, where $r = \widetilde{O}(1)$, the construction of the sampler in \Cref{lem:secondstagemainsampler} costs $\caT_c = \widetilde O(|I_0^j|^{1.5} + \sum_{i\in [N]}|X_i|)$. The total running time of generating $N'$ samples is $ N'\caT_q = \widetilde O(N'(|I_0^j|^{1.5} + \max_{i\in [N]}|X_i|))$, where $N' = \widetilde{O}(\frac{1}{\eps^2})$. The total running time is dominated by
\begin{align*}
\widetilde{O}\left(\sum_{i=1}^N|X_i| + |I_0|\right) + r\caT_c + N'r \caT_q   = \widetilde{O}\left( \frac{|I_0|^{1.5}+\max_{i \in [N]}|X_i|}{\eps^2} + \sum_{i \in [N]}|X_i| \right)
\end{align*}
as desired.
\section{Sum-approximation convolution}\label{sec:sum-app}
In this section, we prove \Cref{lem:sum-approx-conv} in two steps:
First, in \Cref{sec:max+}, we solve the ``weighted witness counting'' version of $(\max,+)$-convolution for bounded monotone arrays. 
Then, in \Cref{sec:sum-app-conv}, we use it as a subroutine to obtain a $(1\pm \delta)$-sum-approximation convolution algorithm, proving \Cref{lem:sum-approx-conv}.

\subsection{Weighted witness counting for bounded monotone  \texorpdfstring{$(\max,+)$}{(max,+)}-convolution}\label{sec:max+}
The $(\max,+)$-convolution of two arrays $A[1\ldots n],B[1\ldots n]$ is defined as the array $C[2\ldots 2n]$ where \[C[k] := \max_{i+j=k}\{A[i]+B[j]\}.\] It can be computed by brute force in $O(n^2)$ time (whereas the best known algorithm is only faster by a $2^{\Omega(\sqrt{\log n})}$ factor \cite{Williamsapsp,bremner2014necklaces}). Significant speed-ups are known for structured input instances:
Chi, Duan, Xie, and Zhang \cite{ChiDX022} gave an $\widetilde O(n^{1.5})$-time randomized algorithm that computes the $(\max,+)$-convolution of two length-$n$ \emph{monotone} arrays (either both non-increasing or both non-decreasing) with integer entries bounded by $O(n)$, improving the previous $O(n^{1.859})$ time bound of \cite{ChanL15}. 
This result was recently generalized by Bringmann, D\"{u}rr, and Polak \cite[Theorem 7]{bringmann2024even} to monotone integer arrays bounded by $M$, achieving time complexity $\widetilde O(nM^{1/2})$.

We slightly extend the algorithm of \cite{bringmann2024even} to additionally compute the ``weighted  witness count'' for each entry of the $(\max,+)$-convolution:
\begin{theorem}
   Given two monotone arrays $A[1\ldots n],B[1\ldots n]$ with integer entries in $\{0, \ldots , M\}$, and two weight arrays $u[1\ldots n],v[1\ldots n]$ with integer entries in $\{0,\dots,U\}$, there is a randomized $\widetilde O(nM^{1/2}\cdot \polylog(U))$-time algorithm that computes the following with probability at least $1 - \frac{1}{\poly(n)}$:
   \begin{itemize}
    \item The $(\max,+)$-convolution $C[2\ldots 2n]$ of $A$ and $B$, and
    \item The weighted witness counts for all $2\le k \le 2n$, defined as 
    \[w[k] :=\sum_{i: A[i]+B[k-i]=C[k]} u[i]\cdot v[k-i].\]
\label{item:weighted}
   \end{itemize}
The $\widetilde{O}(\cdot)$ hides a $\polylog(n,M)$ factor.
\label{thm:weightedwitnesscnt}
\end{theorem}

Observe that $C[2\ldots 2n]$ must be a monotone array with entries from $\{0,\dots,2M\}$.

We remark that, in the unweighted case (i.e., when $u[i]=v[j]=1$ for all $i,j$), the witness counts are actually already computed by the algorithm of \cite{ChiDX022,bringmann2024even} as a by-product. However, the weighted case in \cref{thm:weightedwitnesscnt} does not immediately follow from their algorithm.\footnote{Very roughly speaking, the algorithm of \cite{ChiDX022,bringmann2024even} maintains a few candidates of  (possibly rounded values of) the entry $C[k]$ in the $(\max,+)$-convolution, computes the (unweighted) witness counts for the candidates, and then decides on the maximum candidate with a non-zero witness count.  When computing witness counts, their algorithm uses the fact that the number of pairs $(i,j)$ such that $i+j=k$ and $i\in [i_1,i_2]$ simply equals $i_2-i_1+1$ and can be computed in constant time.
However, the weighted count of such pairs $(i,j)$, namely $\sum_{i\in [i_1,i_2]}u[i]v[k-i]$, is not easy to  efficiently compute for given $k$ and $[i_1,i_2]$.} Nevertheless, we can prove \cref{thm:weightedwitnesscnt} by a white-box modification of their algorithm, as we describe next. (We will only describe the modifications on top of \cite{bringmann2024even}, instead of giving a fully self-contained proof.) %

We first state (possibly rephrased versions of) a few definitions and results from \cite[Appendix B]{bringmann2024even}.\footnote{\cite[Appendix B]{bringmann2024even} studied $(\min,+)$-convolution instead of $(\max,+)$-convolution, but they are equivalent after replacing $A[i]$ by $M-A[i]$ and replacing $B[j]$ by $M-B[j]$.} The algorithm picks a uniform random prime $p$ from $[M^{1/2},2M^{1/2}]$.\footnote{\cite[Appendix B]{bringmann2024even} picked $p$ from $[M^{\alpha},2M^{\alpha}]$, and eventually set the parameter $\alpha$ to $1/2$.}
We say $([i_1,i_2],k)$ is a \emph{segment} if $[i_1,i_2]$ is a maximal interval such that for every $i\in [i_1,i_2]$,  $A[i] = A[i_1]$ and $B[k-i]=B[k-i_1]$. Define the set of \emph{false positive segments} as    %
        \begin{equation}
            \label{eqn:t00}
    T_0^{(0)} := \{ \text{segment } ([i_1,i_2],k) \,\mid \,  A[i_1]+B[k-i_1] \equiv C[k]\hspace{-8pt}\pmod{p} \text{ and }   A[i_1]+  B[k-i_1] \neq  C[k] \}.
        \end{equation}
        Our definitions of segments and false positive segments are simplified from their original definitions in \cite[Appendix B]{bringmann2024even}; see \cref{remark:segment} below.

The algorithm of \cite[Appendix B]{bringmann2024even} not only computes $C[2\dots 2n]$, but also computes the set of false positive segments as some auxiliary data during the execution of the algorithm. Hence, we have the following lemma, which will be useful for proving \cref{thm:weightedwitnesscnt}:
  \begin{lemma}[implicit in {\cite[Appendix B]{bringmann2024even}}]
    \label{lem:falsepositivesegs}
  With at least $0.99$ success probability over the choice of prime $p\in [M^{1/2},2M^{1/2}]$, we can compute 
     $C[2\ldots 2n]$ and the set of false positive segments $T_0^{(0)}$ (defined in \eqref{eqn:t00}) in $\widetilde O(nM^{1/2})$ time (in particular, the output size $|T_0^{(0)}|\le \widetilde O(nM^{1/2})$).
   \end{lemma}

To extract \cref{lem:falsepositivesegs} from the algorithm description in \cite[Appendix B]{bringmann2024even}, there is one small detail remaining: the main algorithm in \cite[Appendix B]{bringmann2024even} only handles input instances satisfying a certain assumption (\cite[Assumption 44]{bringmann2024even}).\footnote{This assumption was introduced to avoid the technicalities arising from wrap-around modulo $p$. It assumes that every input entry is either a dummy $(-\infty)$ or an integer whose remainder modulo $p$ lies in the interval $[0,p/3]$, and that there are $O(M/p)$ contiguous segments of dummy entries. Only the non-dummy entries in the input arrays are required to be monotone. Dummy entries are not considered in the definition of segments. } To lift this assumption, they decompose the original $(\max,+)$-convolution instance into $9$ instances each satisfying the assumption.  These 9 instances are solved separately using their main algorithm, and finally their answers are combined by keeping the entrywise maximum as the final answer $C[k]$ in the $(\max,+)$-convolution (\cite[Lemma 45]{bringmann2024even}). 
In our case, we can lift this assumption in a similar way: apply their main algorithm to each of the 9 instances and obtain its set of false positive segments. Then we take the union of the false positive segment sets for these instances, and remove those segments $([i_1,i_2],k)$ where $A[i_1]+B[k-i_1]$ does not equal the final (entrywise maximum) value of $C[k]$. In this way we obtain the set of false positive segments for the original instance.

   \begin{remark}
       \label{remark:segment}
The original definitions of segments and $T_0^{(0)}$ in \cite{bringmann2024even} had extra parameters $\ell$ and $b$ which are not needed for our purpose. 
In the following we verify that our simpler definitions coincide with the $\ell=0$ and $b=0$ special case of the original definitions. 
       \begin{itemize}
        \item \cite[Definition 47]{bringmann2024even} originally defined segments $([i_1,i_2],k)_\ell$ with an extra parameter $\ell$. The conditions in their definition were $A^{(\ell)}[i]=A^{(\ell)}[i_1], B^{(\ell)}[k-i]=B^{(\ell)}[k-i_1], \tilde A[i]=\tilde A[i_1]$, and  $\tilde B[k-i]=\tilde B[k-i_1]$ for all $i\in [i_1,i_2]$, where 
 $\tilde A[i] := \left \lfloor \frac{A[i]}{p}\right \rfloor, A^{(\ell)}[i] := \lfloor \frac{A[i] \bmod p}{2^\ell}\rfloor$ (and similarly for $\tilde B[j],B^{(\ell)}[j]$) as defined in the beginning of \cite[Section B.1.2]{bringmann2024even}.
Here, we specialize to the $\ell=0$ case, and drop the $\ell$ subscript. Since any integer is uniquely determined by its quotient and remainder modulo $p$,  the conditions above are equivalent to  $A[i] = A[i_1]$ and $B[k-i]=B[k-i_1]$.
\item The original definition of $T_b^{(\ell)}$ appeared in \cite[Section B.1.2]{bringmann2024even}. For $b=\ell=0$, their definition  was $\{ \text{segment } ([i_1,i_2],k) \,\mid \,  A^{(0)}[i_1]+B^{(0)}[k-i_1] = C^{(0)}[k] \text{ and }  \tilde A[i_1]+ \tilde B[k-i_1] \neq \tilde C[k] \}$, where $\tilde A[i] := \left \lfloor \frac{A[i]}{p}\right \rfloor, A^{(0)}[i] := A[i] \bmod p$ (and similarly for $\tilde B[j],B^{(0)}[j],\tilde C[k],C^{(0)}[k]$). 
Their definition was only relevant under the assumption (\cite[Assumption 44]{bringmann2024even}) that $(A[i_1]\bmod p) \le p/3$ and $(B[k-i_1]\bmod p) \le p/3$. In this case, we verify that their definition coincides with our definition in \eqref{eqn:t00}: since $(A[i_1]+B[k-i_1])\bmod p = (A[i_1]\bmod p) + (B[k-i_1]\bmod p) = A^{(0)}[i_1] + B^{(0)}[k-i_1]$, we have $A^{(0)}[i_1] + B^{(0)}[k-i_1]= C^{(0)}$ if and only if $A[i_1]+B[k-i_1] \equiv C[k] \pmod{p}$. Then, assuming $A^{(0)}[i_1] + B^{(0)}[k-i_1]= C^{(0)}$ holds, we have $  C[k] -  A[i_1]- B[k-i_1] = p(\tilde C[k] - \tilde A[i_1]-\tilde B[k-i_1])  + C^{(0)}[k] - A^{(0)}[i_1]-B^{(0)}[k-i_1] = p(\tilde C[k] - \tilde A[i_1]-\tilde B[k-i_1])$, so $\tilde A[i_1]+ \tilde B[k-i_1] \neq \tilde C[k]$ if and only if $ A[i_1]+  B[k-i_1] \neq  C[k]$.
       \end{itemize}
   \end{remark}

Now we prove \cref{thm:weightedwitnesscnt} using \cref{lem:falsepositivesegs}.
\begin{proofof}{Proof of \cref{thm:weightedwitnesscnt}}
    Pick a uniform random prime $p\in [M^{1/2},2M^{1/2}]$.
     Define polynomials $f(x,y) = \sum_{i=1}^n u[i]x^i y^{A[i]\bmod p}$ and $g(x,y) = \sum_{j=1}^n v[j]x^j y^{B[j]\bmod p}$ which have $x$-degree at most $n$ and $y$-degree at most $p-1$, and compute their product  using FFT in $\widetilde O(np\polylog U)$ time. Then, for every $k\in [2, 2n]$, we obtain  
     \[ w'[k]:= \sum_{i:\, A[i]+B[k-i]\equiv C[k]\hspace{-7pt}\pmod{p}} u[i]v[k-i]\]
     by summing the $x^ky^{C[k]\bmod p}$ term's and the $x^ky^{(C[k]\bmod p)+p}$ term's coefficients in the product $f(x,y) g(x,y)$.  To compute the desired weighted witness count $w[k] = \sum_{i:A[i]+B[k-i]=C[k]}u[i]\cdot v[k-i]$, it remains to subtract from $w'[k]$ the contribution made by the false positive witnesses $i$ where $A[i]+B[k-i]\equiv C[k] \pmod{p}$ and $A[i]+B[k-i]\neq C[k]$. 
     
     Since $A[1\ldots n]$ is monotone with entries from $\{0,\dots,M\}$,  we can partition $[1,n]$ into a collection $\caI$ of at most $\min\{2M,n\}$ disjoint intervals, such that each interval $I\in \caI$ has length $|I|\le O(\frac{n}{M}+1)$ and has the same value of $A[i]$ for all $i\in I$.\footnote{In this section, an \emph{interval $[a,b]$} refers to the set of integers $\{a,a+1,\dots,b\}$, and we say the \emph{length} of this interval is $b-a+1$.} 
     Similarly, let $\caJ$ be a collection of at most $\min\{2M,n\}$ intervals that partition $[1,n]$ so that each $J\in\caJ$ has the same value of $B[j]$ for all $j\in J$, and $|J|\le O(\frac{n}{M}+1)$.  
     For an interval $I\in \caI$, we use the shorthand $A[I]:= A[i]$ where $i\in I$. For $a\in \{0,\dots,M\}$, let $\caI^{(a)}:= \{ I\in \caI: A[I]=a\}$. Similarly define $B[J]$ and $\caJ^{(b)}$.
Denote the sum of two intervals $I= [i_1,i_2],J= [j_1, j_2]$ as $I+J = [i_1+j_1,i_2+j_2]$.

We say a pair of intervals $(I,J)\in \caI\times \caJ$ is \emph{relevant} for an index $k \in [2,2n]$, if $k\in I+J$, $A[I]+B[J] \equiv C[k] \pmod{p} $, and $A[I]+B[J]\neq C[k]$.
     We say $(I,J)\in \caI\times\caJ$ is \emph{relevant} if there exists $k\in I+J$ such that $(I,J)$ is relevant for $k$.  
     For every relevant interval pair $(I,J)\in \caI\times \caJ$, define an  array $\hat w^{(I,J)}$ of length $|I+J|$ by
     \begin{equation}
     \hat w^{(I,J)} [k] := \sum_{i\in I,j\in J, i+j=k} u[i]\cdot v[j], \text{ for all }k\in I+J.
     \label{eqn:hatwij}
     \end{equation}
Then,  the correct weighted witness count for every $2\le k\le 2n$ is given by
\begin{equation}
    \label{eqn:wk}
  w[k] =  w'[k] - \sum_{(I,J)\text{ relevant for }k}\hat w^{(I, J)}[k].
\end{equation}

In order to use \eqref{eqn:wk} to compute $w[k]$ for all $k\in [2,2n]$, we need to first find all the relevant interval pairs $(I,J)\in \caI\times \caJ$, and then compute the array $\hat w^{(I,J)}$ in \eqref{eqn:hatwij} for each relevant pair $(I,J)$, by convolving the two arrays $\{u[i]\}_{i\in I}, \{v[j]\}_{j\in J}$ using FFT in $\widetilde O((|I|+|J|)\polylog U)$ time.

Now we describe how to find all the pairs $(I,J)\in \caI\times \caJ$ that are relevant for any given $k\in [2,2n]$.
   At the very beginning, apply \cref{lem:falsepositivesegs} to compute all the false positive segments $T_0^{(0)}$. 
   Given $k \in [2,2n]$, we go through every false positive segment $([i_1,i_2],k)\in T_0^{(0)}$,  find all interval pairs $(I,J) \in \caI^{(A[i_1])}\times \caJ^{(B[k-i_1])}$ such that $k \in I+J$, and mark these pairs as relevant for $k$.
    The correctness of this procedure immediately follows from the definition of false positive segments in \eqref{eqn:t00} and the definition of relevant interval pairs. 
    To efficiently implement this procedure, we use the following claim (with $\caI'\gets \caI^{(A[i_1])}$ and $\caJ'\gets \caJ^{(B[k-i_1])}$) which is based on standard binary search and two pointers:
     \begin{claim}
         \label{claim:twopointer}
        Let $\caI'=\{[i_0+1, i_1], [i_1+1,i_2],\dots,[i_{|\caI'|-1}+1, i_{|\caI'|}]\}$ and $\caJ'=\{[j_0+1, j_1], [j_1+1,j_2],\dots,[j_{|\caJ'|-1}+1, j_{|\caJ'|}]\}$. Given integer $k$, let $\caP_k:=\{(I,J)\in \caI'\times \caJ': k\in I+J\}$. Then $|\caP_k|\le |\caI'|+|\caJ'|-1$, and $\caP_k$ can be reported in $O(\log (|\caI'|+|\caJ'|)) + O(|\caP_k|)$ time (assuming the arrays $(i_0,\dots,i_{|\caI'|})$ and $(j_0,\dots,j_{|\caJ'|})$ are already stored in memory).
     \end{claim}
     \begin{proofof}{Proof of \cref{claim:twopointer}}
         Assume $k\in [i_0+1+j_0+1, i_{|\caI'|}+j_{|\caJ'|}]$ (otherwise $\caP_k=\emptyset$). Assume $i_1+ j_{|\caJ'|}\ge k$ (otherwise, use binary search to find the largest $1\le p< |\caI'|$ such that $i_p + j_{|\caJ'|}<k$ and remove the useless intervals $[i_0+1,i_1],\dots,[i_{p-1}+1,i_{p}]$). Assume $(i_0+1)+(j_{|\caJ'|-1}+1)\le k$ (otherwise, use binary search to find the smallest $1<q \le  |\caJ'|$ such that $(i_0+1)+(j_{q-1}+1)>k$, and remove the useless intervals $[j_{q-1}+1,j_{q}],\dots,[j_{|\caJ'|-1}+1, j_{|\caJ'|}]$).  Then the following pseudocode clearly reports all $(I,J)\in (\caI'\times \caJ')$ such that $k\in I+J$, with time complexity linear in the output size:
         \begin{itemize}
            \item Initialize two pointers $p\gets 1,q\gets |\caJ'|$.
                \item While $q\ge 1$ and $p\le |\caI'|$ and $k\in [i_{p-1}+1,i_p] + [j_{q-1}+1,j_q]$:
                \item \begin{itemize}
                    \item Report $([i_{p-1}+1,i_p] , [j_{q-1}+1,j_q])$.
                    \item If $q > 1$ and $k\in [i_{p-1}+1,i_p] + [j_{q-2}+1,j_{q-1}]$, then $q\gets q-1$;
                        \item Else If $p<|\caI'|$ and $k\in [i_{p}+1,i_{p+1}] + [j_{q-1}+1,j_{q}]$, then  $p\gets p+1$;
                        \item Else:  $p\gets p+1, q\gets q-1$.
                \end{itemize}
         \end{itemize}
         Each output interval pair can be charged to the decrement of $q-p$, which ranges from $|\caJ'|-1$ to $1-|\caI'|$ before the while loop terminates, so the output size is at most $|\caI'|+|\caJ'|-1$.
     \end{proofof}
     We apply the procedure described above to all $k\in [2,2n]$. In this way we obtain all the relevant interval pairs $(I,J) \in \caI\times \caJ$ (a relevant interval pair $(I,J)$ may be reported multiple times, once for every $k\in I+J$).

We now analyze the overall time complexity.
\cref{lem:falsepositivesegs} runs in $\widetilde O(nM^{1/2})$ time and outputs $T_0^{(0)}$ of size $\widetilde O(nM^{1/2})$ (with $0.99$ success probability over the random choice of prime $p\in [M^{1/2},2M^{1/2}]$). The initial FFT for  computing $w'[2\ldots 2n]$ takes $\widetilde O(np\polylog U) = \widetilde O(nM^{1/2}\polylog U)$ time. Then we use \cref{claim:twopointer} to find all relevant interval pairs $(I,J)\in \caI\times \caJ$, spending $O(\log n)$ time for each false positive segment in $T_0^{(0)}$, and $O(1)$ time for each tuple in $\{(I,J,k): (I,J)\in \caI\times \caJ \text{ is relevant for } k\}$ (and hence $O(|I+J|)$ time for each relevant interval pair $(I,J)$). Then, for each relevant interval pair $(I,J)$ we use FFT to compute $\hat w^{(I,J)}$ in $\widetilde O(|I+J|\polylog U)\le \widetilde O((\frac{n}{M}+1)\polylog U)$ time.
Therefore, the total time complexity is 
\begin{equation}
    \label{eqn:totaltime}
 \widetilde O(nM^{1/2}\polylog U) + \widetilde O((\tfrac{n}{M}+1)\polylog U)\cdot [\#\text{relevant pairs } (I,J)\in \caI\times \caJ],
\end{equation}

To bound \eqref{eqn:totaltime}, it remains to analyze the expected number of relevant pairs $(I,J)\in \caI\times \caJ$, over the random prime $p\in [M^{1/2},2M^{1/2}]$.
 Recall that $(I,J)$ is relevant if and only if there exists $k\in I+J$ such that $C[k]-(A[I]+B[J])$ is a non-zero multiple of $p$. By a union bound over the \emph{distinct} values of $C[k]-(A[I]+B[J])$ where $k\in I+J$, we know the probability that this happens is  $\Pr_p[(I,J)\text{ relevant}]  \le  O (\frac{\log M}{M^{1/2}} ) \cdot |\{C[k]-(A[I]+B[J]): k\in I+J\}| =O (\frac{\log M}{M^{1/2}} ) \cdot |\{C[k]: k\in I+J\}|$. Then, by linearity of expectation,
   \begin{align*}
       \Ex_{p}[\#\text{relevant pairs } (I,J)\in \caI\times \caJ]  &\le \sum_{I\in \caI,J\in \caJ} O\big (\frac{\log M}{M^{1/2}}\big )\cdot |\{C[k]: k\in I+J\}|\\
       & \le O\big (\frac{\log M}{M^{1/2}}\big )\cdot \sum_{I\in \caI,J\in \caJ}  \big (1 + \sum_{k\in I+J}\mathbf{1}\big [C[k]\neq C[k-1]\big ]\big ).
 \end{align*}
 To bound this expectation, note that
 \begin{align*}
         \sum_{I\in \caI,J\in \caJ} \sum_{k\in I+J}\mathbf{1}\big [C[k]\neq C[k-1]\big ]
      &=  \sum_{k=2}^{2n}\mathbf{1}\big [C[k]\neq C[k-1]\big ] \cdot |\{(I,J)\in \caI\times \caJ: k \in I+J\}|\\
      &\le   \sum_{k=2}^{2n}\mathbf{1}\big [C[k]\neq C[k-1]\big ] \cdot (|\caI|+|\caJ|-1) \tag{by \cref{claim:twopointer}}\\
      &\le   \min\{2n-1, 2M\} \cdot (|\caI|+|\caJ|-1). \tag{since $C[2\ldots 2n]$ is monotone with entries in $\{0,\dots,2M\}$}
 \end{align*}
Hence,
 \begin{align*}
            \Ex_{p}[\#\text{relevant pairs } (I,J)\in \caI\times \caJ]        &\le O\big (\frac{\log M}{M^{1/2}}\big ) \cdot \big ( |\caI||\caJ| + \min\{2n-1, 2M\} \cdot (|\caI|+|\caJ|-1)\big) \\ &\le O\big (\frac{\log M}{M^{1/2}}\big )\cdot (\min\{n,M\})^2. \tag{by $|\caI|,|\caJ|\le O(\min\{n,M\})$} 
 \end{align*}
 Plugging this into \eqref{eqn:totaltime} gives the claimed total expected runtime $\widetilde O(nM^{1/2}\polylog U)$.
 By Markov's inequality and a union bound, with at least $0.9$ probability the algorithm terminates in time $\widetilde O(nM^{1/2}\polylog U)$ with the correct answer. To boost the probability from $0.9$ to $1 - \frac{1}{\poly(n)}$, one can independently repeat the algorithm for $O(\log n)$ times and take the majority of the outputs.
\end{proofof}

\subsection{Sum-approximation convolution algorithm (proof of  \texorpdfstring{\Cref{lem:sum-approx-conv}}{})} \label{sec:sum-app-conv}
Recall that the input functions are $f,g\colon \{0,1,\dots,n\} \to \{0\}\cup [1,2^M]$ (assume $f(x)=g(x)=0$ for all $x>n$) given as length-$(n+1)$ arrays, and our goal is to compute a $(1\pm O(\delta))$-sum-approximation of $f\star g$, as defined in \cref{sec:sumapproxconv} (the approximation factor can be improved to $(1\pm \delta)$ after scaling $\delta$ by a constant factor). 
Since we allow $(1\pm O(\delta))$-approximation, we can assume each input number is represented as a floating-point number $m\cdot 2^e$, where $m\in \{0\} \cup  [1,2)$ has $\log (1/\delta)+O(1)$ bits of  precision, and $e\in \{0,1,\dots,M\}$.

   Let $D$ be the smallest power of two such that 
   \begin{equation}
    D \ge (2n+1)^2/\delta.
    \label{eqn:defnd}
   \end{equation}

 Based on the input function $f$, define an integer array $A^*[0\dots n]$ as follows: for every $0\le x\le n$, 
 \begin{itemize}
     \item If $f(x)=0$, let $A^*[x]=-\infty$.   
     \item Otherwise, $1\le f(x)\le 2^M$.  Let $A^*[x]$ be the unique integer such that $f(x) \in [D^{A^*[x]}, D^{A^*[x]+1})$. Note that $0\le A^*[x] \le M/\log_2 D$.
 \end{itemize}
  Similarly, define an integer array $B^*[0\dots n]$ based on the input function $g$.

   Decompose $f$ into three nonnegative functions $f_0,f_1,f_2$ as follows: for every $0\le x\le n$, 
   \begin{itemize}
    \item If $f(x)=0$, let $f_i(x)=0$ for all $i\in \{0,1,2\}$.
    \item Otherwise, $A^*[x]\neq -\infty$. Let $f_{A^*[x] \bmod 3}(x)=f(x)$, and let $f_{i}(x)=0$ for $i\neq A^*[x] \bmod 3$. 
   \end{itemize}
   In this way, we have $f(x)=f_0(x)+f_1(x)+f_2(x)$ for all $x\in \{0,1,\dots,n\}$.
   Similarly, decompose $g$ into the sum of three nonnegative functions $g_0,g_1,g_2$. 
   
   Note that $f\star g = \sum_{i=0}^2\sum_{j=0}^2 f_i\star g_j$.   
   We separately compute a $(1\pm \delta)$-sum-approximation of every $f_i\star g_j$. Then, adding these nine results together gives a $(1\pm \delta)$-sum-approximation of $f\star g$. Finally we round the output numbers again to $\log(1/\delta) +O(1) $ bits of  precision, which incurs another $(1\pm \delta)$ approximation factor.
   
In the following, we focus on computing a  $(1\pm \delta)$-sum-approximation of $f_0\star g_0$. The remaining $f_i \star g_j$ can be approximated in the same way.

Let 
\[A^*_0[x] :=\begin{cases}
    A^*[x] & f_0(x)\neq 0,\\
    -\infty & f_0(x)= 0.
\end{cases} \]
Note that all non-infinity entries in $A^*_0[0\dots n]$ are multiples of 3. 
Also,
$f_0(x) \in [D^{A_0^*[x]},D^{A_0^*[x]+1})$
whenever $f_0(x)\neq 0$.
Define the prefix maximum 
   \[A[x]:=\max_{0\le x'\le x} A_0^*[x'],\] and define a weight array $u[0\dots n]$ by  \[u[x]:= \begin{cases}
    f_0(x)  D^{-A[x]} & A^*_0[x] = A[x] \neq -\infty,\\
    0 & \text{otherwise}.
   \end{cases}\]
   Observe that $u[x] \in \{0\} \cup [1,D)$. Moreover, since $f_0(x)$ is an input floating-point number with $\log (1/\delta)+O(1)$ bits of precision, and $D^{-A[x]}$ is a power of two, we know $u[x]$ can be represented exactly as an integer of $\log (1/\delta)+O(1)$ bits multiplied by a power of two.

We analogously define array $B_0^*[0\dots n]$, its prefix maximum array $B[0\dots n]$, and a weight array $v[0\dots n]$, based on the function $g_0$. Again, all the non-infinity entries in $B^*_0[0\dots n]$ and $B[0\dots n]$ are multiples of 3, and $v[y] \in \{0\}\cup [1,D)$ for all $0\le y\le n$.

Since $A[0\dots n]$ and $B[0\dots n]$ are monotone non-decreasing arrays with integer entries from $\{-\infty\} \cup [0, M/\log_2 D]$, we can use \cref{thm:weightedwitnesscnt} to compute their $(\max,+)$-convolution $C[z] := \max_{x+y=z}\{A[x]+B[y]\}$\footnote{When invoking \cref{thm:weightedwitnesscnt}, we can replace the $-\infty$ input entries by $-10\lceil M/\log_2 D\rceil$ without affecting the answer.}, and also obtain the weighted witness counts \[w[z]:= \sum_{x: A[x]+B[z-x]=C[z]}u[x]\cdot v[z-x].\]
Since $u[x],v[y] \in \{0\}\cup [1,D)$, we have
\begin{equation}
\label{eqn:wz}    
0\le w[z]< (n+1)D^2
\end{equation}
for all $0\le z\le 2n$. Before invoking \cref{thm:weightedwitnesscnt}, since $u[x],v[y]$ only have $\log(1/\delta)+O(1)$ bits of precision, we can scale all of them by the same power of two so that they become $(\log D+\log (1/\delta)+O(1))$-bit integers.  The runtime of \cref{thm:weightedwitnesscnt} is $\widetilde O(n \sqrt{M/\log_2 D}\cdot \polylog(D/\delta)) = O(n\sqrt{M}\polylog(n/\delta))$.

After computing $C[0\dots 2n]$ and $w[0\dots 2n]$,  we define  $h\colon \{0,1,\dots,2n\} \to \{0\}\cup \R_{\ge 1}$ by \[h(z):= D^{C[z]} \cdot w[z].\]
  Our goal is to show that $h$ can be returned as the desired $(1\pm \delta)$-sum-approximation of $f_0\star g_0$. In other words,
  \begin{equation}
      \label{eqn:approx}
   (1-\delta)\cdot (f_0\star g_0)^\le (z)\,\le \,  h^\le (z) \,\le\, (1+\delta) \cdot (f_0\star g_0)^\le (z)
  \end{equation}
  for all $z\in \{0,\dots,2n\}$.

  Now we proceed to the proof of \eqref{eqn:approx}. 
  Since $A$ and $B$ are prefix maximum arrays of $A_0^*$ and $B_0^*$ respectively, the $(\max,+)$-convolution of $A$ and $B$ can be alternatively expressed as
  \[  C[z] =\max_{x+y\le z}\{A_0^*[x]+B_0^*[y]\}.\]
We assume $C[z]\neq -\infty$; otherwise, it is easy to see that $h^\le(z)=(f_0\star g_0)^\le (z)=0$, so \eqref{eqn:approx} already holds.
 Consider the set of index pairs attaining maximum in the above expression,
 \begin{equation}
     \label{eqn:ik}
I_z := \{(x,y): x+y\le z,  A_0^*[x]+B_0^*[y] = C[z]\}.
 \end{equation}
We will show the following two bounds which relate $h^{\le}(z)$ with $(f_0\star g_0)^\le (z)$ through the intermediate quantity $\sum_{(x,y)\in I_z}f_0(x)g_0(y)$:
\begin{align}
   (1-\delta)\cdot (f_0\star g_0)^{\le }(z) \,\le\, \sum_{ (x,y)\in I_z} f_0(x)g_0(y) \,&\le \, (f_0\star g_0)^{\le }(z).\label{eqn:approx1}\\
0 \,\le\, h^{\le}(z) -  \sum_{(x,y)\in I_z}f_0(x)g_0(y)\,&\le\,  \delta \cdot (f_0\star g_0)^\le(z), \label{eqn:approx2}
\end{align}
Adding \eqref{eqn:approx1} and \eqref{eqn:approx2} immediately gives the desired bound in \eqref{eqn:approx}.

\begin{proofof}{Proof of \eqref{eqn:approx1}}
By definition we have 
\[ (f_0\star g_0)^{\le}(z) = \sum_{ (x,y)\in I_z} f_0(x)g_0(y) + \sum_{x+y\le z:  (x,y)\notin I_z} f_0(x)g_0(y),\]
which immediately implies the second half of \eqref{eqn:approx1}. In particular, we can pick an arbitrary $(x',y') \in I_z$ and get
 \begin{equation}
     \label{eqn:dck}
   (f_0\star g_0)^{\le}(z) \ge f_0(x') g_0(y') \ge D^{A_0^*[x']}    D^{B_0^*[y']} = D^{C[z]}.
 \end{equation}

 To prove the first half of \eqref{eqn:approx1}, it suffices to show $\sum_{x+y\le z:  (x,y)\notin I_z} f_0(x)g_0(y) \le \delta\cdot (f_0\star g_0)^\le (z)$.  We use the fact that all (non-infinity) entries in arrays $A_0^*,B_0^*, C$ are multiples of 3: for any $x,y$ where $x+y\le z$, $(x,y)\notin I_z$ (and $f_0(x)g_0(y)\neq 0$), we have $A^*_0[x]+B^*_0[y]<C[z]$, which means $A^*_0[x]+B^*_0[y]\le C[z]-3$. 
  Hence,
  \begin{align*}
\sum_{x+y\le z:  (x,y)\notin I_z} f_0(x)g_0(y) & \le \sum_{x+y\le z:  (x,y)\notin I_z} D^{(A^*_0[x]+1) + (B^*_0[y]+1)}\\
& \le  (z+1)^2 \cdot  D^{(C[z]-3)+1+1}\\
 & \le \delta \cdot  (f_0\star g_0)^{\le }(z). \tag{by \eqref{eqn:defnd} and \eqref{eqn:dck}} 
  \end{align*}
  This establishes \eqref{eqn:approx1}.
\end{proofof}
 
\begin{proofof}{Proof of \eqref{eqn:approx2}}
  Since $C[0]\le C[1]\le \dots \le C[z]$, we can write 
  $h^\le(z)  = \sum_{z'=0}^z D^{C[z']}\cdot w[z']$ as
  \begin{equation}
      \label{eqn:hk}
  h^\le(z)  = \sum_{z'\le z: C[z']<C[z]} D^{C[z']}\cdot w[z'] +  \sum_{ z'\le z: C[z']=C[z]}  D^{C[z]}\cdot w[z'].
  \end{equation}

 To bound the first summand in \eqref{eqn:hk}, we use the fact that $C$ only contains multiples of $3$, so $C[z']<C[z]$ means $C[z']\le C[z]-3$:
  \begin{align*}
0\le  \sum_{ z'\le z: C[z']<C[z]} D^{C[z']}\cdot w[z']  & \le \sum_{0\le z'\le z: C[z']<C[z]} D^{C[z]-3}\cdot w[z'] \\
  &\le (z+1)\cdot D^{C[z]-3}\cdot (n+1)D^2 \tag{by \eqref{eqn:wz}}\\
  & \le \delta \cdot (f_0\star g_0)^\le(z). \tag{by \eqref{eqn:defnd} and \eqref{eqn:dck}}
  \end{align*}

  The second summand in \eqref{eqn:hk} can be written as
  \begin{align*}
 \sum_{z'\le z: C[z']=C[z]}  D^{C[z]}\cdot w[z'] &=    \sum_{z'\le z: C[z']=C[z]}  D^{C[z]}\cdot \sum_{x:A[x]+B[z'-x]= C[z']} u[x] v[z'-x]\\
&=    \sum_{x,y: x+y\le z,A[x]+B[y]=C[z]}  D^{C[z]}u[x] v[y]\\
   & =   \sum_{\substack{x,y:x+y\le z,\\ A[x]+B[y]= C[z],\\ A[x]=A^*_0[x], \\B[y]=B^*_0[y]}}  D^{C[z]}f_0(x)D^{-A[x]}\cdot  g_0(y)D^{-B[y]}\\
   & = \sum_{(x,y)\in I_z}f_0(x)g_0(y),
  \end{align*}
  where the last equality is because the set of $(x,y)$ being summed over is exactly $I_z$, which follows from the definition of $I_z$ and the following simple claim:
\begin{claim}
    \label{claim:aiaibjbj}
    For all $(x,y)\in I_z$, it holds that $A[x]=A^*_0[x]$ and $B[y]=B^*_0[y]$.
\end{claim}
\begin{proof}
 For $(x,y)\in I_z$, we have $x+y\le z$ and \[A_0^*[x]+B_0^*[y] = C[z] = \max_{x'+y'\le z} \{A_0^*[x']+B_0^*[y']\} \ge \max_{x'\le x,y'\le y} \{A_0^*[x']+B_0^*[y']\} = A[x]+B[y].\] On the other hand, $A_0^*[x]\le A[x]$ and $B_0^*[y]\le B[y]$, so all the inequalities here are equalities. 
\end{proof}
We have obtained bounds on both summands in \eqref{eqn:hk}. Together they imply the desired inequality in \eqref{eqn:approx2}.
\end{proofof}

\bibliographystyle{alphaurl}
\bibliography{main}

\end{document}